\documentclass[10pt]{article}

\usepackage{amsmath,amsfonts,amssymb,amsthm}
\usepackage[margin=0.8in]{geometry}
\usepackage{graphicx,xcolor}
\usepackage{float}
\usepackage{physics}
\usepackage[colorlinks]{hyperref}
\usepackage[capitalize]{cleveref}
\usepackage{algorithm}
\usepackage{algorithmic}
\usepackage{tikz}
\usepackage{braket}
\usepackage{cite}

\usepackage{changes}
\usepackage{appendix}

\usepackage{mathrsfs}
\usepackage{mathtools}
\mathtoolsset{centercolon}
\usepackage{bm}
\usepackage{multirow}
\usepackage{bbm}
\usepackage{subcaption}

\usetikzlibrary{arrows, shapes.gates.logic.US, calc}
\usetikzlibrary{shapes}
\usetikzlibrary{plotmarks}
\usetikzlibrary{quantikz2}
\usetikzlibrary{patterns}
\usepackage{tabularx}
\usepackage{subcaption}
\usepackage[final]{pdfpages}

\newcommand{\starnorm}[1]{\norm{#1}_{\star}}

\newcommand{\R}{{\mathbb{R}}}

\def \Re{\mathop{\rm Re}}

\def\R{{\mathbb R}}

\def\le{\leqslant}
\def\ge{\geqslant}

\DeclareMathOperator{\polylog}{\mathrm{polylog}}

\renewcommand{\Re}{\mathop{\mathrm{Re}}}

\DeclareMathOperator{\diag}{diag}

\newtheorem{theorem}{Theorem}

\newtheorem{lemma}[theorem]{Lemma}

\newtheorem{result}{Result}

\theoremstyle{definition}
\newtheorem{definition}{Definition}

\newtheorem{remark}{Remark}

\allowdisplaybreaks

\renewcommand{\sec}[1]{\hyperref[sec:#1]{Section~\ref*{sec:#1}}}
\newcommand{\app}[1]{\hyperref[app:#1]{Appendix~\ref*{app:#1}}}
\newcommand{\thm}[1]{\hyperref[thm:#1]{Theorem~\ref*{thm:#1}}}
\newcommand{\lem}[1]{\hyperref[lem:#1]{Lemma~\ref*{lem:#1}}}
\newcommand{\cor}[1]{\hyperref[cor:#1]{Corollary~\ref*{cor:#1}}}
\newcommand{\prb}[1]{\hyperref[prb:#1]{Problem~\ref*{prb:#1}}}
\newcommand{\fgr}[1]{\hyperref[fgr:#1]{Figure~\ref*{fgr:#1}}}
\newcommand{\tab}[1]{\hyperref[tab:#1]{Table~\ref*{tab:#1}}}

\newcommand{\beq}{\begin{equation}}
\newcommand{\eeq}{\end{equation}}
\newcommand{\beqa}{\begin{eqnarray}}
\newcommand{\eeqa}{\end{eqnarray}}
\DeclareMathOperator*{\cond}{cond}

\title{
\vspace{-5mm}
Structure-Preserving Quantum Method of Lines for \\Evolutionary PDEs with Mixed Boundary Conditions
}

\author{Yixuan Liang$^{1,2}$, Jin-Peng Liu$^{1,3,4,\thanks{Corresponding author: liujinpeng@tsinghua.edu.cn}}$ \\ 
\footnotesize $^{1}$ Yau Mathematical Sciences Center, Tsinghua University\\
\footnotesize $^{2}$ Qiuzhen College, Tsinghua University\\
\footnotesize $^{3}$ Institute for Applied Mathematics, Tsinghua University\\
\footnotesize $^{4}$ Beijing Institute of Mathematical Sciences and Applications\\
}

\date{}

\begin{document}

\maketitle

\vspace{-5mm}

\begin{abstract}
We give detailed analysis and circuit design of structure-preserving quantum algorithms for second-order linear evolutionary PDEs, including parabolic equations and hyperbolic equations with mixed Dirichlet, Neumann, and periodic boundary conditions and source terms. While prior quantum algorithms usually neglect the stability problem from the PDE-to-ODE reduction, our method-of-lines approach investigates the boundary lifting via Coons interpolation and boundary-aware discretization, so that the resulting semi-discrete systems are stable and compatible with efficient quantum ODE primitives. For the parabolic problem, we use a diagonal similarity transform to ensure the semi-discrete generator must have a positive semi-definite Hermitian part, and then solve the resulting ODE system by the optimal linear combination of Hamiltonian simulation (LCHS). For the hyperbolic problem, we rewrite the semi-discrete equation as an equivalent first-order system and solve it by Hamiltonian simulation. We implement our quantum algorithms with explicit block-encoding constructions and circuit implementations, as well as demonstrating the end-to-end complexity bounds together with spatial and quadrature error estimates. We conduct classical numerical experiments on the convection-diffusion equation, inhomogeneous heat equation, and Klein-Gordon equation to validate our structure-preserving analysis and algorithmic constructions.
\end{abstract}

\tableofcontents

\newpage

\section{Introduction}

Partial differential equations (PDEs) provide a fundamental mathematical language for diffusion, transport, wave propagation, electromagnetism, and many other time-dependent processes in science and engineering. In large-scale and high-dimensional settings, however, accurate numerical simulation of PDEs remains computationally demanding. 
Quantum algorithms offer a different computational paradigm: one aims to prepare a quantum state whose amplitudes encode a normalized approximation to the solution at the final time. This perspective has led to a substantial literature on quantum algorithms for differential equations and boundary value problems.

In this paper, we consider two classes of second-order linear evolutionary PDEs on a $d$-dimensional rectangular domain $\Omega=[0,a_1]\times\cdots\times[0,a_d]$. The first is a linear parabolic problem of the form
\begin{align}\label{intro 1}
 u_t = \Delta u + \sum_{l=1}^d c_l u_{x_l}  + f(x,t),
\end{align}
and the second is a linear hyperbolic problem of the form
\begin{align}\label{intro 2}
u_{tt}  = \Delta u + \sum_{l=1}^d c_l u_{x_l} - c_0^2 u + f(x,t),
\end{align}
subject to mixed Dirichlet, Neumann, and periodic boundary conditions in different coordinate directions. For the hyperbolic problem, we additionally assume that the coefficients of the first-order derivative terms vanish in the periodic directions, in order to ensure stability of the analytical solution. Our computational goal is: given an error tolerance $\varepsilon\in(0,1)$, prepare a quantum state $|\widetilde{v}\rangle$ such that
\[
\bigl\|\,|[u(x,T)]_x\rangle-|\widetilde{v}\rangle\bigr\|_2\le \varepsilon,
\]
where $[u(x,T)]_x$ denotes the solution sampled on a spatial grid at the final time $T$.

A standard classical and quantum route for solving PDEs is to discretize them directly into linear systems and then apply linear-system solvers. In the quantum setting, this leads to algorithms based on HHL-type \cite{harrow2009quantum} or more general quantum linear systems algorithms (QLSAs), which have been successful for elliptic equations such as the Poisson equation and for finite-element or spectral discretizations of boundary value problems \cite{cao2013poisson,montanaro2016fem,childs2021pde}. However, for time-dependent PDEs with initial data and source terms, a direct linear-system formulation typically incurs repeated queries for the preparation of the initial condition and the inhomogeneous term or complicated design of advanced quantum subroutines.

An alternative is the \emph{method of lines}: discretize the spatial variables, thereby reducing the PDE to a large system of ODEs, and then apply an efficient quantum ODE solver. This route is especially attractive because if the semi-discrete ODE system has the right operator structure, one can leverage state-of-the-art quantum primitives. In particular, when the coefficient matrix is anti-Hermitian, the problem reduces to Hamiltonian simulation, for which efficient algorithms based on quantum signal processing and quantum singular value transformation (QSVT) are available \cite{low2017qsp,qsvt}. More generally, if the coefficient matrix has positive semi-definite real part, one can apply the linear-combination-of-Hamiltonian-simulation (LCHS) framework \cite{lchs1,lchs2,lchs3} or Schrodingersation framework \cite{jin2024prl,jin2023pra,hu2024circuits,jin2024heatbc,jin2023maxwell,ma2024maxwellcircuit,jin2024inhom}.

However, passing from PDE to ODEs are not a purely formal step. The main technical challenge is not merely discretization, but \emph{structure-preserving discretization}: the semi-discrete operator must fall into the admissible class of the downstream quantum primitive. For the parabolic problem, to obtain an efficient quantum solver, the resulting ODE system must have a coefficient matrix with \emph{positive semi-definite} real part in order to be compatible with LCHS. For the hyperbolic problem, the semi-discrete dynamics must be rewritten into an anti-Hermitian form suitable for Hamiltonian simulation. This requirement becomes particularly delicate under mixed boundary conditions, especially for Neumann boundaries, where standard finite-difference choices may not preserve the structural properties required by LCHS or Hamiltonian simulation.

\subsection{Main results}

To benefit from efficient quantum algorithms for ODEs, the semi-discrete system must be
\emph{structurally compatible} with the downstream quantum primitive.
Accordingly, the main contribution of this work is to investigate \emph{boundary-aware} finite-difference schemes and transformations that preserve precisely the operator structures required by LCHS (its Hermitian part of the semi-discrete generator is positive
semidefinite) or Hamiltonian simulation (its semi-discrete generator is anti-Hermitian).

We now summarize the two algorithmic pipelines and the resulting complexity bounds.

\paragraph{Parabolic PDEs.}
For the parabolic problem \eqref{intro 1}, our solver proceeds in four steps:
\begin{enumerate}
    \item We employ the boundary lifting via Coons interpolation, investigate a boundary-aware finite-difference discretization under mixed Dirichlet, Neumann, and periodic boundary conditions, and choose the spatial mesh sizes $N_l$ according to the method-of-lines error $\epsilon_N$, where the error $\varepsilon_N$ is second-order, i.e. $\varepsilon_N=\mathcal{O}(1/N_l^2).$  The discretization produces a semi-discrete ODE system with the sign and symmetry structure needed for the subsequent LCHS-compatible transformation.
    \item If there are first-order terms $c_lu_{x_l}$, we further introduce a diagonal similarity transform to complete the structure-preserving construction, so that the resulting semi-discrete coefficient matrix has positive-semidefinite real part and therefore becomes compatible with the LCHS framework.
    \item We explicitly construct block-encodings of the transformed coefficient matrix using basic gates.
    \item We then discretize the time integral in the inhomogeneous term by composite Gauss--Legendre quadrature, and finally apply the optimal LCHS method to the resulting ODEs.
\end{enumerate}

\begin{result}[Informal version of Theorem~\ref{para final thm}]\label{result1}
Assume that the solution $u$ is sufficiently smooth, and define
\[
\Gamma
:=
\max_{1\le l\le d}
\Bigl(
\max\{\|\partial_{x_l}^3 u\|_\infty,\|\partial_{x_l}^4 u\|_\infty\}
\Bigr),
\qquad
q
:=
\frac{
\|[u_0(x)]_x\|_\star+\int_0^T \|[f(x,s)]_x\|_\star\,ds
}{
\|[u(x,T)]_x\|_\star
}.
\]
Then, for any $\epsilon\in(0,1)$, one can choose the spatial discretization parameters,
the quadrature parameters, and the internal LCHS precisions so as to prepare a state
$|\widetilde{v}(T)\rangle$ satisfying
\(
\bigl\|\,|[u(x,T)]_x\rangle-|\widetilde{v}(T)\rangle\bigr\|_2\le \epsilon.
\)
The overall costs are
\[
\begin{cases}
\widetilde{\mathcal{O}}\!\left(
q\,
\dfrac{T^{2} d^{2}\Gamma}
{\|[u(x,T)]_x\|_\star\,\epsilon}
\right)
& \text{basic gates},\\[1.0em]
\widetilde{\mathcal{O}}\!\left(
q\,
\dfrac{T^{2} d\,\Gamma}
{\|[u(x,T)]_x\|_\star\,\epsilon}
\right)
& \text{queries to a relevant quadrature oracle },\\[1.0em]
\mathcal{O}(q)
& \text{queries to the state-preparation oracles } ,\\[1.0em]
\widetilde{\mathcal{O}}\!\left(
d\log\frac{T\Gamma}{\|[u(x,T)]_x\|_\star\,\epsilon}
\right)
& \text{qubits}.
\end{cases}
\]
\end{result}

Result~\ref{result1} presents a complete complexity bound within the parabolic PDE solver framework. Its main significance is that the spatial discretization, the quadrature of the inhomogeneous term, and the internal LCHS precision are balanced in a single end-to-end theorem guaranteeing the final accuracy~$\epsilon$. 
The resulting quadratic dependence on $T,d$ and first-order dependence on $\varepsilon$ are not claimed to be optimal; rather, they mainly reflect the present method-of-lines discretization, the tensor-product finite-difference construction, and the current quadrature strategy. 

\medskip
\noindent\textbf{Stability analysis.} For simulating the ODE system $v'(t)=-Av(t)+b(t)$, the LCHS method needs a stronger structural requirement than the usual spectral stability condition $\Re \lambda(A)\ge 0$. In the LCHS framework, one needs the Hermitian part of the semi-discrete generator itself to be positive semidefinite $L=\frac{1}{2}(A+A^\dag) \succeq 0$, rather than merely requiring all eigenvalues to have nonnegative real parts. If this condition fails, one is typically forced to shift the spectrum to recover positivity, which in turn introduces an exponential factor of the form $\exp(\mu T)$, where $\mu$ is the size of the required shift. 
In many PDE discretizations, such as ghost-point discretization, the factor is $\mu=\mathcal{O}(N^2)$, where $N=\mathcal{O}(1/\sqrt{\varepsilon})$ is the number of nodes. At this point, the exponential factor $\exp(\mu T)$ is too large and unacceptable in practice even for short-time simulation.
Our structure-preserving discretization precisely produces semi-discrete operators which satisfy this stronger positivity condition by midpoint nodes in Section \ref{section 3.2} and similarity transform in Section \ref{section 3.3}, so that we do not need to shift the spectrum.

\paragraph{Hyperbolic PDEs.}
For the hyperbolic problem \eqref{intro 2}, the overall route is analogous, but the quantum primitive is different:
\begin{enumerate}
    \item We employ the similar boundary lifting, investigate a boundary-aware spatial discretization under mixed Dirichlet, Neumann, and periodic boundary conditions, and choose $N_l$ according to the spatial discretization error $\epsilon_N$, where the error $\varepsilon_N$ is also second-order. This discretization is chosen so that, together with the subsequent similarity transform, the resulting semi-discrete second-order ODE system has the structure needed for an efficient reduction to Hamiltonian simulation.
    \item We then apply the corresponding similarity transform and rewrite the semi-discrete second-order ODE system as an equivalent first-order system by introducing auxiliary variables. The resulting extended system can be represented in a form compatible with Hamiltonian simulation.
    \item We explicitly construct block-encodings of the matrices in the extended system using basic gates.
    \item We discretize the inhomogeneous time integral by composite Gauss--Legendre quadrature, and then implement the solution operator through block-encoding and Hamiltonian simulation.
\end{enumerate}

\begin{result}[Informal version of Theorem~\ref{hy final thm}]\label{result2}
Assume that the solution $u$ is sufficiently smooth, and define
\[
\Gamma
:=
\max_{1\le l\le d}
\Bigl(
\max\{\|\partial_{x_l}^3 u\|_\infty,\|\partial_{x_l}^4 u\|_\infty\}
\Bigr),
\qquad
q
:=
\frac{
\|[u_0(x)]_x\|_\star
+
T\|[\phi(x)]_x\|_\star
+
\int_0^T\int_0^s \|[f(x,\tau)]_x\|_\star\,d\tau\,ds
}{
\|[u(x,T)]_x\|_\star
}.
\]
Then, for any $\epsilon\in(0,1)$, one can choose the discretization and quadrature parameters so as to prepare a state
$|\widetilde{v}(T)\rangle$ satisfying
\(
\bigl\|\,|[u(x,T)]_x\rangle-|\widetilde{v}(T)\rangle\bigr\|_2\le \epsilon.
\)
The overall costs are
\[
\begin{cases}
\widetilde{\mathcal{O}}\!\left(
q\,
T^{2} d^{3/2}
\left(
\dfrac{\Gamma}{\|[u(x,T)]_x\|_\star\,\epsilon}
\right)^{1/2}
\right)
& \text{basic gates},\\[1.0em]
\widetilde{\mathcal{O}}\!\left(
q\,
T^{2} d^{1/2}
\left(
\dfrac{\Gamma}{\|[u(x,T)]_x\|_\star\,\epsilon}
\right)^{1/2}
\right)
& \text{queries to a relevant quadrature oracle },\\[1.0em]
\mathcal{O}(q)
& \text{queries to the state-preparation oracles } ,\\[1.0em]
\widetilde{\mathcal{O}}\!\left(
d\log\dfrac{T\Gamma}{\|[u(x,T)]_x\|_\star\,\epsilon}
\right)
& \text{qubits}.
\end{cases}
\]
\end{result}

Similarly, Result~\ref{result2} presents a complete complexity bound for the present hyperbolic PDE solver framework, rather than as an optimal one. The dependence on $T^2$ and $d^{3/2}$ comes from the current semi-discretization, the first-order reformulation of the second-order dynamics, and the corresponding Hamiltonian-simulation-based implementation. Nevertheless, the theorem shows that after converting the hyperbolic problem into a Hamiltonian-simulation-compatible form, one obtains explicit end-to-end resource bounds with polynomial dependence on the evolution time, dimension, and target accuracy. Compared with the parabolic case, the dependence on $d$ and $\epsilon$ is milder here, reflecting the different downstream quantum primitive used in the hyperbolic solver.

\paragraph{Numerical experiments.}
Our classical numerical experiments are designed to validate the main mathematical and algorithmic ingredients of the proposed framework. By the case study of the convection-diffusion equation, inhomogeneous heat equation, and Klein-Gordon equation, we confirm that the PDE-to-ODE reduction, the structure-preserving discretization, and the subsequent quantum-algorithmic constructions work together as predicted. Specifically, we verified the following items:
\begin{itemize}
    \item We tested whether the boundary-aware discretizations and transformations introduced in this paper indeed produce semi-discrete operators compatible with LCHS in the parabolic case and with Hamiltonian simulation in the hyperbolic case (Figure \ref{para1 1},  Figure \ref{para1 2}, Figure \ref{para2 1} and Figure \ref{hyper1 1}).
    \item We illustrated the second-order accuracy of the resulting approximations for representative PDEs with mixed boundary conditions and source terms (Figure \ref{para4 1} and Figure \ref{hyper1 1}).
    \item Furthermore, we compared the practical performance of different kernel functions in the development of the LCHS method (Figure \ref{para3 1} and Figure \ref{para3 2}).
\end{itemize}

\paragraph{Main contributions.}
The two results above should be viewed as end-to-end PDE solver statements rather than abstract ODE consequences.
More specifically, the present paper contributes the following:
\begin{itemize}
    \item We develop a structure-preserving, boundary-aware method-of-lines framework for two classes of second-order linear evolutionary PDEs, including mixed Dirichlet, Neumann, and periodic boundary conditions.
    \item For the parabolic problem, we combine a carefully chosen finite-difference discretization and a diagonal similarity transform so that the semi-discrete system satisfies the structural assumptions required by LCHS.
    \item For the hyperbolic problem, we convert the semi-discrete second-order equation into an equivalent first-order system whose solution can be implemented by Hamiltonian simulation.
    \item In both cases, we treat inhomogeneous terms explicitly, derive spatial and quadrature error bounds, construct the required block-encodings and oracles, and obtain explicit gate, query, and qubit complexity bounds for preparing the final-time solution state.
\end{itemize}

\subsection{Related works}

\paragraph{Quantum linear system algorithms for linear ODEs and PDEs.}

One major line of work on quantum differential-equation algorithms is based on first discretizing the equation into a linear system and then applying a quantum linear systems algorithm. 
At the ODE level, quantum algorithms include multistep methods, Taylor-series-based methods, spectral methods, and linear-system-based Dyson-series methods~\cite{berry2014ode,berry2017ode,childs2020spectral,berry2022dysonode}. Their common feature is that they reformulate the problem globally as a structured linear system and then solve it by quantum linear-algebraic primitives. These works provide an important benchmark for quantum differential-equation algorithms. 

At the PDE level, several important quantum algorithms also follow a linear-system-based philosophy. Early examples include quantum algorithms for the Poisson equation on finite-difference grids \cite{cao2013poisson}. A more refined analysis was later developed for finite-element methods by Montanaro and Pallister, who showed that quantum linear-system techniques can lead to polynomial speedups \cite{montanaro2016fem}. Later, high-precision quantum algorithms for PDEs were further developed, which combine adaptive-order finite differences and spectral methods with high-precision quantum linear-system solvers; in particular, researchers obtained a finite-difference algorithm for the Poisson equation and a spectral algorithm for more general second-order elliptic equations, with polylogarithmic dependence on the error tolerance \cite{childs2021pde,kharazi2025explicit}.
At the same time, they are structurally different from the present paper. Their starting point is a linear-system formulation of the discretized problem, whereas our approach proceeds through the method of lines and then invokes quantum ODE primitives.

\paragraph{Evolution-based quantum algorithms for linear ODEs and PDEs.}
Another major route to quantum algorithms for time-dependent PDEs is to discretize only the spatial variables and then solve the resulting large ODE system by a quantum evolution algorithm. This method-of-lines perspective seeks to implement the underlying time evolution more directly, and is therefore particularly natural for evolutionary PDEs. At the ODE level, a representative example is the quantum time-marching method~\cite{fang2023timemarching}, which advances the solution over short time intervals without reformulating the whole problem as a single linear system. From the complexity-theoretic side, An, Liu, Wang, and Zhao studied the limitations and fast-forwarding of quantum differential equation solvers, showing that the efficiency of quantum evolution algorithms depends strongly on structural properties of the generator and that certain classes admit substantial fast-forwarding improvements~\cite{an2025quantum}. 

Beyond these generic evolution-based ODE solvers, there is an increasingly important class of approaches that exploit additional structure in the semi-discrete operator. Depending on the form of the discretized dynamics, one may reduce the problem directly to Hamiltonian simulation, apply LCHS to stable non-unitary systems, or use transformation frameworks such as Schr\"odingerisation. Since our own algorithms belong to this structure-exploiting method-of-lines paradigm, we review these three directions separately below.

\paragraph{LCHS for linear ODEs and PDEs.}
For general non-unitary linear dynamics, An, Liu, and Lin introduced the LCHS framework, which expresses a stable non-unitary propagator as a linear combination of unitary Hamiltonian evolutions \cite{lchs1}. 
This framework was subsequently sharpened by An, Childs, and Lin, who introduced improved kernel functions and obtained near-optimal dependence on all relevant parameters \cite{lchs2}. Very recently, Low and Somma generalized this line of work and achieved optimal query complexity, while also showing exponential convergence of a uniform trapezoidal discretization \cite{lchs3}. More recently, randomized implementations of LCHS have been proposed to reduce circuit overheads and improve practical efficiency for simulating linear non-unitary dynamics, providing further algorithmic support for structure-based PDE solvers built on LCHS~\cite{yang2025circuit}.

Compared with Schrodingerisation in the next paragraph, the literature on applying LCHS directly to PDEs is still limited. A recent representative work by Sato, Tezuka, Kondo, and Yamamoto proposed an LCHS-based quantum algorithm for second-order linear PDEs of nonconservative systems with spatially varying parameters, using finite-difference discretization and explicit qubit-operator constructions \cite{sato2025lchspde}. More recently, Lu, Li, Liu, and Liu extended LCHS to an infinite-dimensional setting and discussed applications to linear parabolic PDEs and related infinite-dimensional dynamics \cite{lu2025inflchs}. And an end-to-end LCHS-based quantum algorithm was developed for rapidly distorted turbulence \cite{meng2026toward}.

\paragraph{Schrodingerisation for linear ODEs and PDEs.}

A second major route to quantum PDE algorithms is Schrodingerisation. Introduced by Jin, Liu, and Yu, this framework uses a warped phase transformation to map linear ODEs and PDEs with non-unitary dynamics into Schrodinger-type systems in one higher dimension \cite{jin2024prl,jin2023pra}. 
At the level of explicit circuit constructions and complexity analysis, several PDE-focused Schrodingerisation papers are now available. Hu, Jin, Liu, and Zhang presented explicit quantum circuits for solving PDEs via Schrodingerisation and illustrated the construction using the heat equation and the advection equation \cite{hu2024circuits}. Jin, Liu, and Yu further studied the heat equation with physical boundary conditions and provided detailed circuit constructions together with complexity analyses \cite{jin2024heatbc}. Schrodingerisation has also been developed for Maxwell's equations \cite{jin2023maxwell,ma2024maxwellcircuit}. In addition, Jin, Liu, and Ma analyzed Schrodingerisation-based quantum algorithms for general linear dynamical systems with inhomogeneous terms \cite{jin2024inhom}. From a broader theoretical viewpoint, Li recently proposed a moment-matching dilation framework for linear non-unitary dynamics, showing that both Schrodingerisation and LCHS can be viewed as special cases of a more general unitary-embedding principle~\cite{li2025linear}. 

These Schrodingerisation works are closely related to the present paper. Existing Schrodingerisation-based algorithmic results have developed explicit constructions for several important model problems. At the same time, the current literature also indicates that explicit quantum solvers for more general PDE settings remain highly problem-dependent. By contrast, the present work advances along a different axis: we develop an LCHS- and Hamiltonian-simulation-based method-of-lines framework for two classes of second-order linear evolutionary PDEs, incorporating mixed Dirichlet/Neumann/periodic boundary conditions, source terms, spatial discretization error analysis, and explicit end-to-end complexity bounds. In this sense, our contribution is a complete PDE-to-algorithm construction together with error analysis, in which the semi-discrete operators are deliberately designed so as to be compatible with LCHS or Hamiltonian simulation.

These Schrodingerisation works are closely related to the present paper, but they emphasize a different aspect of the PDE-to-quantum-algorithm pipeline. Existing Schrodingerisation-based results mainly develop transformation-based formulations, explicit circuit constructions, and complexity analyses for several important model problems. 
But two further differences are worth noting. First, the existing Schrodingerisation-based PDE and circuit papers formulate their main analyses after spatial semi-discretization, focusing on the Schr\"odingerised discrete dynamics and the associated circuit implementation, rather than on a single end-to-end theorem that also balances the method-of-lines error of the original PDE discretization. Second, for Neumann boundary conditions, the standard ghost-point closures adopted in existing Schrodingerisation work \cite{jin2024heatbc} do not preserve the stronger semidefinite structure needed in our setting. Combined with the later Schrodingerisation analysis of unstable modes \cite{jin2024inhom}, where positive eigenvalues of the Hermitian part enter the recovery cost exponentially through a factor of the form \(e^{p^\star}\) with \(p^\star \gtrsim \lambda_{\max}^+(H_1)T\), where $\lambda_{\max}^+ = \mathcal{O}(N^2) = \mathcal{O}(1/\epsilon)$ in the discretized PDE setting. This indicates that such ghost-point discretizations are not suitable for controlled error bounds in our LCHS framework.

By contrast, the present paper places additional emphasis on the semi-discrete operator itself. In our setting, the treatment of mixed Dirichlet/Neumann/periodic boundary conditions, the inclusion of source terms, and the control of long-time end-to-end error all depend crucially on how the spatial discretization is designed before the downstream quantum primitive is applied. In particular, for Neumann boundaries, we are led to adopt a different discretization strategy, since standard ghost-point closures do not provide the operator structure required in our LCHS- and Hamiltonian-simulation-based framework for controlled long-time and end-to-end error bounds. In this sense, our contribution is a complete PDE-to-algorithm construction, together with error analysis, in which the semi-discrete operators are deliberately designed so as to be compatible with LCHS or Hamiltonian simulation.

\paragraph{Hamiltonian simulation algorithms for linear PDEs.}

There also exist specific quantum algorithms that are based more directly on Hamiltonian simulation. A representative example is the quantum wave equation solver of Costa, Jordan, and Ostrander~\cite{waveequ}, which treats the wave equation under Dirichlet and Neumann boundary conditions by reducing the discretized dynamics to a form amenable to Hamiltonian simulation. Our treatment of the hyperbolic problem is partly inspired by this paper.
Sato \emph{et al.} developed scalable quantum circuits for Hamiltonian simulation of certain hyperbolic PDEs and applied them to advection and wave equations \cite{sato2024hyperbolic}. Sato, Tezuka, Kondo, and Yamamoto later proposed a Hamiltonian-simulation-based quantum algorithm for the advection equation using discrete time-marching operators embedded into Hamiltonian simulations \cite{sato2024advection}. 

The present paper intersects with this line in the hyperbolic case. Our distinction is that we simultaneously treat a parabolic class via LCHS and a hyperbolic class via Hamiltonian simulation within a unified method-of-lines perspective, while also incorporating mixed boundary conditions and source terms and analyzing the complexity.

\paragraph{Quantum algorithms for nonlinear ODEs and PDEs.}
Quantum algorithms for nonlinear differential equations are much less understood than their linear counterparts and are largely beyond the scope of the present paper. The main rigorous route studied so far is based on linearization in an enlarged space, most notably through Carleman linearization, which embeds a polynomial nonlinear system into a higher-dimensional linear one and then applies quantum algorithms for linear systems or linear ODEs~\cite{LiuKoldenKroviLoureiroTrivisaChilds2021,Krovi2023improvedquantum,LiuAnFangWangLowJordan2023,CostaSchleichMoralesBerry2025,wu2025quantum,jennings2025quantum,li2026efficient,wang2026quantumalgorithmsnonlineardifferential}. 
Other nonlinear routes include early non-Hermitian-Hamiltonian-style formulations and more recent homotopy-based approaches, such as quantum homotopy perturbation for dissipative nonlinear ODEs and HAM-based frameworks for nonlinear PDEs~\cite{Lloyd2020NonlinearDE,XueWuGuo2021QHPM,Liao2024HAMSchrodingerisation,XueEtAl2025QHAM}.
Related linear-representation and generalized-solution viewpoints include Koopman--von Neumann, Liouville, level-set, and Young-measure approaches~\cite{Joseph2020KvN,JinLiu2024LevelSet,LiouvilleComparison2023,jennings2026quantumkoopmanalgorithms,jin2026quantumalgorithmsyoungmeasures}.
These works provide an important parallel line of research, but they address nonlinear dynamics through substantially different lifting mechanisms and are not directly comparable to the present linear PDE framework.

\paragraph{Comparison with the present work.}

For clarity, we summarize the positioning of this paper relative to the above literature. Compared with QLSA-based PDE algorithms, our method avoids a direct linear-system formulation and instead leverages quantum ODE primitives through the method of lines. Compared with Schr\"odingerisation, our approach is more specialized in scope, but it provides a more explicit structure-preserving discretization framework for mixed boundary conditions and a direct PDE-to-algorithm route based on LCHS and Hamiltonian simulation. Compared with the existing LCHS-based PDE literature, our work provides a unified treatment of one parabolic and one hyperbolic family, both with source terms and mixed Dirichlet/Neumann/periodic boundary conditions, together with PDE-level error analysis and end-to-end complexity bounds.

We also make repeated use of standard quantum algorithmic primitives such as Hamiltonian simulation, QSVT, and the linear combination of unitaries (LCU) subroutine \cite{lcu,low2017qsp,qsvt}.

\subsection{Organization}

The rest of the paper is organized as follows. In Section~\ref{section 2}, we review the block-encoding and QSVT-based Hamiltonian simulation together with the optimal LCHS framework used later in the paper. In Section~\ref{section 3}, we develop the method-of-lines discretization, boundary treatment, similarity transform, and error analysis for the parabolic problem, and in Section~\ref{section 4} we present the corresponding quantum algorithm and complexity analysis. In Sections~\ref{section 5} and~\ref{section 6}, we carry out the analogous program for the hyperbolic problem, including the auxiliary-variable reformulation that reduces the semi-discrete dynamics to Hamiltonian simulation. In Section~\ref{section 7} we present numerical results illustrating the constructions and validating the proposed algorithms. Finally, in Section~\ref{section 8}, we discuss some natural questions and directions.

\section{Preliminaries}\label{section 2}
This section reviews the basic notions of block-encoding and Hamiltonian simulation that will be used throughout the paper. We then summarize the optimal Linear-Combination-of-Hamiltonian-Simulation (LCHS) framework for solving linear ODEs, following the formulation in~\cite{lchs3}.

\subsection{Notations}
\begin{itemize}
\item We use $\|\cdot \|_2$ to denote the $2$-norm of a matrix or a vector. And $\mathrm{cond}(A):=\|A\|_2 \| A^{-1}\|_2$ denotes the condition number in the $2$-norm. 

\item $f=\mathcal{O}(g)$ or $f\lesssim g$ means that there exists a constant $M>0$ such that $f(x)\leq Mg(x)$ for all $x$. $f=\widetilde{\mathcal{O}}(g)$ means $f=\mathcal{O}(g\cdot \polylog(g))$. 

\item For a function $f(x)=f(x_1,\cdots,x_d): \Omega\subset \mathbb{R}^d\rightarrow \mathbb{R}$, suppose the grid nodes in $l$-th dimension are $\{x^{(j)}_l\}_{j=0}^{N_l-1},\ l=1,\cdots,d$. Restricting the function $f$ to these nodes yields a vector, denoted by $[f(x)]_{x_1,\cdots,x_d}$ or $[f(x)]_{x}$.

\item In PDE solver, we need to balance the norm based on the number of grid nodes. For a vector $u\in \mathbb{R}^N$, we define its mean $l_2$ norm as $\|u\|_{\star}:=\sqrt{\frac{1}{N}}\|u\|_2$.

\item For a non-zero vector $u$, we use $|u\rangle:=u/\|u\|_2$ to denote its normalized unit vector.

\item $\lceil x\rceil$ denotes the greatest integer greater than or equal to $x$.
\end{itemize}

\subsection{Block-encoding and Hamiltonian simulation}
Block-encoding provides a standard interface that allows quantum circuits (which are unitary by nature) to access general matrices. It plays a central role in quantum algorithms for matrix functions, including Hamiltonian simulation via the Quantum Singular Value Transformation (QSVT) framework~\cite{qsvt}.
\begin{definition}[Block-encoding]\label{def blockencoding}
Let $A\in\mathbb{C}^{N\times N}$ be an $n$-qubit matrix with $N=2^n$. If there exist $\alpha\in\mathbb{R}_{+}$, $\varepsilon\in\mathbb{R}_{+}$, and an $(m+n)$-qubit unitary $U_A$ such that
\begin{equation*}
\left\|A-\alpha\Big( (\langle 0^m|\otimes I_N)\,U_A\,(|0^m\rangle\otimes I_N)\Big)\right\|_2\le \varepsilon,
\end{equation*}
then $U_A$ is called an $(\alpha,m,\varepsilon)$-block-encoding of $A$, denoted by $U_A\in(\alpha,m,\varepsilon)\mathrm{BE}(A)$.
If $\varepsilon=0$, we say the block-encoding is exact and write $U_A\in(\alpha,m)\mathrm{BE}(A)$.
\end{definition}
Intuitively, block-encoding means that when the $m$ ancilla qubits are projected onto $|0^m\rangle$, the induced transformation on the $n$-qubit system register approximates $A/\alpha$. The lemmas about construction of block-encoding are in Appendix \ref{appendix blockencoding}.

Consider the Schrodinger equation
\begin{equation*}
\frac{\mathrm{d}}{\mathrm{d}t}u(t) = -i H u(t),
\end{equation*}
where $H$ is Hermitian. Its analytic solution is $u(t)=e^{-itH}u(0)$. So to solve it numerically, we need to produce the unitary $e^{-itH}$, which is precisely the \emph{Hamiltonian simulation} problem. By the Quantum-Singular-Value-Transformation method \cite{qsvt}, given a block-encoding of $H$, one can implement $e^{-itH}$ with optimal asymptotic query complexity in all relevant parameters. 

\begin{theorem}[Hamiltonian simulation based on QSVT{\cite[Corollary~60]{qsvt}}]\label{thm qsvt hs}
    Let $\varepsilon \in (0,\frac{1}{2})$, $t \in \mathbb{R}$ and $\alpha \in \mathbb{R}_+$.
Let $U$ be an $(\alpha,a,0)$-block-encoding of the unknown Hamiltonian $H$.
In order to implement an $\varepsilon$-precise Hamiltonian simulation unitary $V$ which is an $(1,a+2,\varepsilon)$-block-encoding of $e^{itH}$, it is necessary and sufficient to use the unitary $U$
 a total number of times
\[\Theta\left( \alpha|t| + \frac{\log(1/\varepsilon)}{\log\!\big(e + \log(1/\varepsilon)/(\alpha|t|)\big)} \right
).\]
\end{theorem}

\subsection{The optimal LCHS method}
We consider the linear inhomogeneous ODE on $t\in[0,T]$,
\begin{align}\label{linear ODEs}
\frac{\mathrm{d}}{\mathrm{d}t}u(t) \;=\; -A(t)u(t) + b(t), 
\quad u(0)=u_0,
\end{align}
where $u(t)\in\mathbb{C}^N$, $A(t)\in\mathbb{C}^{N\times N}$, and $b(t)\in\mathbb{C}^N$. 
Define the time-ordered propagator from time $s$ to time $t$ as
\begin{align*}
U_s(t) \;:=\; \mathcal{T}\exp\!\left(-\int_s^t A(\tau)\,d\tau\right),
\quad 0\leq s\leq t\leq T,
\end{align*} where $\mathcal{T}$ denotes the time ordering operator. We can get its analytic solution \begin{align*}
u(t) \;=\; U_0(t)\,u_0 \;+\; \int_0^t U_s(t)\,b(s)\,ds.
\end{align*}
So we only need to consider how to block-encode the matrix $U_s(t)$ for arbitrary $s\leq t$.

The LCHS framework transforms the general matrix $U_s(t)$ into an integral of Hamiltonian simulation.
Concretely, assume that $A(\tau)$ admits a decomposition
\begin{equation*}
A(\tau)=L(\tau)+iH(\tau),
\end{equation*}
where $L(\tau)=\frac{1}{2}(A(\tau)+A^\dag(\tau))$ and $H(\tau)=\frac{1}{2i}(A(\tau)-A^\dag(\tau))$ are Hermitian. We assume $L(\tau)\succeq 0$ in the regimes where LCHS applies.
LCHS introduces a kernel function $\widehat{f}(k)$ and a truncation parameter $R>0$ such that
\begin{equation*}
U_s(t)\approx O_R(t),\quad
O_R(t):=\frac{1}{\sqrt{2\pi}}\int_{-R}^{R}\widehat{f}(k)\,U_s(t,k)\,dk,
\end{equation*}
where the integrand $U_s(t,k)$ is a time-ordered unitary evolution,
\begin{equation*}
U_s(t,k):=\mathcal{T}\exp\!\left(-i\int_s^t \big(kL(\tau)+H(\tau)\big)\,d\tau\right).
\end{equation*}
The formula reduces a non-unitary dynamics problem to a continuous linear combination of Hamiltonian simulations, enabling the use of QSVT/Dyson-series Hamiltonian simulation subroutines together with the standard LCU (linear-combination-of-unitaries) method.

To obtain a finite quantum circuit, the integral in $O_R(t)$ is discretized using a uniform trapezoidal rule with step size $h>0$. Let $h$ be the quadrature step size, then
\begin{equation*}
O_R(t)\approx \frac{h}{\sqrt{2\pi}}\sum_{j=-R/h}^{R/h}\widehat{f}(jh)\,U_s(t,jh),
\end{equation*}
which is a linear combination of unitary evolutions indexed by $j$.
In practice, each $U_s(t,jh)$ is implemented only approximately; we denote by $\widetilde{U}_s(t,jh)$ the resulting Hamiltonian simulation primitive.

Following~\cite{lchs3}, we adopt a specific kernel family $\widehat{f}_2(k;\gamma,c)$ with parameters $\gamma>0$ and $c>0$:
\begin{equation*}
\widehat{f}_2(k;\gamma,c)
:=\sqrt{\frac{2}{\pi}}\,\frac{e^{-c(ik-1)}}{1+k^2}\,
\exp\!\left(-\frac{k^2+1}{4\gamma^2}\right).
\end{equation*}
The parameters $(\gamma,R)$ and the quadrature step size $h$ are chosen to control two independent error sources:
(i) the LCHS truncation approximation error $\varepsilon_{\mathrm{lchs}}$, and
(ii) the quadrature discretization error $\varepsilon_{\mathrm{quad}}$.

A convenient framework is as follows:
\begin{enumerate}
\item \textbf{Kernel parameters and truncation range.}
For the truncation error $\varepsilon_{\mathrm{lchs}}$, fix $c>0$ and choose
\begin{align}\label{lchs gamma}
\gamma=\frac{1}{c}\sqrt{c+\log\!\left(\frac{1+1/(2\pi)}{\varepsilon_{\mathrm{lchs}}}\right)}
=\mathcal{O}\!\left(\sqrt{\log\frac{1}{\varepsilon_{\mathrm{lchs}}}}\right),
\quad
R=2c\gamma^2=\mathcal{O}\!\left(\log\frac{1}{\varepsilon_{\mathrm{lchs}}}\right).
\end{align}

\item \textbf{Quadrature step size.}
For the quadrature error $\varepsilon_{\mathrm{quad}}$, choose $h$ such that
\begin{align}\label{lchs h}
h\le
\frac{\pi}{\frac{1}{2}\|L\|_{L^1}+\log\!\left(\frac{64\,e^{3c/2}}{15\,\varepsilon_{\mathrm{quad}}}\right)},
\end{align}
where $\|L\|_{L^1}$ denotes the $L^1$-norm in time.

\item \textbf{Construction of the oracles.} By Hamiltonian simulation, we construct the oracle
\begin{align*}
    \widetilde{\text{SEL}}_R = \sum_{j=-R/h}^{R/h} |j\rangle \langle j |\otimes \widetilde{U}_s(t, jh),
\end{align*}
where $\|U_s(t, jh)-\widetilde{U}_s(t, jh)\|_2\leq \varepsilon_{\text{hs}}$.
From the LCU coefficients, we construct the oracles
\begin{align*}
\text{PREP}|0\rangle \propto \sum_{j=-R/h}^{R/h} |\hat{f}_2(hj;\gamma,c)|^{1/2} |j\rangle ,\quad   
\overline{\text{PREP}}|0\rangle \propto \sum_{j=-R/h}^{R/h} e^{i \mathrm{Arg}[\hat{f}_2(hj;\gamma,c)]} |\hat{f}_2(hj;\gamma,c)|^{1/2} |j\rangle.\end{align*}

\item \textbf{LCU block-encoding.}
The composite operator
\begin{equation*}
\mathrm{PREP}^{\dagger}\cdot \widetilde{\mathrm{SEL}}_R \cdot \overline{\mathrm{PREP}}
\end{equation*}
is a block-encoding of the target matrix $U_s(t)$ up to the total error $\varepsilon=\varepsilon_{\text{lchs}}+\varepsilon_{\text{quad}}+\varepsilon_{\text{hs}}$.
\end{enumerate}

The complexity of the optimal LCHS method is specified by the following theorem. 
\begin{theorem}[Block-encoding of optimal LCHS  {\cite[Theorem~4]{lchs3}}]\label{thm lchs}
Assume $L(\tau)=\frac{1}{2}(A(\tau)+A^\dag(\tau)) \succeq 0$ for $\tau\in [0,t]$. For any $\varepsilon\in (0,4/5]$, under the construction above, we may block-encode $U_s(t)$ by 
\[ \operatorname{PREP}^\dagger \cdot \widetilde{\operatorname{SEL}}_R \cdot \overline{\operatorname
{PREP}}\in ( \alpha, \mathcal{O}(\log R/h), \varepsilon)\text{BE}(U_s(t)),\]with normalization factor $$\alpha=\dfrac{h}{\sqrt{2\pi}} \sum_{j=-R/h}^{R/h} |\hat{f}_2(hj; \gamma, c)|=\mathcal{O}(1).$$
Moreover, $\operatorname{PREP}$ and $\overline{\operatorname{PREP}}$
 costs $\mathcal{O} \!\left(\bigl(\log(\|L\|_{L^1} + \log \tfrac{1}{\varepsilon})\bigr) \log^{5/2} \tfrac{1}{\varepsilon}\right)$
two-qubit gates. 
$\widetilde{\operatorname{SEL}}_R$ has the same query complexity $Q$
 and gate complexity as simulating
$\mathcal{T} e^{-i \int_0^t \pm RL(s) + H(s)\,\mathrm{d}s}$
plus $\mathcal{O}\bigl(Q\bigl(\log(\|L\|_{L^1} + \log \tfrac{1}{\varepsilon})\bigr)\bigr)$
 two-qubit gates. \end{theorem}
For solving ODEs \eqref{linear ODEs}, the optimal LCHS method achieves optimal query complexity for both the initial state oracles and the block-encoding of coefficient matrix.

\section{Method of lines for linear parabolic PDEs}\label{section 3}

In this chapter we solve the parabolic PDE by method of lines and the optimal LCHS method. We begin by formulating the parabolic problem under three type of boundary conditions. Then we establish uniqueness of classical solutions and eliminate inhomogeneous boundary data, thereby ensuring well-posedness of the target PDE. Next, we apply the method of lines: spatial derivatives are discretized by finite differences on a grid, yielding a system of semi-discrete linear ODEs of the form 
$v'(t)=-Av(t)+b(t)$. To prepare this system for later LCHS algorithm, we introduce a similarity transform that places the coefficient matrix $A$ into a form with positive semi-definite real part. Finally we derive error bounds in the mean $l_2$ norm and the scale of number of nodes.

\subsection{Problem setting}
Let $T>0$ and $\Omega = [0,a_1]\times\cdots\times[0,a_d]\subset\mathbb{R}^d$ with $a_l>0,\ l=1,\cdots,d$.
We consider the second-order linear parabolic PDE
\begin{align}\label{parabolic pde}
\begin{cases}
u_t = \Delta u + \displaystyle\sum_{l=1}^d c_l\,u_{x_l} + f(x,t), & (x,t)\in \Omega\times[0,T],\\
u(x,0) = u_0(x), & x\in\Omega,
\end{cases}
\end{align}
where $c_l\in\mathbb{R}$ ($l=1,\dots,d$) are constants and $f,u_0$ are given.

\medskip
\noindent\textbf{Boundary conditions.} We impose boundary conditions dimension-by-dimension.
Let $S:=\{1,\dots,d\}$ and fix a disjoint partition
\[
S = S_1 \sqcup S_2 \sqcup S_3,
\]
corresponding to Dirichlet, Neumann, and periodic directions, respectively.

\smallskip
\noindent\emph{(i) Homogeneous Dirichlet boundaries.}
For each $l\in S_1$ we impose
\begin{equation}\label{boundary1}
u(x|_{x_l=0},t)=u(x|_{x_l=a_l},t)=0.
\end{equation}

\smallskip
\noindent\emph{(ii) Homogeneous Neumann boundaries.}
For each $l\in S_2$ we impose
\begin{equation}\label{boundary2}
\frac{\partial u}{\partial x_l}(x|_{x_l=0},t)=\frac{\partial u}{\partial x_l}(x|_{x_l=a_l},t)=0.
\end{equation}

\smallskip
\noindent\emph{(iii) Periodic boundaries.}
For each $l\in S_3$ we impose the periodicity conditions
\begin{equation}\label{boundary3}
\begin{cases}
u(x|_{x_l=0},t)=u(x|_{x_l=a_l},t),\\
\dfrac{\partial u}{\partial x_l}(x|_{x_l=0},t)=\dfrac{\partial u}{\partial x_l}(x|_{x_l=a_l},t),
\end{cases}
\end{equation}
with initial condition $u_0(x|_{x_l = 0})=u_0(x|_{x_l = a_l})$ and non-homogeneity $f(x|_{x_l = 0}, t)=f(x|_{x_l = a_l}, t)$.

The goal of this paper is to design an efficient quantum algorithm that produces a quantum state that is $\varepsilon$-close to $|[u(x,T)]_x \rangle$, which is the normalized solution discretized on certain nodes at final time $T$.
For completeness, we state a standard uniqueness result for classical solutions, ensuring that the algorithm is meaningful.
\begin{theorem}[Uniqueness]\label{thm parabolic unique}
If the equation \eqref{parabolic pde} with boundary conditions (\ref{boundary1}-\ref{boundary3}) admits a solution $u \in C^{2,1}(\Omega \times [0, T])$, then its solution is unique.
\end{theorem}

The proof can be found in Appendix \ref{appendix proof}.

\medskip
\noindent\textbf{Boundary lifting.} In some applications, one often encounters inhomogeneous Dirichlet and Neumann boundary conditions. We record a convenient transformation, based on \emph{Coons interpolation}, that reduces such inhomogeneous boundary conditions to homogeneous boundary conditions and add the influence of them to the inhomogeneous term $f$ and the initial value $u_0$.

For $l\in S_1$ (Dirichlet) we prescribe
\begin{align}\label{boundary1 nonhomo}
u(x|_{x_l=0},t)=g^{(L)}_l(x_1,\dots,\widehat{x_l},\dots,x_d,t),\quad
u(x|_{x_l=a_l},t)=g^{(R)}_l(x_1,\dots,\widehat{x_l},\dots,x_d,t),
\end{align}
and for $l\in S_2$ (Neumann) we prescribe the outward normal derivative
\begin{align}\label{boundary2 nonhomo}
-\frac{\partial u}{\partial x_l}(x|_{x_l=0},t)=g^{(L)}_l(x_1,\dots,\widehat{x_l},\dots,x_d,t),\quad
\frac{\partial u}{\partial x_l}(x|_{x_l=a_l},t)=g^{(R)}_l(x_1,\dots,\widehat{x_l},\dots,x_d,t).
\end{align}
Here $\widehat{x_l}$ means that the variable $x_l$ is omitted, and the minus sign at $x_l=0$ corresponds to the
outward normal direction.
Then we introduce the linear differential operator
\[
\mathcal{L}u := \Delta u + \sum_{l=1}^d c_l\,u_{x_l}.
\]
For a nonempty index set $A=\{l_1,\dots,l_a\}\subseteq S_1$ we define
\begin{equation}\label{thm parabolic vanish interpA}
\operatorname{Interp}_A(x,t)
:= \sum_{\delta_1,\dots,\delta_a\in\{0,1\}}
\Bigg(\prod_{k=1}^a \frac{x_{l_k}-\delta_k a_{l_k}}{a_{l_k}}\Bigg)
(-1)^{\delta_1+\cdots+\delta_a}\,
u\Big|_{x_{l_k}=a_{l_k}-\delta_k a_{l_k},\,k=1,\dots,a}.
\end{equation}
For a nonempty index set $B=\{l_1,\dots,l_b\}\subseteq S_2$ we define
\begin{equation}\label{thm parabolic vanish interpB}
\operatorname{Interp}'_B(x,t)
:= \sum_{\delta_1,\dots,\delta_b\in\{0,1\}}
\Bigg(\prod_{k=1}^b \frac{\frac12 x_{l_k}^2-\delta_k a_{l_k}x_{l_k}}{a_{l_k}}\Bigg)
(-1)^{\delta_1+\cdots+\delta_b}\,
\frac{\partial^b v_1}{\partial x_{l_1}\cdots\partial x_{l_b}}
\Bigg|_{x_{l_k}=a_{l_k}-\delta_k a_{l_k},\,k=1,\dots,b}.
\end{equation}
(Here $v_1$ is defined in \eqref{thm parabolic vanish -1} below.) We state this theorem as follows, and provide its proof in Appendix \ref{appendix proof}.


\begin{theorem}[Boundary lifting via Coons interpolation]\label{thm parabolic vanish}
Consider \eqref{parabolic pde} with periodic boundary conditions \eqref{boundary3} in directions $S_3$, and with the
inhomogeneous Dirichlet and Neumann boundary conditions~\eqref{boundary1 nonhomo}-\eqref{boundary2 nonhomo} in directions $S_1$ and $S_2$.
Let
\begin{equation}\label{thm parabolic vanish -1}
v_1(x,t) = u(x,t) + \sum_{\emptyset\neq A\subseteq S_1} (-1)^{|A|}\operatorname{Interp}_A(x,t),
\end{equation}
and
\begin{equation}\label{thm parabolic vanish -2}
v_2(x,t) = v_1(x,t) + \sum_{\emptyset\neq B\subseteq S_2} (-1)^{|B|}\operatorname{Interp}'_B(x,t).
\end{equation}
Then $v_2$ satisfies the same differential operator as \eqref{parabolic pde}
\[
(v_2)_t = \mathcal{L}v_2 + \widetilde{f}(x,t),\quad v_2(x,0)=\widetilde{u}_0(x),
\]
together with the homogeneous boundary conditions \eqref{boundary1}--\eqref{boundary2} and the periodic conditions
\eqref{boundary3}. The transformed inhomogeneous term and initial value are
\[
\widehat{f}(x,t) :=
-\mathcal{L}\Big(\sum_{\emptyset\neq A\subseteq S_1} (-1)^{|A|}\operatorname{Interp}_A(x,t)\Big)
+\frac{\partial}{\partial t}\Big(\sum_{\emptyset\neq A\subseteq S_1} (-1)^{|A|}\operatorname{Interp}_A(x,t)\Big)
+ f(x,t),
\]
\[
\widetilde{f}(x,t) :=
-\mathcal{L}\Big(\sum_{\emptyset\neq B\subseteq S_2} (-1)^{|B|}\operatorname{Interp}'_B(x,t)\Big)
+\frac{\partial}{\partial t}\Big(\sum_{\emptyset\neq B\subseteq S_2} (-1)^{|B|}\operatorname{Interp}'_B(x,t)\Big)
+ \widehat{f}(x,t),
\]
and
\[
\widehat{u}_0(x) := u_0(x) + \sum_{\emptyset\neq A\subseteq S_1} (-1)^{|A|}\operatorname{Interp}_A(x,0),\quad
\widetilde{u}_0(x) := \widehat{u}_0(x) + \sum_{\emptyset\neq B\subseteq S_2} (-1)^{|B|}\operatorname{Interp}'_B(x,0).
\] \end{theorem}
\begin{remark}[Practical computability]
Although \eqref{thm parabolic vanish interpA} and \eqref{thm parabolic vanish interpB} are written using traces of $u$ and traces of derivatives of $v_1$,
they are computable from the prescribed boundary data $g_l^{(L)},g_l^{(R)}$.
For $\operatorname{Interp}_A$, we have $u|_{x_l=0}=g_l^{(L)}$ and $u|_{x_l=a_l}=g_l^{(R)}$, so $\operatorname{Interp}_A$
depends only on $g^{(L/R)}_l$.
For $\operatorname{Interp}'_B$, note that $v_1=u+\text{(known Dirichlet extension)}$, hence $\partial_{x_l}v_1$ on a Neumann
face $x_l=0$ or $x_l=a_l$ can be obtained from the Neumann data for $u$ plus the explicitly differentiable Dirichlet extension.
Higher-order mixed derivatives appearing in \eqref{thm parabolic vanish interpB} are tangential derivatives of these known boundary functions. So $\operatorname{Interp}'_B$ also
depends on $g^{(L/R)}_l$.
\end{remark}

A similar interpolation can be used to enforce the initial value $\widetilde{u}_0\equiv 0$ as well.
We do not carry this out here because it may introduce an additional inhomogeneous term even when $f\equiv 0$,
which is undesirable in settings where one wishes to keep the PDE homogeneous in the interior.
So in this paper, we assume that the Dirichlet and Neumann boundary values are all zero, which is helpful for subsequent error analysis.

\subsection{Method of lines discretization}\label{section 3.2}

 In this subsection we discretize \eqref{parabolic pde} in space by the method of lines and obtain a semi-discrete ODE system
\begin{equation}\label{parabolic approx odes}
\begin{cases}
\displaystyle \frac{\mathrm{d}v}{\mathrm{d}t} = -A v + b(t),\\[2mm]
v(0)=v_0.
\end{cases}
\end{equation}

For the three types of boundary conditions, we construct their finite difference schemes respectively:

\medskip
\noindent\textbf{(i) Dirichlet boundary.} For $l$ in $S_1$, let the step size be $h_l = \dfrac{a_l}{N_l + 1}$ and nodes be 
$x_l^{(j)} =  (j + 1) h_l,\ j = 0, \cdots, N_l - 1$.
\begin{center}
\begin{tikzpicture}[x=1cm,y=1cm]
  \draw (0,0) -- (10,0);

  \node at (6,-0.4) {$\cdots$};
  \node[below] at (0,-0.1) {$0 = x^{(-1)}$};
  \node[below] at (2,-0.1) {$x^{(0)}$};
  \node[below] at (4,-0.1) {$x^{(1)}$};
  \node[below] at (8,-0.1) {$x^{(N_l-1)}$};
  \node[below] at (10,-0.1) {$x^{(N_l)} = a_l$};
\fill (2,0) circle (0.09);
\fill (4,0) circle (0.09);
\fill (8,0) circle (0.09);
\draw[fill=white] (10,0) circle (0.12);
\draw[fill=white] (0,0) circle (0.12);
\end{tikzpicture}\end{center}
Then we use the second-order difference formulas
\begin{align}
u_{x_l x_l}(\cdots, x_l^{(j)}, \cdots, t) &= \frac{1}{h_l^2} \left[ u(\cdots, x_l^{(j+1)}, \cdots, t) + u(\cdots, x_l^{(j-1)}, \cdots, t) \right.\notag \\
&\quad \left. - 2 u(\cdots, x_l^{(j)}, \cdots, t) \right] -\dfrac{1}{12}u_{x_lx_lx_lx_l}(\cdots, \xi, \cdots, t)h_l^2,\label{diffe ddx} \\
u_{x_l}(\cdots, x_l^{(j)}, \cdots, t) &= \frac{1}{2h_l} \left[ u(\cdots, x_l^{(j+1)}, \cdots, t) - u(\cdots, x_l^{(j-1)}, \cdots, t) \right] -\dfrac{1}{6}u_{x_lx_lx_l}(\cdots, \xi, \cdots, t)h_l^2.\label{diffe dx}
\end{align}
At the boundary nodes, by the Dirichlet boundary condition, we have $u(\cdots, x_l^{(-1)}, \cdots, t)=u(\cdots, x_l^{(N_l)}, \cdots, t)=0$. We define the difference matrices
\begin{align*}
    D_{D,\Delta}(N_l):=\begin{pmatrix}
        2 & -1  & &  \\
        -1 & 2 &\ddots &  \\
         &\ddots &\ddots & -1   \\
         & &-1 & 2
    \end{pmatrix}_{N_l},\ D_{D,\pm}(N_l):=\begin{pmatrix}
        0 & -1  & &  \\
        1 & 0 &\ddots &  \\
         &\ddots &\ddots & -1  \\
         & & 1 & 0  
    \end{pmatrix}_{N_l}.
\end{align*}
So for Dirichlet boundary, the differential operator can be approximated as follows
\begin{align}\label{para dirichlet matrix}
    \dfrac{\partial^2}{\partial x_l ^2}&\approx -\dfrac{1}{h_l^2}D_{D,\Delta}(N_l),\ \dfrac{\partial}{\partial x_l }\approx -\dfrac{1}{2h_l}D_{D,\pm}(N_l),\notag\\
    A_l&:=\dfrac{1}{h_l^2}D_{D,\Delta}(N_l)+\dfrac{c_l}{2h_l}D_{D,\pm}(N_l)\approx -\dfrac{\partial^2}{\partial x_l ^2}-c_l\dfrac{\partial}{\partial x_l }.
\end{align}

\medskip
\noindent\textbf{(ii) Neumann boundary.} For $l$ in $S_2$, let the step size be $h_l = \dfrac{a_l}{N_l}$ and nodes be
$x_l^{(j)} =  (j + \dfrac{1}{2}) h_l,\ j = 0, \cdots, N_l - 1$. 
\begin{center}
    \begin{tikzpicture}[x=1cm,y=1cm]
  \draw
 (0,0) -- (10,0);
  \draw[fill=white] (0,0) circle (0.12);
  \draw[fill=white] (2,0) circle (0.12);    
  \draw[fill=white] (4,0) circle (0.12);
  \draw[fill=white] (8,0) circle (0.12);
  \draw[fill=white] (10,0) circle (0.12);
  \fill (1,0) circle (0.09);   
  \fill (3,0) circle (0.09);   
  \fill (9,0) circle (0.09);   
  \node at (6,-0.4) {$\cdots$
};
  \node[below] at (0,-0.15) {$0$
};
  \node[below] at (1,-0.1) {$x^{(0)}$
};
  \node[below] at (3,-0.1) {$x^{(1)}$
};
  \node[below] at (9,-0.1) {$x^{(N_l-1)}$
};
  \node[below] at (10,-0.15) {$a_l$
};

\end{tikzpicture}
\end{center}
At the interior nodes, we use the same difference formulas (\ref{diffe ddx}-\ref{diffe dx}). At the boundary nodes (i.e. $j=0,N_l-1$), for $\dfrac{\partial^2}{\partial x_l^2}$, we have the first-order difference formula
\begin{align*}
u_{x_l x_l}(\cdots, x_l^{(0)}, \cdots, t) &= \frac{1}{h_l} \left[ u_{x_l}(\cdots, h_l, \cdots, t) - u_{x_l}(\cdots, 0, \cdots, t) \right] + \mathcal{O}(h_l^2) \\
&= \frac{1}{h_l} \left[ \frac{1}{h_l} (u(\cdots, x_l^{(1)}, \cdots, t) - u(\cdots, x_l^{(0)}, \cdots, t))+ \mathcal{O}(h_l^2) \right] + \mathcal{O}(h_l^2) \\
&= \frac{1}{h_l^2} \left[u(\cdots, x_l^{(1)}, \cdots, t)-u(\cdots, x_l^{(0)}, \cdots, t) \right] + \mathcal{O}(h_l).
\end{align*}
Similarly, for $j=N_l-1$ we have
\begin{align*}
u_{x_l x_l}(\cdots, x_l^{(N_l-1)}, \cdots, t) &= \frac{1}{h_l^2} \left[ -u(\cdots, x_l^{(N_l-1)}, \cdots, t) + u(\cdots, x_l^{(N_l-2)}, \cdots, t) \right] + \mathcal{O}(h_l).
\end{align*}
For $\dfrac{\partial}{\partial x_l}$, we have $\dfrac{u(\cdots, x_l^{(0)}, \cdots, t)-u(\cdots, x_l^{(-1)}, \cdots, t)}{2h_l}=\dfrac{h_l u_{x_l}(\cdots, 0, \cdots, t)+\mathcal{O}(h_l^3)}{2h_l}=\mathcal{O}(h_l^2).$ Then we get the second-order difference formula \begin{align*}u_{x_l}(\cdots, x_l^{(0)}, \cdots, t)&=\dfrac{u(\cdots, x_l^{(1)}, \cdots, t)-u(\cdots, x_l^{(-1)}, \cdots, t)}{2h_l}+\mathcal{O}(h_l^2)\\
&=\dfrac{u(\cdots, x_l^{(1)}, \cdots, t)-u(\cdots, x_l^{(0)}, \cdots, t)}{2h_l}+\mathcal{O}(h_l^2).\end{align*}
Similarly, for $j=N_l-1$ we have
$$u_{x_l}(\cdots, x_l^{(N_l-1)}, \cdots, t)
=\dfrac{u(\cdots, x_l^{(N_l-1)}, \cdots, t)-u(\cdots, x_l^{(N_l-2)}, \cdots, t)}{2h_l}+\mathcal{O}(h_l^2).$$
We define the difference matrices
\begin{align*}    
D_{N,\Delta}(N_l):=\begin{pmatrix}
        1 & -1  & & & \\
        -1 & 2 & -1 & & \\
         &-1 &\ddots &\ddots &  \\
         & &\ddots & 2 & -1 \\
         & & & -1 & 1
    \end{pmatrix}_{N_l},\ 
D_{N,\pm}(N_l):=\begin{pmatrix}
        1 & -1  & & &  \\
        1 & 0 &\ddots & & \\
         &\ddots &\ddots &-1 &  \\
         & &1 &0 & -1 \\
         & & & 1 & -1
    \end{pmatrix}_{N_l}.
\end{align*}
So for Neumann boundary, the differential operator can be approximated as follows
\begin{align}\label{para neumann matrix}
    \dfrac{\partial^2}{\partial x_l ^2}&\approx -\dfrac{1}{h_l^2}D_{N,\Delta}(N_l),\ \dfrac{\partial}{\partial x_l }\approx -\dfrac{1}{2h_l}D_{N,\pm}(N_l),\notag \\
    A_l&:=\dfrac{1}{h_l^2}D_{N,\Delta}(N_l)+\dfrac{c_l}{2h_l}D_{N,\pm}(N_l)\approx -\dfrac{\partial^2}{\partial x_l ^2}-c_l\dfrac{\partial}{\partial x_l }.
\end{align}
\begin{remark}
    For Neumann boundary conditions we place the nodes at cell midpoints and discretize the second-order derivative at the boundary by a first-order difference. This choice is not made for convenience, but to enforce the operator property required by the LCHS framework: the coefficient matrix $A$ must have a positive semi-definite real part. With midpoint nodes and the first-order local truncation error at boundary nodes, the resulting finite-difference matrix will satisfy $\frac{1}{2}(A+A^\dag)\succeq 0$ after the transformation in next subsection. And later we will show that the overall error can remain second-order.
    By contrast, if one uses the traditional grid including interval endpoints, or discretizes the second-order derivative by a second-order difference, we observe that the real part of $A$ will have a negative eigenvalue of scale $\mathcal{O}(1/h_l^2)=\mathcal{O}(N_l^2)$. This violates the LCHS applicability conditions. More importantly, we can not even use the transformation $\tilde{v}=e^{-ct}v,\ \tilde{A}=A+cI$, because it will lead to an error of scale $\mathcal{O}(\exp(N_l^2))$, which makes the LCHS-based solver ineffective.
\end{remark}

\medskip
\noindent\textbf{(iii) Periodic boundary.} For $l$ in $S_3$, let the step size be $h_l = \dfrac{a_l}{N_l}$ and nodes be
$x_l^{(j)} =  j h_l,\ j = 0, \cdots, N_l - 1$. 
\begin{center}
    \begin{tikzpicture}[x=1cm,y=1cm]
  \draw (0,0) -- (10,0);

  \node at (5,-0.4) {$\cdots$};
\fill (0,0) circle (0.09);
\fill (2,0) circle (0.09);
\fill (8,0) circle (0.09);
\draw[fill=white] (10,0) circle (0.12);
  \node[below] at (0,-0.1) {$0 = x^{(0)}$};
  \node[below] at (2,-0.1) {$x^{(1)}$};
  \node[below] at (8,-0.1) {$x^{(N_l-1)}$};
  \node[below] at (10,-0.1) {$x^{(N_l)} = a_l$};
\end{tikzpicture}
\end{center}
At the interior nodes, we also use (\ref{diffe ddx}-\ref{diffe dx}). At the boundary grids, by the periodic boundary \eqref{boundary3}, we have $u(\cdots, x_l^{(-1)}, \cdots, t) = u(\cdots, x_l^{(N_l-1)}, \cdots, t)$. So
\begin{align*}
u_{x_l x_l}(\cdots, x_l^{(0)}, \cdots, t) &= \frac{1}{h_l^2} \left[ u(\cdots, x_l^{(-1)}, \cdots, t) + u(\cdots, x_l^{(1)}, \cdots, t)  - 2 u(\cdots, x_l^{(0)}, \cdots, t) \right] +\mathcal{O}(h_l^2) \\
&= \frac{1}{h_l^2} \left[ u(\cdots, x_l^{(N_l-1)}, \cdots, t) + u(\cdots, x_l^{(1)}, \cdots, t) - 2 u(\cdots, x_l^{(0)}, \cdots, t) \right] +\mathcal{O}(h_l^2).
\end{align*}
$u_{x_l x_l}(\cdots, x_l^{(N_l-1)}, \cdots, t)$, $u_{x_l}(\cdots, x_l^{(0)}, \cdots, t)$, $u_{x_l}(\cdots, x_l^{(N_l-1)}, \cdots, t)$ can be approximated similarly. We define the difference matrices
\begin{align*}    
D_{P,\Delta}(N_l):=\begin{pmatrix}
        2 & -1  & & & -1 \\
        -1 & 2 &\ddots & & \\
         &\ddots &\ddots &-1 &  \\
         & &-1 &2 & -1 \\
        -1 & & & -1 & 2
    \end{pmatrix}_{N_l},\ 
D_{P,\pm}(N_l):=\begin{pmatrix}
        0 & -1  & & & 1 \\
        1 & 0 &\ddots & & \\
         &\ddots &\ddots &-1 &  \\
         & &1 &0 & -1 \\
        -1 & & & 1 & 0
    \end{pmatrix}_{N_l}.
\end{align*}
So for periodic boundary, the differential operator can be approximated as follows 
\begin{align}\label{para periodic matrix}
    \dfrac{\partial^2}{\partial x_l ^2}&\approx -\dfrac{1}{h_l^2}D_{P,\Delta}(N_l),\ \dfrac{\partial}{\partial x_l }\approx -\dfrac{1}{2h_l}D_{P,\pm}(N_l),\notag \\
    A_l&:=\dfrac{1}{h_l^2}D_{P,\Delta}(N_l)+\dfrac{c_l}{2h_l}D_{P,\pm}(N_l)\approx -\dfrac{\partial^2}{\partial x_l ^2}-c_l\dfrac{\partial}{\partial x_l }.
\end{align}

Let $N:=N_1\cdots N_d$ and collect the node values into a vector
\[
v(t)\in\mathbb{R}^N,\quad
v_{j_1,\dots,j_d}(t)\approx u(x_1^{(j_1)},\dots,x_d^{(j_d)},t).
\]
Define the Kronecker-sum operator
\begin{equation}\label{parabolic coefficient matrix}
A := \sum_{l=1}^d I_{N_1}\otimes\cdots\otimes I_{N_{l-1}}\otimes A_l\otimes I_{N_{l+1}}\otimes\cdots\otimes I_{N_d}.
\end{equation}
Let $b(t)$ and $v_0$ be the inhomogeneous term and initial value sampled on the grid:
\[
b(t) := [f(x,t)]_{x}\in\mathbb{R}^N,\quad
v_0 := [u_0(x)]_{x}\in\mathbb{R}^N.
\]
Then the method of lines yields the approximate ODE system \eqref{parabolic approx odes}.

\subsection{Similarity transform}\label{section 3.3}

To use the LCHS method, we need to ensure that the coefficient matrix $A$ has positive semi-definite real part. But we can prove that the difference matrix for the Neumann boundary $A_l,\ l\in S_2$ has an indefinite real part, whose negative eigenvalue has scale $\mathcal{O}(1)$. Of course we can use transformation $\tilde{v}=e^{-ct}v,\ \tilde{A}=A+cI$ where $c=\mathcal{O}(1)$. But it also causes an error of $\mathcal{O}(e^{cT})$, which performs poorly for long-time simulation.

So in this subsection, we make the real part of coefficient matrix positive semi-definite  by a similarity transformation.
For $l\in S_1\cup S_2$, we assume $|c_l|h_l/2< 1$ and let \begin{align}\label{para trans matrix Pl}
    P_l:=\diag(1, \theta_l, \cdots,\theta_l^{N_l-1}),\ \theta_l:=\sqrt{\dfrac{1+c_lh_l/2}{1-c_lh_l/2}}. 
\end{align}
For $l\in S_3$ we set $P_l:=I_{N_l}$.
For $l\in S_1$, by Lemma \ref{lem dirichlet matrix}, the difference matrix in \eqref{para dirichlet matrix} becomes positive definite after the similarity transformation:
\begin{align}\label{para diri tilde Al}
    A_l&=\dfrac{1}{h_l^2}\begin{pmatrix}        
2 & -1-c_lh_l/2  & &  \\        
-1+c_lh_l/2 & 2 &\ddots &  \\         
&\ddots &\ddots & -1-c_lh_l/2   \\         
& &-1+c_lh_l/2 & 2    
\end{pmatrix}_{N_l} \notag \\
&=P_l^{-1}\cdot \dfrac{1}{h_l^2}\begin{pmatrix}        
2 & -\sqrt{1-c_l^2h_l^2/4}  & &  \\        
-\sqrt{1-c_l^2h_l^2/4} & 2 &\ddots &  \\         
&\ddots &\ddots & -\sqrt{1-c_l^2h_l^2/4}   \\         
& &-\sqrt{1-c_l^2h_l^2/4} & 2    
\end{pmatrix}_{N_l}\cdot P_l \notag \\
&:=P_l^{-1} \widetilde{A}_lP_l.
\end{align}
For $l\in S_2$, by Lemma \ref{lem neumann matrix}, the difference matrix in \eqref{para neumann matrix} become positive semi-definite after the similarity transformation:
\begin{align}\label{para neu tilde Al}    
A_l&=\dfrac{1}{h_l^2}\begin{pmatrix}        
1+c_lh_l/2 & -1-c_lh_l/2  & & & \\        
-1+c_lh_l/2 & 2 & -1-c_lh_l/2 & & \\         
&-1+c_lh_l/2 &\ddots &\ddots &  \\         
& &\ddots & 2 & -1-c_lh_l/2 \\         
& & & -1+c_lh_l/2 & 1-c_lh_l/2    
\end{pmatrix}_{N_l}\notag \\
&=P_l^{-1}\cdot \dfrac{1}{h_l^2}\begin{pmatrix}        
1+c_lh_l/2 & -\sqrt{1-c_l^2h_l^2/4}  & & & \\        
-\sqrt{1-c_l^2h_l^2/4} & 2 & -\sqrt{1-c_l^2h_l^2/4} & & \\         
&-\sqrt{1-c_l^2h_l^2/4} &\ddots &\ddots &  \\         
& &\ddots & 2 & -\sqrt{1-c_l^2h_l^2/4} \\         
& & & -\sqrt{1-c_l^2h_l^2/4} & 1-c_lh_l/2    
\end{pmatrix}_{N_l}\cdot P_l\notag \\
&:=P_l^{-1} \widetilde{A}_l P_l.
\end{align}

Define the global similarity transform \begin{align}\label{para trans matrix}
    P:=\bigotimes_{l=1}^d P_l\ \text{and}\ \widetilde{A}:=PAP^{-1}=\sum_{l=1}^d I_{N_1}\otimes\cdots\otimes I_{N_{l-1}}\otimes \widetilde{A}_l\otimes I_{N_{l+1}}\otimes\cdots\otimes I_{N_d}.
\end{align}
Then we get the analytic solution of \eqref{parabolic approx odes}
\begin{align}\label{para analytic solution}
v(T) &=e^{-TA}v_0+ \int_0^T e^{-(T - s)A}b(s) \, ds\notag \\
&=P^{-1}e^{-T\widetilde{A}}Pv_0+ \int_0^T P^{-1}e^{-(T - s)\widetilde{A}} P b(s) \, ds.
\end{align}

Now the coefficient matrix $\widetilde{A}$ has positive semi-definite real part and its imaginary part only comes from the periodic boundary.

\subsection{Error analysis}\label{section 3.4}

Now we analyze the error between the PDE \eqref{parabolic pde} and the ODEs \eqref{parabolic approx odes}. Because the finite difference scheme of Neumann boundary is first-order at boundary nodes, we specifically show that its overall error is second-order in Lemma \ref{lem para neu mol}.  Then in Theorem \ref{thm para mol}, we compute the scale of $N_l$.

\begin{lemma}\label{lem para neu mol}
Let $g:[0,a]\to\R$ be a sufficiently smooth function and $g'(0)=g'(a)=0$. Let $h=a/N$ and
\[
  x_j = \Bigl(j+\tfrac12\Bigr)h,\quad j=0,\dots,N-1,
\]
be the midpoint grid associated with the Neumann boundary discretization. Assume $|ch|/2<1$. Let
\[
A = \frac1{h^2}
\begin{pmatrix}
1 & -1 \\
-1 & 2 & \ddots \\
& \ddots & \ddots & -1 \\
& & -1 & 1
\end{pmatrix}
+
\frac{c}{2h}
\begin{pmatrix}
1 & -1 \\
1 & 0 & \ddots \\
& \ddots & \ddots & -1 \\
& & 1 & -1
\end{pmatrix},
\]
be the difference matrix of $g''+cg'$ and
\[
  P = \diag(1,\theta,\dots,\theta^{N-1}),
  \quad
  \theta = \sqrt{\frac{1+ch/2}{1-ch/2}}.
\]
Let \[\widetilde{A}=PAP^{-1}=DD^T\] be the similarity transformation as in Lemma \ref{lem neumann matrix} and \[\tau := A[g]_x + [g'' + c g']_x\]  be the local truncation error. Then

\noindent (1) $\tau$ has the decomposition \[\tau=\sqrt{1+ch/2}P^{-1}DP\eta+\rho,\] where $\norm{\eta}_{\star}=\mathcal{O}(h^2 \norm{g^{(3)}}_{\infty}),\ \norm{\rho}_{\star}=\mathcal{O}(h^2\norm{g^{(4)}}_{\infty}
  + h^2|c|\norm{g^{(3)}}_{\infty});$

\noindent (2) for every $t>0$,
\[
  \starnorm{e^{-tA}\tau}
  \lesssim
  \cond(P) h^2
  \left(
    \frac1{\sqrt{t}}\norm{g^{(3)}}_{\infty}
    + \norm{g^{(4)}}_{\infty}
    + |c|\norm{g^{(3)}}_{\infty}
  \right).
\]
\end{lemma}

\begin{proof}
(1) Set \(  \alpha:=\frac{ch}{2}\). By Lemma \ref{lem neumann matrix} we have \[D = \frac{1}{h}
\begin{pmatrix}
\sqrt{1+\alpha} &        &        &        &        \\
-\sqrt{1-\alpha} & \sqrt{1+\alpha} &        &        &        \\     
& -\sqrt{1-\alpha} & \ddots &        &        \\     
&      & \ddots & \sqrt{1+\alpha} &        \\     
&      &        & -\sqrt{1-\alpha} & 0
\end{pmatrix}.\]
Let the backward difference matrix be \[E:=\frac1h
\begin{pmatrix}
1 \\
-1 & 1 \\
& -1 & \ddots \\
& & \ddots & 1 \\
& & & -1 & 0
\end{pmatrix}.\] Then $A$ can be decomposed as
\begin{align*}
A= \frac1h
\begin{pmatrix}
1+\alpha \\
-1+\alpha & 1+\alpha \\
& -1+\alpha & \ddots \\
& & \ddots & 1+\alpha \\
& & & -1+\alpha & 0
\end{pmatrix}\cdot E^T
=\sqrt{1+\alpha}P^{-1}DP\cdot E^T.
\end{align*}
Let
\[
  \beta := \bigl(g'(h), g'(2h), \dots, g'((N-1)h), g'(Nh)\bigr)^T.
\]
Then
\begin{align*}
  \tau
  &= \sqrt{1+\alpha}P^{-1}DP\cdot E^T[g]_x + [g'' + c g']_x \\
  &= \sqrt{1+\alpha}P^{-1}DP\bigl(\beta + E^T[g]_x\bigr) + \bigl([g''+c g']_x - \sqrt{1+\alpha}P^{-1}DP\beta\bigr) \\
  &= \sqrt{1+\alpha}P^{-1}DP\eta + \rho,
\end{align*}
where
\[  \eta := \beta + E^T[g]_x,
  \quad
  \rho := [g''+cg']_x - \sqrt{1+\alpha}P^{-1}DP\beta.\]

Then we estimate the norms of $\eta$ and $\rho$. By direct expansion,
\[\eta =
  \begin{pmatrix}
    g'(h) - \dfrac{g(x_1)-g(x_0)}{h} \\
    \vdots \\
    g'((N-1)h) - \dfrac{g(x_{N-1})-g(x_{N-2})}{h} \\
    0
  \end{pmatrix}
  =
  \begin{pmatrix}
    \dfrac1{24} h^2 g^{(3)}(\xi_1) \\
    \vdots \\
    \dfrac1{24} h^2 g^{(3)}(\xi_{N-1}) \\
    0
  \end{pmatrix}
\]
for suitable intermediate points $\xi_j$. Therefore
\begin{equation}\label{lem para neu mol 2}
  \starnorm{\eta}
  \leq \frac1{24} h^2 \norm{g^{(3)}}_{\infty}.
\end{equation}

Then we compute
\[
  \rho =
  \begin{pmatrix}
    g''(x_0)+cg'(x_0)-\dfrac{(1+\alpha)g'(h)+(-1+\alpha)g'(0)}{h} \\
    \vdots \\
    g''(x_{N-2})+cg'(x_{N-2})-\dfrac{(1+\alpha)g'((N-1)h)+(-1+\alpha)g'((N-2)h)}{h} \\
    g''(x_{N-1})+cg'(x_{N-1})-\dfrac{(1+\alpha)g'(Nh)+(-1+\alpha)g'((N-1)h)}{h}
  \end{pmatrix}.
\]
Rearranging the coefficients and recalling $\alpha=ch/2$ give
\[
  \rho =
  \begin{pmatrix}
    g''(x_0)-\dfrac{g'(h)-g'(0)}{h} \\
    \vdots \\
    g''(x_{N-2})-\dfrac{g'((N-1)h)-g'((N-2)h)}{h} \\
    g''(x_{N-1})-\dfrac{g'(Nh)-g'((N-1)h)}{h}
  \end{pmatrix}
  + c
  \begin{pmatrix}
    g'(x_0)-\dfrac{g'(h)+g'(0)}2 \\
    \vdots \\
    g'(x_{N-2})-\dfrac{g'((N-1)h)+g'((N-2)h)}2 \\
    g'(x_{N-1})-\dfrac{g'(Nh)+g'((N-1)h)}2
  \end{pmatrix}.
\]
By Taylor expansion at the midpoint locations,
\[ \rho=
  \frac1{24}h^2
  \begin{pmatrix}
    g^{(4)}(\zeta_0) \\
    \vdots \\
    g^{(4)}(\zeta_{N-1})
  \end{pmatrix}
  + c\,\frac18 h^2
  \begin{pmatrix}
    g^{(3)}(\widetilde\zeta_0) \\
    \vdots \\
    g^{(3)}(\widetilde\zeta_{N-1})
  \end{pmatrix}\]
for suitable intermediate points $\zeta_j$ and $\widetilde\zeta_j$. Hence
\begin{equation}\label{lem para neu mol 3}
  \starnorm{\rho}
  \le \frac1{24}h^2\norm{g^{(4)}}_{\infty}
  + \frac{|c|}{8}h^2\norm{g^{(3)}}_{\infty}.
\end{equation}

(2) By the conclusion in (1),
\begin{equation}\label{lem para neu mol 1}
  e^{-tA}\tau =e^{-tA}\sqrt{1+\alpha}P^{-1}DP\eta + e^{-tA}\rho.
\end{equation}
Since $A=P^{-1}\widetilde{A}P$ and $\widetilde{A}=DD^T$, we have
\begin{align*}
  e^{-tA}\sqrt{1+\alpha}P^{-1}DP=\sqrt{1+\alpha} P^{-1} e^{-t\widetilde A} DP = \sqrt{1+\alpha}\, P^{-1} e^{-tDD^T} D P.
\end{align*}
Let \(  B := D^T D \succeq 0.\)
Using the power-series expansion,
\[
  e^{-t D D^T} D
  = \sum_{k=0}^{\infty} \frac1{k!}(-t D D^T)^k D
  =  D \sum_{k=0}^{\infty} \frac1{k!}(-tD^T D)^k
  =  D e^{-tB}.
\]
Hence
\begin{align*}
  \norm{e^{-t D D^T} D}_{2}
  = \norm{ D e^{-tB}}_2 = \sqrt{\lambda_{\max}\bigl(e^{-tB} D^T D e^{-tB}\bigr)} = \sqrt{\lambda_{\max}(B e^{-2tB})}.
\end{align*}
For the scalar function $f(\lambda)=\lambda e^{-2t\lambda}$, the maximizer occurs at \(  \lambda = \frac1{2t},\) so \(  f(\lambda) \le f\Bigl(\frac1{2t}\Bigr) = \frac1{2et}.\)
Therefore
\[
  \norm{e^{-t D D^T} D}_{2}
  \le \frac1{\sqrt{2et}}.
\]
It follows that
\begin{equation}\label{lem para neu mol 4}
  \norm{e^{-tA}\sqrt{1+\alpha}P^{-1}DP}_2
  \le \sqrt{1+\alpha}\,\cond(P)\,\frac1{\sqrt{2et}}
  \lesssim \cond(P)\,\frac1{\sqrt{t}}.
\end{equation}
Since $\widetilde{A}\succeq 0$, we have
$  \norm{e^{-tA}}_2
  = \norm{P^{-1}e^{-t\widetilde A}P}_2
  \leq \cond(P).$
Combining \eqref{lem para neu mol 2}, \eqref{lem para neu mol 3}, \eqref{lem para neu mol 1} and \eqref{lem para neu mol 4}, we obtain
\begin{align*}
  \starnorm{e^{-tA}\tau}
  &\leq \norm{e^{-tA}\sqrt{1+\alpha}P^{-1}DP}_2\,\starnorm{\eta} + \norm{e^{-tA}}_2\,\starnorm{\rho} \\
  &\lesssim \cond(P) h^2
  \left(
    \frac1{\sqrt{t}}\norm{g^{(3)}}_{\infty}
    + \norm{g^{(4)}}_{\infty}
    + |c|\norm{g^{(3)}}_{\infty}
  \right).
\end{align*}
This completes the proof.
\end{proof}

\begin{theorem}\label{thm para mol}
Assume $u$ is sufficiently smooth. 
For error $\varepsilon_N$, if $N_l$ satisfies \begin{align}\label{thm para mol N1}     
N_l \gtrsim \sqrt{ \frac{Td}{\varepsilon_N} } a_l \exp(\sum_{l' \in S_1,S_2} \frac{|c_{l'}|a_{l'}}{4})\left( \max \{ |c_l|, 1 \} \cdot \max \{ \| u_{x_l}^{(3)} \|_{\infty}, \| u_{x_l}^{(4)} \|_{\infty} \} \right)^{\frac{1}{2}},
\end{align}
then the overall error between the solution of PDE \eqref{parabolic pde}  and the solution of its semi-discrete ODEs \eqref{parabolic approx odes} can be controlled by \( \varepsilon_N \), i.e. 
\[
\| [ u(x, T) ]_{x} - v(T) \|_{\star} \leq \varepsilon_N.
\]
\end{theorem}
\begin{proof}
Let the local truncation error and the overall error be
\[
  \tau(t) := \bigl[\Delta u(x,t) + \sum_{l=1}^d c_l u_{x_l}(x,t)\bigr]_x + A[u(x,t)]_x,
  \quad
  e(t) := [u(x,t)]_x - v(t).
\]
Then $e(t)$ satisfies
\[
  \begin{cases}
    \dfrac{\dd}{\dd t} e = -Ae + \tau(t),\\[0.3em]
    e(0)=0.
  \end{cases}
\]
Solving the ODE gives
\[
  e(T) = \int_0^T e^{-(T-s)A}\tau(s)\,\dd s
  = \sum_{l=1}^d \int_0^T e^{-(T-s)A}\tau_l(s)\,\dd s,
\]
where
\[
  \tau_l(t)
  :=
  [u_{x_lx_l}(x,t)+c_lu_{x_l}(x,t)]_x
  + I^{\otimes(l-1)}\otimes A_l\otimes I^{\otimes(d-l)}[u(x,t)]_x.
\]
To guarantee $\starnorm{e(T)}\le \varepsilon_N$, it suffices to show that, for each $l$,
\begin{align}\label{thm para mol 2}
  \left\|\int_0^T e^{-(T-s)A}\tau_l(s)\,\dd s\right\|_{\star}
  \le \frac{\varepsilon_N}{d}.
\end{align}

\medskip
\noindent
\textbf{Step 1: norm of $P$.}
By \eqref{para trans matrix}, we have
\[
\mathrm{cond}(P) = \| P \|_2 \| P^{-1} \|_2 = \prod_{l=1}^d \| P_l \|_2 \cdot \prod_{l=1}^d \| P_l^{-1} \|_2 = \prod_{l \in S_1,S_2} \| P_l \|_2 \| P_l^{-1} \|
_2.
\]
For 
$P_l$
, by \eqref{para trans matrix Pl}, we have
\[ \| P_l \|_2 = \begin{cases} 
\theta_l^{N_l - 1}, & c_l \geq 0 \\
1, & c_l < 0 
\end{cases}, \quad \| P_l^{-1} \|_2 = \begin{cases} 
1, & c_l \geq 0 \\
\theta_l^{-N_l + 1}, & c_l < 0 
\end{cases}.\]
So the condition number of $P_l$ satisfies
\begin{align*} \| P_l \|_2 \cdot \| P_l^{-1} \|_2 &= \begin{cases} 
\theta_l^{N_l - 1}, & c_l \geq 0 \\
\left( \dfrac{1}{\theta_l} \right
)^{N_l - 1}, & c_l < 0 
\end{cases} = \sqrt{\dfrac{1 + |c_l| h_l/2}{1 - |c_l| h_l/2}}^{N_l - 1}\\
&= \left( 1 + \frac{|c_l| h_l}{1 - |c_l| h_l/2} \right)^{\frac{N_l - 1}{2}} \leq \exp\left( \frac{|c_l| h_l}{1 - |c_l| h_l/2} \cdot \frac{N_l - 1}{2} \right
)\\
&\leq \exp\left( \frac{|c_l| a_l}{N_l - |c_l| a_l/2} \cdot \frac{N_l - 1}{2} \right) \leq C_0 \exp(\frac
{|c_l| a_l}{2})
\end{align*}for some constant $C_0$.
Then for some constant $C$
\begin{align}\label{para condP}
\mathrm{cond}(P) \leq C \exp\left( \sum_{l \in S_1,S_2} \frac{|c_l| a_l}{2} \right).
\end{align}

\medskip
\noindent
\textbf{Step 2: Dirichlet and periodic boundaries.}
For $l$ in \( S_1 \) and \( S_3 \), by \eqref{diffe ddx} and \eqref{diffe dx}, we have the truncation error is second-order
\begin{align*}
&\left| u_{x_l x_l} (\cdots, x_l^{(j)}, \cdots, t) + \frac{1}{h_l^2} \left( 2u(\cdots, x_l^{(j)}, \cdots, t) - u(\cdots, x_l^{(j+1)}, \cdots, t) - u(\cdots, x_l^{(j-1)}, \cdots, t) \right) \right| 
\leq \frac{1}{12} h_l^2 \max|u^{(4)}_{x_l}|, \\
&\left| u_{x_l} (\cdots, x_l^{(j)}, \cdots, t) + \frac{1}{2h_l} \left( -u(\cdots, x_l^{(j+1)},\cdots, t) + u(\cdots, x_l^{(j-1)},\cdots, t) \right) \right| \leq \frac{1}{6} h_l^2 \max|u^{(3)}_{x_l}|.    
\end{align*}
So as a \( N_1 \cdots N_d \)-dimensional vector, every element of $\tau_l(t)$ satisfies
$$\left|(u_{x_lx_l}+c_lu_{x_l})(x_1^{(j_1)},\cdots,x_d^{(j_d)},t)+(I^{\otimes (l-1)} \otimes A_l \otimes I^{\otimes(d - l)}[u]_{x})_{j_1,\cdots,j_d}\right| \leq \left(\frac{1}{12}\|u_{x_l}^{(4)}\|_{\infty}+\frac{|c_l|}{6}\|u_{x_l}^{(3)}\|_{\infty}\right)h^2_l.$$
By the definition of $\|\cdot\|_{\star}$, we have 
\begin{align}\label{thm hy mol 1}
\|\tau_l(t)\|_{\star} \leq  \left(\frac{1}{12}\|u_{x_l}^{(4)}\|_{\infty}+\frac{|c_l|}{6}\|u_{x_l}^{(3)}\|_{\infty}\right)h^2_l
\leq \left(\frac{1}{12}\|u_{x_l}^{(4)}\|_{\infty}+\frac{|c_l|}{6}\|u_{x_l}^{(3)}\|_{\infty}\right) \frac{a_l^2}{N_l^2}.\end{align}
So
\begin{align*}
  \left\|\int_0^T e^{-(T-s)A}\tau_l(s)\,\dd s\right\|_{\star}
  &\le
  \left\|\int_0^T P^{-1}e^{-(T-s)\widetilde A}P\tau_l(s)\,\dd s\right\|_{\star} \\
  &\le
  \int_0^T \cond(P)\,\norm{e^{-(T-s)\widetilde A}}_2\,\starnorm{\tau_l(s)}\,\dd s \\
  &\le
  T\cond(P)
  \left(
    \frac1{12}\|u_{x_l}^{(4)}\|_{\infty} + \frac{|c_l|}{6}\|u_{x_l}^{(3)}\|_{\infty}
  \right)
  \frac{a_l^2}{N_l^2}.
\end{align*}
Hence to let \eqref{thm para mol 2} hold true, it is enough to require
\[
  N_l \geq
  \sqrt{\frac{C T d}{\varepsilon_N}}
  \, a_l
  \exp\!\left(\sum_{l'\in S_1\cup S_2}\frac{|c_{l'}|a_{l'}}{4}\right)
  \left(
    \frac1{12}\|u_{x_l}^{(4)}\|_{\infty} + \frac{|c_l|}{6}\|u_{x_l}^{(3)}\|_{\infty}
  \right)^{1/2}.
\]

\medskip
\noindent
\textbf{Step 3: Neumann boundaries.}
For $l\in S_2$, the one-dimensional lemma \ref{lem para neu mol} above can be applied in the $l$-th direction. By (2) in Lemma \ref{lem para neu mol}, we have
\[\|I^{\otimes (l-1)}\otimes e^{-tA_l} \otimes^{\otimes (d-l)} \tau_l(s)\|_{\star}\lesssim \mathrm{cond}(P_l)h_l^2\left(\frac1{\sqrt{t}}\norm{u_{x_l}^{(3)}}_{\infty}    
+ \norm{u_{x_l}^{(4)}}_{\infty}    
+ |c_l|\norm{u_{x_l}^{(3)}}_{\infty}\right).\]
Using
\(e^{-(T-s)A} = \bigotimes_{l'=1}^d e^{-(T-s)A_{l'}},\)
we get
\begin{align*}  
\left\|\int_0^T e^{-(T-s)A}\tau_l(s)\,\dd s\right\|_{\star} &= \left\|\int_0^T\bigotimes_{l'\neq l} e^{-(T-s)A_{l'}} I^{\otimes (l-1)}\otimes e^{-(T-s)A_l} \otimes^{\otimes (d-l)} \tau_l(s)\,\dd s\right\|_{\star}\\
&\lesssim  
\int_0^T  
\prod_{l'\ne l}\cond(P_{l'})  
\cdot  
\cond(P_l) h_l^2  
\left(    
\frac1{\sqrt{T-s}}\norm{u_{x_l}^{(3)}}_{\infty}    
+ \norm{u_{x_l}^{(4)}}_{\infty}    
+ |c_l|\norm{u_{x_l}^{(3)}}_{\infty}  
\right)  
\dd s\\
&\lesssim 
\cond(P)  
\left(    
\int_0^T \frac{\dd s}{\sqrt{T-s}} + T  
\right)  
\max\{|c_l|,1\}  
\max\{\|u_{x_l}^{(3)}\|_{\infty},\|u_{x_l}^{(4)}\|_{\infty}\}  
\frac{a_l^2}{N_l^2} \\  
&\lesssim  
T\cond(P)  
\max\{|c_l|,1\}  
\max\{\|u_{x_l}^{(3)}\|_{\infty},\|u_{x_l}^{(4)}\|_{\infty}\}  
\frac{a_l^2}{N_l^2}.
\end{align*}
Then let \eqref{thm para mol 2} hold true, we derive the required scale of $N_l$.
\end{proof}

\section{Quantum algorithm for linear parabolic PDEs}\label{section 4}
In this chapter, we solve semi-discrete ODEs using the optimal LCHS method. We first employ the Gaussian quadrature formula to handle the inhomogeneous term and obtain the standard form of LCHS. Next, we specify the oracles and construct block encodings of the coefficient matrices using elementary gates.
We then simulate the solution of the semi-discrete ODEs via LCHS framework.
Finally, we derive the total complexity of the overall PDE algorithm.

\subsection{Discretized implementation}
In this subsection, we approximate the solution of the semi-discrete ODE system~\eqref{parabolic approx odes} at a final time $T$.
Recall that the semi-discrete ODEs admit the analytic solution \eqref{para analytic solution}
\begin{align*}
v(T) =P^{-1}e^{-T\widetilde{A}}Pv_0+ \int_0^T P^{-1}e^{-(T - s)\widetilde{A}} P b(s) \, ds.
\end{align*}

We discretize the inhomogeneous integral in $v(T)$ by a composite Gauss-Legendre quadrature. 
Fix a step size $h_t>0$ and let $T/h_t$ be the number of time subintervals.
On each subinterval, we use a $Q_t$-point Gauss--Legendre rule on $[-1,1]$ with nodes $\{y_q\}_{q=0}^{Q_t-1}$
and weights $\{w_q\}_{q=0}^{Q_t-1}$, i.e.,
\(
\int_{-1}^{1} f(y)\,dy \approx \sum_{q=0}^{Q_t-1} w_q f(y_q).
\)
Define the quadrature nodes and coefficients
\begin{equation*}
s_{q_t,m_t} \;:=\; h_t\Big(m_t + \tfrac{1}{2} + \tfrac{1}{2}y_q\Big),  
\quad  
c_{q_t,m_t} \;:=\; \frac{h_t}{2}\, w_{q_t}\, \|b(s_{q_t,m_t})\|_2,  
\quad  
0\leq m_t\leq T/h_t-1,\; 0\leq q_t\leq Q_t-1.
\end{equation*}
Then 
\begin{align*}
    v(T)\approx v_Q(T):= P^{-1}e^{-T\widetilde{A}}Pv_0+ \sum_{m_t=0}^{T/h_t-1} \sum_{q_t=0}^{Q_t-1} c_{q_t,m_t} P^{-1}e^{-(T - s_{q_t,m_t})\widetilde{A}} P |b(s_{q_t,m_t})\rangle.
\end{align*} 
For convenience, we let $M_t=\dfrac{T}{h_t}Q_t$ and re-index $q_t,m_t$ by a single index $j_t$. Then we write \begin{align}\label{para vb}
    v_Q(T)=P^{-1}e^{-T\widetilde{A}}Pv_0+ \sum_{j_t=0}^{M_t-1} c_{j_t} P^{-1}e^{-(T - s_{j_t})\widetilde{A}} P |b(s_{j_t})\rangle.
\end{align} 
Then we derive the scale of $\sum c_{j_t}$ for subsequent LCU operation and the scales of $Q_t$ and $h_t$ in following two lemmas. The proofs can be found in Appendix \ref{appendix proof of gauss-len}.

\begin{lemma}\label{lemma para scaling of c}
    Let $c_{j_t},j_t=0,\cdots,M_t-1$ be the quadrature coefficients defined in \eqref{para vb}, then $c_{j_t}\geq 0$ and $\sum_{j_t=0}^{M_t-1}c_{j_t}=\mathcal{O}(\int_0^T \|b(s)\|_2 ds).$
\end{lemma}

\begin{lemma}\label{thm para len-gauss scale}
Let $v_{Q}(T)$ be defined by~\eqref{para vb}.
Fix an error $\varepsilon_{Q}\in(0,1)$. Assume $$\Xi := \sup\left\{ (\| b^{(p)} \|_{\star})^{\frac{1}{p+1}} \,\big|\, p \geq 0,\, t \in [0,T] \right\} < \infty.$$ If we choose
\[
Q_t \geq \frac{1}{\ln 4} \ln \frac{ T\,\mathrm{cond}(P) \Xi}{\varepsilon_{Q}}, \quad h_t \leq \frac{4 Q_t}{e ( \| \widetilde{A} \|_2 + \Xi)},
\]then $\|v(T)-v_Q(T)\|_{\star}\leq \varepsilon_{Q}$.
\end{lemma}

\begin{remark}
    For $b(t)= [f(x,t)]_{x}$, the mean norm of $b^{(p)}$ is approximately equal to integral mean value of $\partial_t^p f(x,t)$, i.e.,
    \[ \| b^{(p)} \|_{\star}^2 \approx \frac{1}{a_1\cdots a_d}\int_\Omega \left(\dfrac{\partial^p}{\partial t^p}f(x,t)\right)^2 dx. \]
    So $\Xi$ can be estimated by the integral of $\partial_t^p f(x,t)$.
\end{remark}

\subsection{Oracles and Block-encoding of coefficient matrix}

We now specify the quantum access model (oracles) and construct block-encodings about the transition matrix $P$
and transformed coefficient matrix $\widetilde{A}$ in \eqref{para trans matrix}.

\paragraph{State-preparation oracles.}
Recall that $N=\prod_{l=1}^d N_l$ is the total number of spatial nodes, where these $N_l$ are chosen as in
Theorem~\ref{thm para mol}. For quantum implementation, we assume $N_l$ is a power of two:
$N_l = 2^{n_l}$. For the initial state $u_0(x)$ and $v_0=[u_0(x)]_{x}\in\mathbb{R}^N$, we assume access to the oracle 
\begin{align}\label{para oracle Ov}
    O_v: |0\rangle \rightarrow  |v_0\rangle.
\end{align}
For the inhomogeneous term $f(x,t)$ and $b(t)=[f(x,t)]_{x}\in\mathbb{R}^N$, we assume access to the oracle 
\begin{align}\label{para oracle Ob}
    O_b: |j_t\rangle |0\rangle \rightarrow |j_t\rangle |b(s_{j_t})\rangle,
\end{align} where $s_{j_t}$ is the integral nodes in \eqref{para vb}.

For the nodes and coefficients of composite Gauss-Legendre quadrature formula in \eqref{para vb}, we assume access to the oracles
\begin{align}
    O_s&\in (1,1)BE\left(\diag(1-\frac{s_0}{T},1-\frac{s_1}{T},1-\frac{s_{M_t-1}}{T})\right),\label{para oracle Os}\\
    O_c&:\ |0\rangle \rightarrow \frac{1}{\sqrt{\|c\|_1}}\sum_{j_t=0}^{M_t-1}\sqrt{c_{j_t}}|j_t\rangle,\label{para oracle Oc}
\end{align} where $\|c\|_1=\sum_{j_t=0}^{M_t-1} c_{j_t}.$

\paragraph{Gate set and cost model.}
We count basic gates from the set
\[
  \mathcal{G}:=\{X,Y,Z,S,H,\; \mathrm{C}\!-\!X,\mathrm{C}\!-\!Y,\mathrm{C}\!-\!Z,\mathrm{C}\!-\!S,\mathrm{C}\!-\!H,\;\text{Toffoli}\}.
\]
Later, we will construct controlled gates many times in LCU (Lemma \ref{lemma be lcu}). We use the standard fact that if a unitary $U$ is implemented using $s$ gates in $\mathcal{G}$, then an $n$-fold
controlled version $\mathrm{C}^n\!-\!U$ can be implemented using at most $(2n+3)s$ gates in $\mathcal{G}$.

\paragraph{Block-encoding of $P$ and $P^{-1}$.} 
Recall in \eqref{para trans matrix Pl},
\[P_l=\diag(1, \theta_l, \cdots,\theta_l^{N_l-1})=\diag(1,\theta_l^{2^{n_l-1}})\otimes \cdots \otimes \diag(1,\theta_l^{2^1})\otimes  \diag(1,\theta_l^{2^0})=\bigotimes_{j=1}^{n_l}\diag(1,\theta_l^{2^{n_l-j}}).\]
Here, each $2\times 2$ diagonal factor admits an exact $(\max\{1,\theta_l^{2^j}\},1)$ block-encoding using $\mathcal{O}(1)$ gates in $\mathcal{G}$. By Lemma \ref{lemma be tensor}, we can obtain $(\max\{1,\theta_l\}^{N_l-1},n_l)$-block-encoding of $P_l$ using $\mathcal{O}(n_l)$ gates in $\mathcal{G}$. Recall in \eqref{para trans matrix} $P=\bigotimes_{l=1}^d P_l$, by Lemma \ref{lemma be tensor}, we can obtain the block-encoding of $P$
\begin{equation}\label{para be of P}
  U_P \in \bigl(\alpha_P,\, m_P\bigr)\mathrm{BE}(P),
  \quad
  \alpha_P=\prod_{l\in S_1\cup S_2}\max\{1,\theta_l\}^{\,N_l-1},
  \quad
  m_P=\sum_{l\in S_1\cup S_2} n_l,
\end{equation} where its cost is $\mathrm{Gate}(U_P)=\mathcal{O}\!\Big(\sum_{l\in S_1\cup S_2} n_l\Big)$.
Similarly, we can obtain
\begin{equation}\label{para be of Pinv}
  U_{P^{-1}} \in \bigl(\alpha_{P^{-1}},\, m_{P^{-1}}\bigr)\mathrm{BE}(P^{-1}),
  \quad
  \alpha_{P^{-1}}=\prod_{l\in S_1\cup S_2}\max\{1,\theta_l^{-1}\}^{\,N_l-1},
  \quad
  m_{P^{-1}}=\sum_{l\in S_1\cup S_2} n_l,
\end{equation} where its cost is $\mathrm{Gate}(U_{P^{-1}})=\mathcal{O}\!\Big(\sum_{l\in S_1\cup S_2} n_l\Big)$.

\paragraph{Block-encoding of $L$ and $H$.}
For \[ \sigma_{01} =\begin{pmatrix} 0 & 1 \\ 0 & 0 \end{pmatrix}= \dfrac{1}{2}(X + iY),\quad  \sigma_{10}= \begin{pmatrix} 0 & 0 \\ 1 & 0 \end{pmatrix} = \dfrac{1}{2}(X - iY) . \]
by Lemma \ref{lemma be lcu}, we can obtain the $(1,1)$-block-encoding of $\sigma_{10}$ and $\sigma_{01}$ using $5=\mathcal{O}(1)$ gates in $\mathcal{G}$. The LCU process is 
\[ \frac{1}{\sqrt{2}}\begin{pmatrix}    
1 & 1 \\    
1 & -1
\end{pmatrix}\otimes I_2 \cdot \begin{pmatrix}
    X & \\
    & iY
\end{pmatrix}\cdot \frac{1}{\sqrt{2}}\begin{pmatrix}    
1 & 1 \\    
1 & -1
\end{pmatrix}\otimes I_2=\begin{pmatrix}
    \frac{1}{2}(X+iY) & *\\
    * & *
\end{pmatrix}. \]

For $N_l=2^{n_l}$, define the local shift
\begin{align*}
    s_j^-:=I^{\otimes (n_l-j)} \otimes \sigma_{01} \otimes \sigma_{10}^{\otimes (j-1)}.
\end{align*}
By Lemma \ref{lemma be tensor}, we can obtain a $(1,j)$-block-encoding of $s_j^-$ using $5j=\mathcal{O}(j)$ gates in $\mathcal{G}$.
Then let
\begin{align*}
t^-_l := \begin{pmatrix} 0 & 1 & & \\ & \ddots & \ddots & \\ & & \ddots & 1 \\ & & & 0 \end{pmatrix}_{N_l} = \sum_{j=1}^{n_l} I^{\otimes (n_l-j)} \otimes \sigma_{01} \otimes \sigma_{10}^{\otimes (j-1)}=\sum_{j=1}^{n_l} s_j^-.   
\end{align*} 
By Lemma \ref{lemma be lcu}, we can obtain a 
\begin{align}\label{para be of t-}
(n_l,n_l+\lceil\log n_l\rceil)\text{-block-encoding of }t^-_l.
\end{align}
In this LCU process, the operator $\sum_j |j \rangle \langle j | \otimes s_j^-$ needs $\sum_j (2\lceil\log n_l\rceil+3)\cdot 5j=\mathcal{O}(n_l^2\log n_l)$ gates in $\mathcal{G}$. And the LCU coefficient operator only needs $\mathcal{O}(\log n_l)$ gates. So to block-encode $t^-_l$, we need $\mathcal{O}(n_l^2\log n_l)$ gates in $\mathcal{G}$.
Similarly, for $t^+_l:=(t^-_l)^T$, we also can get a $(n_l,n_l+\lceil\log n_l\rceil)$-block-encoding of $t^+_l$ at the same cost.

Based on different boundary conditions, we block-encode \[L_l:=\frac{1}{2}(\widetilde{A}_l+\widetilde{A}_l^\dag),\quad H_l:=\frac{1}{2i}(\widetilde{A}_l-\widetilde{A}_l^\dag)\] separately. Throughout, let $\eta_l:=\sqrt{1-\tfrac14 c_l^2 h_l^2}=\mathcal{O}(1)$.\\
    \textbf{(i) Dirichlet Boundary ($l\in S_1$):} Recall in \eqref{para diri tilde Al}, we have $H_l=O$ and \[ L_l = \frac{1}{h_l^2}(I+I - \eta_lt^-_l - \eta_lt^+_l). \] By Lemma \ref{lemma be lcu}, using the block-encodings of $t^-_l$ and $t^+_l$, we can get a $(\frac{2}{h_l^2}(1+\eta_l n_l),n_l+\lceil \log n_l\rceil+2)$-block-encoding of $L_l$, 
which needs $\mathcal{O}(n_l^2\log n_l)$ gates in $\mathcal{G}$.\\
\textbf{(ii) Neumann Boundary ($l\in S_2$):} Recall in \eqref{para neu tilde Al}, $H_l=O$ and \[ L_l = \frac{1}{h_l^2}\left(I+I - \eta_lt^-_l - \eta_lt^+_l -(1-\frac{1}{2}c_lh_l)\sigma_{00}^{\otimes n_l}-(1+\frac{1}{2}c_lh_l)\sigma_{11}^{\otimes n_l}\right), \]
where $\sigma_{00} =\begin{pmatrix} 1 & 0 \\ 0 & 0 \end{pmatrix}= \dfrac{1}{2}(I + Z),\sigma_{11} =\begin{pmatrix} 0 & 0 \\ 0 & 1 \end{pmatrix}= \dfrac{1}{2}(I-Z) $. Similar to the construction of $\sigma_{01},\sigma_{10}$, we can obtain $(1,n_l)$-block-encoding of $\sigma_{00}^{\otimes n_l}$ and $\sigma_{11}^{\otimes n_l}$ using $\mathcal{O}(n_l)$ gates in $\mathcal{G}$. Then by Lemma \ref{lemma be lcu}, we can obtain a $(\frac{2}{h_l^2}(2+\eta_l n_l),n_l+\lceil \log n_l\rceil+3)$-block-encoding of $L_l$ using $\mathcal{O}(n_l^2\log n_l)$ gates in $\mathcal{G}$. \\
\textbf{(iii) Periodic Boundary ($l\in S_3$):} Recall in \eqref{para neu tilde Al}, \[L_l= \frac{1}{h_l^2}(I+I - t^-_l - t^+_l - \sigma_{01}^{\otimes n} - \sigma_{10}^{\otimes n}),\quad H_l= \frac{1}{2h_l}(-t^-_l + t^+_l + \sigma_{01}^{\otimes n} - \sigma_{10}^{\otimes n}).\]
By Lemma \ref{lemma be lcu}, we can obtain a $(\frac{2}{h_l^2}(2+n_l),n_l+\lceil \log n_l\rceil+3)$-block-encoding of $L_l$ and a $(\frac{1}{h_l}(1+n_l),n_l+\lceil \log n_l\rceil+2)$-block-encoding of $H_l$ using $\mathcal{O}(n_l^2\log n_l)$ gates in $\mathcal{G}.$

Recall \eqref{para trans matrix}, we have the Kronecker-sum structure \begin{equation*}
L = \sum_{l=1}^d I_{N_1}\otimes\cdots\otimes I_{N_{l-1}}\otimes L_l\otimes I_{N_{l+1}}\otimes\cdots\otimes I_{N_d}.
\end{equation*} By Lemma \ref{lemma be lcu}, across the $d$ terms, we can obtain the following block-encoding of $L$ with its cost $\mathrm{Gate}(U_L)$:
\begin{align}\label{para be of L}
  U_L &\in (\alpha_L,\; m_L)\mathrm{BE}(L),
  \quad \alpha_L=\sum_{l\in S_1}\frac{2}{h_l^2}(1+\eta_l n_l)+
  \sum_{l\in S_2}\frac{2}{h_l^2}(2+\eta_l n_l)+\sum_{l\in S_3}\frac{2}{h_l^2}(2+n_l)=\mathcal{O}\left(\sum_{l=1}^d \frac{n_lN_l^2}{a_l^2} \right)
  ,\notag \\
   m_L&=\max_{1\leq l\leq d}\bigl(n_l+\lceil\log n_l\rceil\bigr)+3+\lceil\log d\rceil,\quad
  \mathrm{Gate}(U_L)=\mathcal{O}\!\Big(\log d\sum_{l=1}^{d} n_l^2\log n_l\Big).
  \end{align}

Similarly, since \begin{equation*}
H = \sum_{l\in S_3} I_{N_1}\otimes\cdots\otimes I_{N_{l-1}}\otimes H_l\otimes I_{N_{l+1}}\otimes\cdots\otimes I_{N_d},
\end{equation*} we can obtain the following block-encoding of $H$ with its cost $\mathrm{Gate}(U_H)$:
\begin{align}\label{para be of H}
  U_H &\in (\alpha_H,\; m_H)\mathrm{BE}(H),
  \quad \alpha_H
  = \sum_{l\in S_3}\frac{1}{h_l}(1+n_l)=\mathcal{O}\left(\sum_{l\in S_3} \frac{n_lN_l}{a_l} \right),\notag \\
  m_H&=\max_{l\in S_3}\bigl(n_l+\lceil\log n_l\rceil\bigr)+2+\lceil\log|S_3|\rceil,\quad
  \mathrm{Gate}(U_H) =\mathcal{O}(\log |S_3|\cdot \sum_{l\in S_3}n_l^2 \log n_l).
  \end{align}

\subsection{Implementation of LCHS}

In this subsection, we describe how to implement the approximation to $v_{Q}(T)$ by the optimal LCHS framework. 

Recall from~\eqref{para vb} that
\begin{equation*}
  v_{Q}(T)
  =
  P^{-1}e^{-T \widetilde{A}}P\,v_0
  +
  \sum_{j_t=0}^{M_t-1}
  c_{j_t}\,
  P^{-1}e^{-(T-s_{j_t})\widetilde{A}}P\,|b(s_{j_t})\rangle .
\end{equation*}
We construct the two terms separately.

\paragraph{Homogeneous term.}
We consider block-encoding the matrix
\[M_0:=P^{-1}e^{-T \widetilde{A}}P .\]
Through the similarity transformation $\widetilde{A}=PAP^{-1}$, we have
$L=\frac{\widetilde{A}+(\widetilde{A})^\dagger}{2}\succeq 0$.
So the LCHS framework in Section~2.2 can be applied directly to block-encode $e^{-T \widetilde{A}}$.
Firstly, we run LCHS for $e^{-T \widetilde{A}}$ with internal precision
\(\frac{\varepsilon_v}{\mathrm{cond}(P)}.\) Then by Lemma \ref{lemma be times}, we construct a block-encoding of $M_0$ with precision $\varepsilon_v$.

\begin{lemma}\label{para be of homo}
Let $\varepsilon_v\in(0,1)$.
We can construct a block-encoding
\[
U_{v}\in (\alpha_v,m_v,\varepsilon_v)\mathrm{BE}\!\left(M_0\right),
\]
where
\[
\alpha_v=\mathcal{O}(\mathrm{cond}(P)),
\quad
m_v=\mathcal{O}\!\left(
\log\!
\bigl(T\|L\|_2+\log\frac{\mathrm{cond}(P)}{\varepsilon_v}\bigr)+\max_{1\leq l\leq d}\bigl(n_l+\lceil\log n_l\rceil\bigr)+\log d+\sum_{l\in S_1\cup S_2} n_l
\right).  \] 
And the gate complexity is
\[
\mathrm{Gate}(U_v)=\mathcal{O}\!\left(
\log\!\left(
T\|L\|_2+\log\frac{\mathrm{cond}(P)}{\varepsilon_v}
\right)
\left(
Q_v+\log^{5/2}\frac{\mathrm{cond}(P)}{\varepsilon_v}
\right)+\log d \sum_{l=1}^d n_l^2\log n_l \cdot Q_v
\right),
\]where
\[
Q_v=\mathcal{O}\left( \log\frac{\mathrm{cond}(P)}{\varepsilon_v}T\cdot \sum_{l=1}^d \frac{n_lN_l^2}{a_l^2} \right).
\]
\end{lemma}

\begin{proof}
We first block-encode $e^{-T \widetilde{A}}$.
Applying Theorem \ref{thm lchs} with target precision $\varepsilon_v/\mathrm{cond}(P)$,
we obtain
\[
U_{e,T}\in(\alpha_1,m_1,\varepsilon_v/\mathrm{cond}(P))\mathrm{BE}(e^{-T \widetilde{A}}),
\]
where we choose
\[
R_1=\mathcal{O}\!\left(\log\frac{\mathrm{cond}(P)}{\varepsilon_v}\right),
\quad
\gamma_1=\mathcal{O}\!\left(\sqrt{\log\frac{\mathrm{cond}(P)}{\varepsilon_v}}\right),
\quad
\frac{1}{h_1}
=
\mathcal{O}\!\left(T\|L\|_2+\log\frac{\mathrm{cond}(P)}{\varepsilon_v}\right),
\]
and
\[
\alpha_1
=
\frac{h_1}{\sqrt{2\pi}}
\sum_{j=-R_1/h_1}^{R_1/h_1}
\bigl|\widehat f_2(h_1 j;\gamma_1,c)\bigr|
=\mathcal{O}(1).
\]
By Theorem \ref{thm qsvt hs}, the Hamiltonian simulation subroutine for
\[
e^{-iT(R_1L+H)}
\]
has query complexity
\[
Q_v=
\mathcal{O}\!\left(
(R_1\alpha_L+\alpha_H)T+\log\frac{\mathrm{cond}(P)}{\varepsilon_v}
\right)=\mathcal{O}\left( \log\frac{\mathrm{cond}(P)}{\varepsilon_v}T\cdot \sum_{l=1}^d \frac{n_lN_l^2}{a_l^2} \right)
\]and ancilla qubits
\[\mathcal{O}(m_L)=\mathcal{O}\left(\max_{1\leq l\leq d}\bigl(n_l+\lceil\log n_l\rceil\bigr)+\log d\right),\]
where $\alpha_L$ and $\alpha_H$ are given in~\eqref{para be of L} and~\eqref{para be of H}.
By Theorem \ref{thm lchs}, the overall gate complexity of the LCHS implementation for $e^{-T \widetilde{A}}$ is
\begin{align*}
&\mathcal{O}\left( (\log(T\|L\|_2 + \log \frac{\mathrm{cond}(P)}{\varepsilon_v}) \log^{5/2} \frac{\mathrm{cond}(P)}{\varepsilon_v} \right)+\mathcal{O}\big(( \mathrm{Gate}(U_L)+ \mathrm{Gate}(U_H))\cdot Q_v\big)+\mathcal{O}\!\left(
\log\!\left(
T\|L\|_2+\log\frac{\mathrm{cond}(P)}{\varepsilon_v}
\right)\cdot Q_v\right)\\
&=\mathcal{O}\!\left(
\log\!\left(
T\|L\|_2+\log\frac{\mathrm{cond}(P)}{\varepsilon_v}
\right)
\left(
Q_v+\log^{5/2}\frac{\mathrm{cond}(P)}{\varepsilon_v}
\right)+\log d \sum_{l=1}^d n_l^2\log n_l \cdot Q_v
\right).
\end{align*}
The number of ancilla qubits is \[m_1=\mathcal{O}\left(\log \frac{R_1}{h_1}\right)+\mathcal{O}(m_L)=\mathcal{O}\left(\log\!
\bigl(T\|L\|_2+\log\frac{\mathrm{cond}(P)}{\varepsilon_v}\bigr)+\max_{1\leq l\leq d}\bigl(n_l+\lceil\log n_l\rceil\bigr)+\log d\right).\]

Next, we use the block-encodings of $P$ and $P^{-1}$ in \eqref{para be of P} and \eqref{para be of Pinv} to construct block-encoding of $M_0$. Since
\[
\alpha_P\alpha_{P^{-1}}
=
\prod_{l\in S_1\cup S_2}
\max\{1,\theta_l\}^{N_l-1}
\max\{1,\theta_l^{-1}\}^{N_l-1}
=\prod_{l\in S_1\cup S_2}\mathrm{cond}(P_l)
=\mathrm{cond}(P),
\]
two applications of Lemma \ref{lemma be times} yield a block-encoding of
\(M_0=P^{-1}e^{-T \widetilde{A}}P\)
with normalization factor
\[
\alpha_v=\alpha_{P^{-1}}\alpha_1\alpha_P=\mathrm{cond}(P)\alpha_1
\]
and ancilla number
\[
m_v=m_1+2m_P.
\]
Because $U_{P^{-1}}$ and $U_P$ are exact, the only approximation error comes from
$U_{e,T}$ and is amplified by the factor $\alpha_{P^{-1}}\alpha_P=\mathrm{cond}(P)$.
Hence the final error is \(\varepsilon_v.\)
The cost has the same asymptotic scaling as that of constructing $U_{e,T}$, up to the lower-order
overhead of implementing $U_P$ and $U_{P^{-1}}$.
\end{proof}

With the oracle $O_v$ from~\eqref{para oracle Ov}, the circuit for homogeneous term is
\[
|0\rangle|0\rangle
\xrightarrow{\,I\otimes O_v\,}
|0\rangle|v_0\rangle
\xrightarrow{\,U_v\,}
|0\rangle \frac{1}{\alpha_v}\widetilde M_0 |v_0\rangle + |\perp_v\rangle,
\]
where $\widetilde M_0$ denotes the matrix induced by the top-left block of $U_v$, and
\begin{align}\label{para error of hom}
\|\widetilde M_0-M_0\|_2\le \varepsilon_v.
\end{align} As a significant advantage of LCHS, here we query the initial state oracle $O_v$ only one time.

\paragraph{Inhomogeneous term.}
We now implement all matrices
\[M_{j_t}:=P^{-1}e^{-(T-s_{j_t})\widetilde{A}}P\]
simultaneously by a block-diagonal construction.
Define
\begin{equation}\label{para SEL tildeA}
\mathrm{SEL}_{\widetilde{A}}
:=
\mathrm{diag}\!\bigl((T-s_0)\widetilde{A},\dots,(T-s_{M_t-1})\widetilde{A}\bigr)
=T(\mathrm{diag}\!(
1-\frac{s_0}{T},\,\dots,\,1-\frac{s_{M_t-1}}{T}
)\otimes I_N)(I_{M_t}\otimes \widetilde{A}).
\end{equation}
Accordingly, its Hermitian and skew-Hermitian parts are
\begin{align*}
\mathrm{SEL}_{L}&:=\frac{\mathrm{SEL}_{\widetilde{A}}+\mathrm{SEL}_{\widetilde{A}}^\dagger}{2}=T(\mathrm{diag}\!(
1-\frac{s_0}{T},\,\dots,\,1-\frac{s_{M_t-1}}{T}
)\otimes I_N)(I_{M_t}\otimes L), \\
\quad
\mathrm{SEL}_{H}&:=\frac{\mathrm{SEL}_{\widetilde{A}}-\mathrm{SEL}_{\widetilde{A}}^\dagger}{2i}=T(\mathrm{diag}\!(
1-\frac{s_0}{T},\,\dots,\,1-\frac{s_{M_t-1}}{T}
)\otimes I_N)(I_{M_t}\otimes H).
\end{align*}

We consider block-encoding the matrix
\[\text{SEL}_{M}:=I\otimes P^{-1} e^{-\text{SEL}_{\widetilde{A}}} I\otimes P.\]
As in the homogeneous case, firstly we run LCHS for $e^{-\text{SEL}_{\widetilde{A}}}$ with internal precision
\(\frac{\varepsilon_b}{\mathrm{cond}(P)}.\) Then by Lemma \ref{lemma be times}, we construct block-encoding of $\text{SEL}_{M}$ with precision $\varepsilon_b$.

\begin{lemma}\label{para be of nonh}
Let $\varepsilon_b\in(0,1)$. We can construct a block-encoding
\[
U_b\in(\alpha_b,m_b,\varepsilon_b)\mathrm{BE}\!\left(\text{SEL}_{M}\right),
\]
where
\[
\alpha_b=\mathcal{O}(\mathrm{cond}(P)),
\quad
m_b=\mathcal{O}\!\left(
\log\!
\bigl(T\|L\|_2+\log\frac{\mathrm{cond}(P)}{\varepsilon_b}\bigr)+\max_{1\leq l\leq d}\bigl(n_l+\lceil\log n_l\rceil\bigr)+\log d+\sum_{l\in S_1\cup S_2} n_l
\right),
\]
with the gate complexity
\[
\mathrm{Gate}(U_b)=\mathcal{O}\!\left(
\log\!\left(
T\|L\|_2+\log\frac{\mathrm{cond}(P)}{\varepsilon_b}
\right)
\left(
Q_b+\log^{5/2}\frac{\mathrm{cond}(P)}{\varepsilon_b}
\right)+\log d \sum_{l=1}^d n_l^2\log n_l \cdot Q_b
\right),
\]where \[
Q_b=\mathcal{O}\left( \log\frac{\mathrm{cond}(P)}{\varepsilon_b}T\cdot \sum_{l=1}^d \frac{n_lN_l^2}{a_l^2} \right),
\] and query complexity for $O_s$ (in \eqref{para oracle Os}) \(Q_b.\)
\end{lemma}

\begin{proof}
By~\eqref{para oracle Os}, the oracle $O_s$ is an exact $(1,1)$-block-encoding of $\mathrm{diag}\!(
1-\frac{s_0}{T},\,\dots,\,1-\frac{s_{M_t-1}}{T}
)$.
Combining $O_s$ with the exact block-encodings of $L$ and $H$ from~\eqref{para be of L}
and~\eqref{para be of H}, Lemma \ref{lemma be tensor} gives
\[
U_{\mathrm{SEL}_L}\in (T\alpha_L,m_L+1)\mathrm{BE}(\mathrm{SEL}_L),
\quad
U_{\mathrm{SEL}_H}\in (T\alpha_H,m_H+1)\mathrm{BE}(\mathrm{SEL}_H).
\]
Applying Lemma \ref{lemma be lcu} to the linear combination
\[
R_2\,\mathrm{SEL}_L+\mathrm{SEL}_H
\]
yields an exact block-encoding with normalization factor $T(R_2\alpha_L+\alpha_H)$. And  we have
\[
\|\mathrm{SEL}_{L}\|_2=\max_{j_t}\{ \|(T-s_{j_t})L\|_2 \} \leq T\|L\|_2.
\]

Now apply Theorem \ref{thm lchs} with target precision $\varepsilon_b/\mathrm{cond}(P)$.
We obtain
\[
U_{\mathrm{SE}}
\in (\alpha_2,m_2,\varepsilon_b/\mathrm{cond}(P))\mathrm{BE}(e^{-\mathrm{SEL}_{\widetilde{A}}}),
\]
where
\[
R_2=\mathcal{O}\!\left(\log\frac{\mathrm{cond}(P)}{\varepsilon_b}\right),
\quad
\gamma_2=\mathcal{O}\!\left(\sqrt{\log\frac{\mathrm{cond}(P)}{\varepsilon_b}}\right),
\quad
\frac{1}{h_2}
=
\mathcal{O}\!\left(\|\mathrm{SEL}_L\|_2+\log\frac{\mathrm{cond}(P)}{\varepsilon_b}\right)
=
\mathcal{O}\!\left(T\|L\|_2+\log\frac{\mathrm{cond}(P)}{\varepsilon_b}\right),
\]and
\[
\alpha_2
=
\frac{h_2}{\sqrt{2\pi}}
\sum_{j=-R_2/h_2}^{R_2/h_2}
\bigl|\widehat f_2(h_2 j;\gamma_2,c)\bigr|
=\mathcal{O}(1),
\]
\[
m_2=
\mathcal{O}\!\left(
\log\!\Bigl(\log\frac{\mathrm{cond}(P)}{\varepsilon_b}\cdot
\bigl(T\|L\|_2+\log\frac{\mathrm{cond}(P)}{\varepsilon_b}\bigr)\Bigr)
\right)=\mathcal{O}\!\left(
\log\!
\bigl(T\|L\|_2+\log\frac{\mathrm{cond}(P)}{\varepsilon_b}\bigr)
\right).
\]
The Hamiltonian simulation subroutine now targets
\[
e^{-i(R_2\mathrm{SEL}_L+\mathrm{SEL}_H)},
\]
whose query complexity (by Theorem \ref{thm qsvt hs}) is
\[
Q_b=
\mathcal{O}\!\left(
T(R_2\alpha_L+\alpha_H)+\log\frac{\mathrm{cond}(P)}{\varepsilon_b}
\right)=\mathcal{O}\left( \log\frac{\mathrm{cond}(P)}{\varepsilon_b}T\cdot \sum_{l=1}^d \frac{n_lN_l^2}{a_l^2} \right),
\]
and number of ancilla qubits is
\[\mathcal{O}(m_L)=\mathcal{O}\left(\max_{1\leq l\leq d}\bigl(n_l+\lceil\log n_l\rceil\bigr)+\log d\right).\]
Hence the overall gate complexity of the LCHS implementation for $e^{-\text{SEL}_{\widetilde{A}}}$ is
\begin{align*}
&\mathcal{O}\left( (\log(T\|L\|_2 + \log \frac{\mathrm{cond}(P)}{\varepsilon_b}) \log^{5/2} \frac{\mathrm{cond}(P)}{\varepsilon_b} \right)+\mathcal{O}\big(( \mathrm{Gate}(U_L)+ \mathrm{Gate}(U_H))\cdot Q_b\big)+\mathcal{O}\!\left(
\log\!\left(
T\|L\|_2+\log\frac{\mathrm{cond}(P)}{\varepsilon_b}
\right)\cdot Q_b\right)\\
&=\mathcal{O}\!\left(
\log\!\left(
T\|L\|_2+\log\frac{\mathrm{cond}(P)}{\varepsilon_b}
\right)
\left(
Q_b+\log^{5/2}\frac{\mathrm{cond}(P)}{\varepsilon_b}
\right)+\log d \sum_{l=1}^d n_l^2\log n_l \cdot Q_b
\right).
\end{align*}
The number of ancilla qubits is \[m_2=\mathcal{O}\left(\log \frac{R_2}{h_2}\right)+\mathcal{O}(m_L)=\mathcal{O}\left(\log\!
\bigl(T\|L\|_2+\log\frac{\mathrm{cond}(P)}{\varepsilon_b}\bigr)+\max_{1\leq l\leq d}\bigl(n_l+\lceil\log n_l\rceil\bigr)+\log d\right).\]
And the query complexity to $O_s$ is $Q_b$.

Finally, combining $U_{\mathrm{SE}}$ with the block-encodings of $P^{-1}$ and $P$ in \eqref{para be of P} and \eqref{para be of Pinv}, Lemma \ref{lemma be times} gives a block-encoding of
\(\text{SEL}_{M}=I\otimes P^{-1}e^{-\mathrm{SEL}_{\widetilde{A}}}I\otimes P\)
with normalization factor
\[
\alpha_b=\alpha_{P^{-1}}\alpha_2\alpha_P=\mathrm{cond}(P)\alpha_2
\]
and ancilla number
\[
m_b=m_2+2m_P.
\]
Again, since the block-encodings of $P$ and $P^{-1}$ are exact, the final error is
\(\mathrm{cond}(P)\frac{\varepsilon_b}{\mathrm{cond}(P)}=\varepsilon_b.\) And the cost has the same asymptotic scaling as that of constructing $U_{SE}$.
\end{proof}

Using the oracles $O_b$ and $O_c$ from~\eqref{para oracle Ob} and \eqref{para oracle Oc}, the circuit for inhomogeneous term is
\[
|0\rangle|0\rangle|0\rangle
\xrightarrow{\,I\otimes O_c\otimes I\,}
|0\rangle \frac{1}{\sqrt{\|c\|_1}}\sum_{j_t=0}^{M_t-1}\sqrt{c_{j_t}}\,|j_t\rangle|0\rangle
\xrightarrow{\,I\otimes O_b\,}
|0\rangle \frac{1}{\sqrt{\|c\|_1}}\sum_{j_t=0}^{M_t-1}\sqrt{c_{j_t}}\,|j_t\rangle|b(s_{j_t})\rangle
\]
\[
\xrightarrow{\,U_b\,}
|0\rangle \frac{1}{\sqrt{\|c\|_1}}\sum_{j_t=0}^{M_t-1}\sqrt{c_{j_t}}\,|j_t\rangle
\frac{1}{\alpha_b}\widetilde M_{j_t}|b(s_{j_t})\rangle
+|\perp_b\rangle
\xrightarrow{\,I\otimes O_c^\dagger\otimes I\,}
|0\rangle|0\rangle
\frac{1}{\alpha_b\|c\|_1}
\sum_{j_t=0}^{M_t-1}
c_{j_t}\widetilde M_{j_t}|b(s_{j_t})\rangle
+|\perp_b'\rangle,
\]
where for each $j_t$,
\begin{align}\label{para error of nonin}
\|\widetilde M_{j_t}-M_{j_t}\|_2\le \varepsilon_b.
\end{align}
Again as a significant advantage of LCHS, here we only query $O_b$ one time and $O_c$ two times.

\paragraph{Combining the two terms.}
We now combine the two terms by one additional LCU step.
Prepare one control qubit by
\[
V_{\mathrm{LCU}}|0\rangle
=
\frac{1}{\sqrt{\alpha_v\|v_0\|_2+\alpha_b\|c\|_1}}
\left(
\sqrt{\alpha_v\|v_0\|_2}\,|0\rangle+\sqrt{\alpha_b\|c\|_1}\,|1\rangle
\right).
\]
Conditioned on the control value, apply the homogeneous circuit and the inhomogeneous circuit, respectively.
Then the whole circuit implements
\[
|0\rangle|0\rangle|0\rangle|0\rangle
\longmapsto
|0\rangle|0\rangle|0\rangle\,
\frac{1}{\alpha_v\|v_0\|_2+\alpha_b\|c\|_1}\widetilde{v}(T)
+
|\perp\rangle,
\] where
\begin{equation}\label{para vapp}
\widetilde{v}(T)
:=
\|v_0\|_2\,\widetilde M_0|v_0\rangle
+
\sum_{j_t=0}^{M_t-1}
c_{j_t}\,\widetilde M_{j_t}|b(s_{j_t})\rangle .
\end{equation}
Standard amplitude amplification can then be used to prepare the normalized state
proportional to $\widetilde{v}(T)$. By $\|c\|_1=\mathcal{O}(\int_{0}^T\|b(s)\|_2ds)$ in Lemma \ref{lemma para scaling of c}, here the amplitude of $|\widetilde{v}(T)\rangle$ is \begin{align}\label{para complexity of aa}
    \frac{\|\widetilde{v}(T)\|_2}{\alpha_v\|v_0\|_2+\alpha_b\|c\|_1}=\mathcal{O}\left( \frac{\|\widetilde{v}(T)\|_{\star}}{{\alpha_v \|v_0\|_{\star}}+\alpha_b \int_0^T\|b(s)\|_{\star}ds} \right).
\end{align}

Note that 
\[
v_{Q}(T)-\widetilde{v}(T)
=
\|v_0\|_2\,(M_0-\widetilde M_0)|v_0\rangle
+
\sum_{j_t=0}^{M_t-1}
c_{j_t}(M_{j_t}-\widetilde M_{j_t})|b(s_{j_t})\rangle .
\]
By \eqref{para error of hom} and \eqref{para error of nonin}, we have
\begin{align}\label{para error of lchs}
\|\widetilde{v}(T)-v_{Q}(T)\|_2
\le
\varepsilon_v\|v_0\|_2+\varepsilon_b\|c\|_1 .
\end{align}

\subsection{Total complexity}
Finally, let's recall the process of solving our parabolic PDE \eqref{parabolic pde} and compute the total complexity in Theorem \ref{para final thm}. \\
\textbf{Step 1:} For the spatial discretization error $\| [ u(x, T) ]_{x} - v(T) \|_{\star} \leq \varepsilon_N$, according to Theorem \ref{thm para mol} we choose the grid sizes $N_l$ as
\begin{align*}  
N_l &=\mathcal{O} \left( \sqrt{ \frac{Td|c_l|}{\varepsilon_N} } a_l \exp(\sum_{l' \in S_1,S_2} \frac{|c_{l'}|a_{l'}}{4})\left( \max \{ \| u_{x_l}^{(3)} \|_{\infty}, \| u_{x_l}^{(4)} \|_{\infty} \} \right)^{\frac{1}{2}}\right).
\end{align*}

\noindent\textbf{Step 2:} For the quadrature error $\|v(T)-v_Q(T)\|_{\star}\leq \varepsilon_{Q}$, according to Lemma \ref{thm para len-gauss scale} we choose $Q_t,h_t,M_t$ as
\[M_t=\mathcal{O}(T\|\widetilde{A}\|_2).\]

\noindent\textbf{Step 3:} For the LCHS approximation error $\varepsilon_v,\varepsilon_b$, we use the optimal LCHS to construct two block-encodings for homogeneous term and inhomogeneous term in Lemma \ref{para be of homo} and \ref{para be of nonh}. Then \(\|\widetilde{v}(T)-v_{Q}(T)\|_2
\leq
\varepsilon_v\|v_0\|_2+\varepsilon_b\|c\|_1\) as in \eqref{para error of lchs}.


By Lemma \ref{lemma unit error}, the overall error is
\begin{align*}
\| |[ u(x, T) ]_{x}\rangle-|\widetilde{v}(T)\rangle \|_2&\leq \frac{2}{\| [ u(x, T) ]_{x}\|_{\star}} \| [ u(x, T) ]_{x} -\widetilde{v}(T) \|_{\star}\\
&\leq \frac{2}{\| [ u(x, T) ]_{x}\|_{\star}} \left(\varepsilon_N+\varepsilon_Q+\frac{\varepsilon_v\|v_0\|_2+\varepsilon_b\|c\|_1}{\sqrt{N}}\right)\\
&=\frac{2}{\| [ u(x, T) ]_{x}\|_{\star}} \left(\varepsilon_N+\varepsilon_Q+\varepsilon_v\|v_0\|_{\star}  +\frac{\|c\|_1}{\sqrt{N}}\varepsilon_b\right).
\end{align*}So to let $\| |[ u(x, T) ]_{x}\rangle-|\widetilde{v}(T)\rangle \|_2\leq \varepsilon $, we choose
\begin{align}
    \varepsilon_N&=\varepsilon_Q=\frac{1}{8}\| [ u(x, T) ]_{x}\|_{\star}\varepsilon,\label{para epsilonN}\\
    \varepsilon_v&=\frac{1}{8}\frac{\| [ u(x, T) ]_{x}\|_{\star}}{ \| v_0 \|_{\star} }\varepsilon,\label{para epsilonv}\\
    \varepsilon_b&=\frac{1}{8}\frac{\| [ u(x, T)]_{x}\|_{\star}\sqrt{N}}{ \|c \|_1 }\varepsilon=\mathcal{O}\left(\frac{\| [ u(x, T)]_{x}\|_{\star}}{\int_0^T \|b(s)\|_{\star} ds}\varepsilon\right)\label{para epsilonb},
\end{align}
where the scaling of $\varepsilon_b$ is by Lemma \ref{lemma para scaling of c}.

\begin{theorem}[Total complexity of the parabolic PDE solver]\label{para final thm}
Let $u(x,T)$ be the solution of the parabolic PDE \eqref{parabolic pde} with boundary conditions (\ref{boundary1}-\ref{boundary3}) and assume that the solution $u$ is sufficiently smooth. Let $$\Gamma:=\max_{l=1,\cdots,d}\left(\max \{ \| u_{x_l}^{(3)} \|_{\infty}, \| u_{x_l}^{(4)} \|_{\infty} \}\right), \quad q:=\dfrac{\| [u_0(x)]_x \|_{\star} + \int_0^T \| [f(x,s)]_x \|_{\star} \, ds}{\| [u(x,T)]_x \|_{\star}}.$$ For the error $\varepsilon\in (0,1)$, to get a solution $\widetilde{v}(T)$ such that $\| |[ u(x, T) ]_{x}\rangle-|\widetilde{v}(T)\rangle \|_2\leq \varepsilon $ with $\Omega(1)$ probability and a flag indicating success, we need \\
 (1)  \begin{align*}
\widetilde{\mathcal{O}}\Bigg( q \cdot \dfrac{T^2d^2\Gamma}{\| [u(x,T)]_x \|_{\star}\varepsilon} \Bigg)
 \end{align*} basic gates in $\mathcal{G}$;\\
 (2) \begin{align*}     
\widetilde{\mathcal{O}}\Bigg( q \cdot \dfrac{T^2d\Gamma}{\| [u(x,T)]_x \|_{\star}\varepsilon} \Bigg) 
\end{align*} queries for oracle $O_s$;\\
 (3)   $$\mathcal{O}\left( q \right)$$ queries for oracles $O_v,O_b,O_c$;\\
 (4)   \begin{align*}
&\widetilde{\mathcal{O}}\left( d \log \frac{T\Gamma}{\| [ u(x, T) ]_{x}\|_{\star}\varepsilon}  +\log\log \frac{ \| [u_0(x)]_x \|_{\star} }{\| [ u(x, T) ]_{x}\|_{\star}}+\log\log \frac{ \int_0^T \|[f(x,s)]_x\|_{\star} ds }{\| [ u(x, T) ]_{x}\|_{\star}} \right)     
    \end{align*} qubits.

\end{theorem}

\begin{remark}
    Here we only consider the complexity with respect to the dimension $d$, the error $\varepsilon$, the final time $T$, the initial value $u_0$, the inhomogeneous term $b$ and the solution itself. The gate complexity is quadratic in $T$, which arises from the method of lines and the LCHS approach; quadratic in $d$, which stems from the parallelism advantage of quantum computing; and first-order in $\varepsilon$, which is because the method of lines only achieves polynomial accuracy. Moreover, owing to the advantages of LCHS, the complexity of the oracles associated with the initial value and the inhomogeneous term is optimal.
\end{remark}

\begin{proof}
    (1) Gate complexity and queries to $O_s$: Let $\mathrm{cost}_{AA}$ be the complexity of amplitude amplification. The total gate complexity is 
$$\mathrm{cost}_{AA}(\mathrm{Gate}(U_v)+\mathrm{Gate}(U_b)),$$ where $\mathrm{Gate}(U_v)$ and $\mathrm{Gate}(U_b)$ are the gate complexity of $U_v$ and $U_b$ in Lemma \ref{para be of homo} and Lemma \ref{para be of nonh}.

Since in the block-encoding of $L$ \eqref{para be of L}, $T\|L\|_2\leq T\alpha_L=\mathcal{O}\left(T\cdot \sum_{l=1}^d \frac{n_lN_l^2}{a_l^2} \right)\leq Q_v$, we have
\[
\log\!\left(
T\|L\|_2+\log\frac{\mathrm{cond}(P)}{\varepsilon_v}
\right)
\left(
Q_v+\log^{5/2}\frac{\mathrm{cond}(P)}{\varepsilon_v}
\right)=\widetilde{\mathcal{O}}\left(
Q_v+\log^{5/2}\frac{\mathrm{cond}(P)}{\varepsilon_v}
\right).
\]
So \begin{align*}
\mathrm{Gate}(U_v)&=\widetilde{\mathcal{O}}\left(
Q_v+\log^{5/2}\frac{\mathrm{cond}(P)}{\varepsilon_v}
\right)+\mathcal{O}\left(\log d \sum_{l=1}^d n_l^2\log n_l \cdot Q_v
\right)\\
&=\widetilde{\mathcal{O}}\left(\log d \sum_{l=1}^d n_l^2\log n_l \cdot Q_v +\log^{5/2}\frac{\mathrm{cond}(P)}{\| [ u(x, T)]_{x}\|_{\star}\varepsilon_v}
\right)\\
&\lesssim \widetilde{\mathcal{O}}\left( \log d \cdot T \sum_{l=1}^d\frac{N_l^2}{a_l^2}\sum_{l=1}^d n_l^2 \log^{5/2}\frac{\mathrm{cond}(P)}{\varepsilon_v} \right)\\
&=\widetilde{\mathcal{O}}\left( \log d \cdot T \sum_{l=1}^d\frac{N_l^2}{a_l^2}\sum_{l=1}^d n_l^2 \log^{5/2}\frac{\mathrm{cond}(P)\|v_0\|_{\star}}{ \| [ u(x, T)]_{x}\|_{\star} \varepsilon}\right).\end{align*}
By \eqref{para complexity of aa} and $\alpha_v,\alpha_b=\mathcal{O}(\mathrm{cond}(P))$, the cost of amplitude amplification is
\[
\mathrm{cost}_{AA} = \frac{\alpha_v \|v_0\|_2 + \alpha_b \|c\|_1}{\|\widetilde{v}(T)\|_2} = \mathcal{O}\left( \frac{\cond(P) \left( \|v_0\|_{\star} + \int_0^T \|b(s)\|_{\star} \, ds \right)}{\| [u(x,T)]_{x} \|_{\star}} \right).
\]
So we have $\log \mathrm{cost}_{AA}\gtrsim \log \dfrac{{\mathrm{cond}(P)} \|v_0\|_{\star} }{\| [ u(x, T)]_{x}\|_{\star}}$. And since $(\log N_l)^2\gtrsim \log \dfrac{1}{\varepsilon}$, we have 
$$ \mathrm{cost}_{AA}\cdot \mathrm{Gate}(U_v)=\widetilde{\mathcal{O}}\left(\mathrm{cost}_{AA}  \cdot \log d \cdot T\sum_{l=1}^d\frac{N_l^2}{a_l^2} \sum_{l=1}^d n_l^2\right).$$   
Since $N_l=2^{n_l}$ and $\sum_{l=1}^dn_l^2\leq d\max_l n_l^2$, $$\mathrm{cost}_{AA}\cdot \mathrm{Gate}(U_v)=\widetilde{\mathcal{O}}\left(\mathrm{cost}_{AA}  \cdot d \cdot T\sum_{l=1}^d\frac{N_l^2}{a_l^2} \right).$$

Similarly, we have the complexity for inhomogeneous term satisfies
\[
\mathrm{cost}_{AA}\cdot \mathrm{Gate}(U_b)=\widetilde{\mathcal{O}}\left(\mathrm{cost}_{AA}  \cdot d \cdot T\sum_{l=1}^d\frac{N_l^2}{a_l^2} \right).
\]
By the scale of $N_l$ in Theorem \ref{thm para mol} and $\mathrm{cond}(P)=\mathcal{O}\left(\exp\left( \sum_{l \in S_1,S_2} \frac{|c_l| a_l}{2} \right)\right)$, the total complexity is
\begin{align*}
&\mathrm{cost}_{AA}(\mathrm{Gate}(U_v)+\mathrm{Gate}(U_b)) = \widetilde{\mathcal{O}}\left( \frac{\cond(P) \left( \|v_0\|_{\star} + \int_0^T \|b(s)\|_{\star} \, ds \right)}{\| [u(x,T)]_x \|_{\star}} dT \sum_{l=1}^d \frac{N_l^2}{a_l^2} \right)\\
&= \widetilde{\mathcal{O}}\Bigg( \frac{\| [u_0(x)]_x \|_{\star} + \int_0^T \| [f(x,s)]_x \|_{\star} \, ds}{\| [u(x,T)]_x \|_{\star}} \cdot dT \cdot\exp\left( \sum_{l \in S_1\cup S_2} |c_l| a_l \right) \frac{Td}{\| [u(x,T)]_x \|_{\star} \varepsilon} \Gamma \sum_{l} |c_l|\Bigg).
\end{align*}

(2) Query complexity to $O_s$: By Lemma \ref{para be of nonh}, a single implementation of LCHS requires
\[
Q_b = \mathcal{O}\left( \log \frac{\cond(P)}{\varepsilon_b} \cdot T \sum_{l=1}^d \frac{n_l N_l^2}{a_l^2} \right),
\]queries to $O_s$.
Then the query complexity is
$$\mathrm{cost}_{AA}\cdot Q_b=\widetilde{\mathcal{O}}\left(  \mathrm{cost}_{AA}\cdot T\sum_{l=1}^d\frac{N_l^2}{a_l^2}\right).$$
Similar to case (1) (The complexity only lacks a single term $d$.), we can get query complexity to $O_s$ is
\begin{align*} \widetilde{\mathcal{O}}\Bigg( \frac{\| [u_0(x)]_x \|_{\star} + \int_0^T \| [f(x,s)]_x \|_{\star} \, ds}{\| [u(x,T)]_x \|_{\star}} \cdot T \cdot \exp\left( \sum_{l \in S_1\cup S_2} |c_l| a_l \right) \frac{Td}{\| [u(x,T)]_x \|_{\star} \varepsilon} \Gamma \sum_{l} |c_l|\Bigg).\end{align*}

(3) Queries to $O_v$, $O_b$, $O_c$: By Lemma \ref{para be of homo} and Lemma \ref{para be of nonh}, a single implementation of LCHS requires $\mathcal{O}(1)$ queries of $O_v$, $O_b$, $O_c$. So the complexity is
\[
\mathcal{O}(1) \cdot \mathrm{cost}_{AA} = \mathcal{O}\left( \exp\left( \sum_{l \in S_1\cup S_2} \frac{1}{2} |c_l| a_l \right) \cdot \frac{\| [u_0(x)]_x \|_{\star} + \int_0^T \| [f(x,s)]_x \|_{\star} \, ds}{\| [u(x,T)]_x \|_{\star}} \right).
\]

(4) Number of qubits: We need to consider three parts of qubits. To encode the solution vector $\widetilde{v}(T)$, we need
\(\log N\) qubits. To encode the linear combination of $M_t$ terms in \eqref{para vb}, we need $\log M_t$ qubits. To encode the block-encoding in Lemma \ref{para be of homo} and \ref{para be of nonh}, we need $\max\{m_v,m_b\}$ qubits. So the total number is $$\mathcal{O}(\log N+\log M_t+\max\{m_v,m_b\})=\mathcal{O}(\log N+\log T+\log \|\widetilde{A}\|_2)+\mathcal{O}(\max\{m_v,m_b\}).$$

By Lemma \ref{para be of homo} and \eqref{para epsilonv}, we have \[m_v\lesssim \log(T\|L\|_2+\log \mathrm{cond}(P)+\log \dfrac{ \| v_0 \|_{\star} }{\| [ u(x, T) ]_{x}\|_{\star}}+\log \frac{1}{\varepsilon})+\log d+\log N.\]
Since $\log \|L\|_2\leq \log\alpha_L\lesssim \log N$ and $\log \mathrm{cond}(P)\lesssim \log N$, we have \[m_v\lesssim \log\log \dfrac{ \| [u_0]_x \|_{\star} }{\| [ u(x, T) ]_{x}\|_{\star}}+\log d+\log N+\log T.\]
Similarly, \[ m_b\lesssim \log\log \frac{ \int_0^T \|[f(x,s)]_x\|_{\star} ds }{\| [ u(x, T) ]_{x}\|_{\star}}+\log d+\log N+\log T.  \]
And since $\|\widetilde{A}\|_2\leq \|L\|_2+\|H\|_2\leq \alpha_L+\alpha_H\lesssim \log N,$ we have $$\log N+\log M_t+\max\{m_v,m_b\}\lesssim \log d+\log N+\log T+\log\log \dfrac{ \| [u_0]_x \|_{\star} }{\| [ u(x, T) ]_{x}\|_{\star}}+\log\log \frac{ \int_0^T \|[f(x,s)]_x\|_{\star} ds }{\| [ u(x, T) ]_{x}\|_{\star}}.$$

By the scale of $N_l$ in Theorem \ref{thm para mol}, we have \[\log N_l=\mathcal{O}\left(\log\left( \frac{Tda_l}{\| [ u(x, T) ]_{x}\|_{\star}\varepsilon}|c_l| \cdot \Gamma \right)+\sum_{l' \in S_1,S_2} |c_{l'}|a_{l'} \right). \]
So the total number of qubits is
\begin{align*}
&\mathcal{O}(\log N+\log M_t+\max\{m_v,m_b\})\\
=&\mathcal{O}\left( \sum_l \log N_l+\log\log \dfrac{ \| [u_0]_x \|_{\star} }{\| [ u(x, T) ]_{x}\|_{\star}}+\log\log \frac{ \int_0^T \|[f(x,s)]_x\|_{\star} ds }{\| [ u(x, T) ]_{x}\|_{\star}} \right)\\
= &\widetilde{\mathcal{O}}\left( d \log \frac{Td\Gamma}{\| [ u(x, T) ]_{x}\|_{\star}\varepsilon}  +d \log\left( a_l|c_l|  \right)+d \sum_{l \in S_1\cup S_2} |c_l|a_l\right.\\
&\left.+\log\log \frac{ \| [u_0]_x \|_{\star} }{\| [ u(x, T) ]_{x}\|_{\star}}+\log\log \frac{ \int_0^T \|[f(x,s)]_x\|_{\star} ds }{\| [ u(x, T) ]_{x}\|_{\star}} \right).
\end{align*}
\end{proof}

\section{Method of lines for linear hyperbolic PDEs}\label{section 5}
This chapter carries out the same continuous-to-discrete-to-quantum-ready program for the hyperbolic PDE setting. We first state the hyperbolic problem with mixed boundary conditions. Then as in the parabolic case, we prove uniqueness of classical solutions and handle inhomogeneous boundary data by Coons interpolation. We then discretize the spatial operators by the method of lines to obtain a semi-discrete second-order ODE system, and get its analytic solution by extending the state with auxiliary variables to an equivalent first-order ODEs. Finally we derive error bounds in the mean $l_2$ norm and the scale of number of nodes.

\subsection{Problem setting}
In this chapter, we consider the second-order linear hyperbolic PDE on the $d$-dimensional rectangle
$\Omega=[0,a_1]\times\cdots\times[0,a_d]$:
\begin{equation}\label{hy pde}
\begin{cases}
u_{tt} = \Delta u + \displaystyle\sum_{l=1}^d c_l\,u_{x_l} - c_0^2\,u + f(x,t), & (x,t)\in\Omega\times[0,T],\\
u(x,0)=u_0(x),\quad u_t(x,0)=\phi(x), & x\in\Omega,
\end{cases}
\end{equation}
where $c_0\ge 0$, $c_l\in\mathbb{R}$ ($l=1,\dots,d$) are constants and $f,u_0,\phi$ are given.

\medskip
\noindent\textbf{Boundary conditions.}
As in Section~3, we treat boundary conditions dimension-by-dimension. Let
$S=\{1,\dots,d\}=S_1\sqcup S_2\sqcup S_3$, where $S_1,S_2,S_3$ correspond to
Dirichlet, Neumann and periodic directions, respectively. For $l\in S_1$ we impose homogeneous Dirichlet boundary conditions,
\begin{equation}\label{hy boundary1}
u(x|_{x_l = 0}, t) =u(x|_{x_l = a_l}, t)=0;
\end{equation}
for $l\in S_2$ we impose homogeneous Neumann boundary conditions,
\begin{equation}\label{hy boundary2}
\dfrac{\partial}{\partial x_l} u(x|_{x_l = 0}, t) =\dfrac{\partial}{\partial x_l} u(x|_{x_l = a_l}, t)=0;
\end{equation}
and for $l\in S_3$ we impose periodic boundary conditions (including compatibility of initial value and inhomogeneous term),
\begin{equation}\label{hy boundary3}
\begin{cases}
u(x|_{x_l = 0},t) = u(x|_{x_l = a_l},t), \quad
\dfrac{\partial}{\partial x_l} u(x|_{x_l = 0}, t) = \dfrac{\partial}{\partial x_l} u(x|_{x_l = a_l}, t),\\
u_0(x|_{x_l = 0})=u_0(x|_{x_l = a_l}),\quad
\phi(x|_{x_l = 0})=\phi(x|_{x_l = a_l}),\quad
f(x|_{x_l = 0}, t)=f(x|_{x_l = a_l}, t).
\end{cases}
\end{equation}
We additionally restrict the first-derivative coefficient to vanish in periodic directions, i.e.
\begin{equation}\label{hy stable condition}
c_l=0,\quad \forall\,l\in S_3.
\end{equation}
\begin{remark}
This restriction is motivated by the stability of the exact solution rather than by discretization artifacts. Indeed, when the problem is periodic in $x_l$, a Fourier-series representation shows that a nonzero $c_l$ can introduce exponentially growing modes in time, so that the solution norm increases like $e^{\gamma t}$ for some $\gamma>0$. In contrast, the same phenomenon does not occur for Dirichlet or Neumann boundaries in the hyperbolic PDE, and it also does not arise in the parabolic PDE under any of the three boundary types considered. Therefore, to avoid ill-posed or intrinsically unstable dynamics in the periodic hyperbolic case, we assume that no first-order spatial derivative term is present along periodic directions. 
\end{remark}

Like Section~3, we prove the uniqueness of the solution and consider the boundary lifting for inhomogeneous boundary by Coons interpolation. The proof of Theorem \ref{thm hy unique} can be found in Appendix \ref{appendix proof}.
\begin{theorem}[Uniqueness]\label{thm hy unique}
If the equation \eqref{hy pde} with boundary conditions (\ref{hy boundary1}-\ref{hy boundary3}) admits a solution $u \in C^{2,2}(\Omega \times [0, T])$, then its solution is unique. 
\end{theorem}

\begin{theorem}[Boundary lifting through Coons interpolation]\label{thm hy vanish}
Consider \eqref{hy pde} with inhomogeneous Dirichlet boundary \eqref{boundary1 nonhomo} in directions $S_1$
and inhomogeneous Neumann boundary \eqref{boundary2 nonhomo} in directions $S_2$ and periodic boundary in directions $S_3$.
Let $\mathrm{Interp}_A$ (for Dirichlet directions) and $\mathrm{Interp}'_B$ (for Neumann directions) be the
Coons interpolation terms defined as \eqref{thm parabolic vanish interpA} and \eqref{thm parabolic vanish interpB}.
Define
\[
v_1:=u+\sum_{\emptyset\neq A\subseteq S_1}(-1)^{|A|}\,\mathrm{Interp}_A(x,t),
\quad
v_2:=v_1+\sum_{\emptyset\neq B\subseteq S_2}(-1)^{|B|}\,\mathrm{Interp}'_B(x,t).
\]
Then for the linear differential operator
\(\mathcal{L}u := \Delta u + \sum_{l=1}^d c_l\,u_{x_l}-c_0^2u\) in \eqref{hy pde}, $v_2$ satisfies the same hyperbolic operator with \emph{homogeneous} boundary conditions
\eqref{hy boundary1}--\eqref{hy boundary2} and periodicity \eqref{hy boundary3}, but with transformed inhomogeneous term and initial values
\[
(v_2)_{tt}=\mathcal{L} v_2+\widetilde f(x,t),\quad
v_2(x,0)=\widetilde u_0(x),\quad (v_2)_t(x,0)=\widetilde\phi(x),
\]
where
\begin{align*}        
\tilde{f}(x,t)=&- \mathcal{L}\left(\sum_{\emptyset \neq B \subseteq S_2} (-1)^{|B|}\text{Interp}_B'(x,t)\right) + \left(\sum_{\emptyset \neq B \subseteq S_2} (-1)^{|B|}\text{Interp}_B'(x,t)\right)_{tt}\\
&-\mathcal{L}\left(\sum_{\emptyset \neq A \subseteq S_1} (-1)^{|A|}\text{Interp}_A(x,t)\right)+ \left(\sum_{\emptyset \neq A \subseteq S_1} (-1)^{|A|}\text{Interp}_A(x,t)\right)_{tt}+ f(x,t),\\
\tilde{u}_0(x)=&u_0(x)+\sum_{\emptyset \neq A \subseteq S_1} (-1)^{|A|}\text{Interp}_A(x,0)+\sum_{\emptyset \neq B \subseteq S_2} (-1)^{|B|}\text{Interp}_B'(x,0),\\
\tilde{\varphi}(x)=&\varphi(x)+\left(\sum_{\emptyset \neq A \subseteq S_1} (-1)^{|A|}\text{Interp}_A(x,0)\right)_t+\left(\sum_{\emptyset \neq B \subseteq S_2} (-1)^{|B|}\text{Interp}_B'(x,0)\right)_t.    
\end{align*}
\end{theorem}

\begin{remark}
The proof is the same as Theorem \ref{thm parabolic vanish} in parabolic PDE. And the practical computability can be verified similarly. In the remainder of this paper we consider the homogeneous boundary values. 
This assumption can also simplify the discretization and the error analysis without loss of generality.
\end{remark}

\subsection{Method of lines discretization}\label{section 5.2}

We discretize \eqref{hy pde} in space by the method of lines, using the \emph{same grids and the same
finite difference matrices} $A_l$ in \eqref{para dirichlet matrix}-\eqref{para periodic matrix} as the parabolic PDE.
In this way, we can get a similar semi-discrete ODEs. The only change is that the derivative of $t$ is second-order.
Let $N:=\prod_{l=1}^d N_l$ and define the semi-discrete solution vector and coefficient matrix
\[
v(t)\approx [u(x,t)]_{x}\in\mathbb{R}^N,\quad b(t):=[f(x,t)]_{x}\in\mathbb{R}^N,
\quad v_0:=[u_0(x)]_{x},\quad \varphi:=[\phi(x)]_{x},
\]
\[
A:=\sum_{l=1}^d I^{\otimes (l-1)}\otimes A_l\otimes I^{\otimes (d-l)}.
\]
Then the semi-discrete ODEs is
\begin{align}\label{hy approx odes}
\begin{cases}
\dfrac{\mathrm{d}^2}{\mathrm{d}t^2} v(t) = -A\,v(t)-c_0^2 v(t)+b(t),\\
v(0)=v_0,\quad \dfrac{\mathrm{d}}{\mathrm{d}t} v(0)=\varphi.
\end{cases}
\end{align}

\subsection{Similarity transform and analytic solution}\label{section 5.3}
However, \eqref{hy approx odes} is a second-order ODEs, which cannot be solved by LCHS method. Now we get its analytic solution by using the same similarity transformation as parabolic PDE and extending it to another larger ODEs. We will find the solution has the form which can be simulated by Hamiltonian simulation. Here we adopt the method from \cite{waveequ}.

Following Section~3.3, for each $l\in S_1\cup S_2$ we introduce the same diagonal matrix 
$P_l$ as in \eqref{para trans matrix Pl} and for $l\in S_3$ we set $P_l=I$.
Define the global transformation
\[
P:=\bigotimes_{l=1}^d P_l,\quad \widetilde A:=PAP^{-1}=\sum_{l=1}^d I^{\otimes(l-1)}\otimes \widetilde A_l\otimes I^{\otimes(d-l)}.
\]
 Recall $A_l$ in \eqref{para diri tilde Al}-\eqref{para neu tilde Al} and consider its Cholesky decomposition. For $l$ in $S_1$, by Lemma \ref{lem dirichlet matrix},
we have \begin{align} \label{hy coefficient matrix s1}
\widetilde{A}_l=P_lA_lP_l^{-1}&=\dfrac{1}{h_l^2}\begin{pmatrix}     
2 & -\sqrt{1-c_l^2h_l^2/4}  & &  \\        
-\sqrt{1-c_l^2h_l^2/4} & 2 &\ddots &  \\         
&\ddots &\ddots & -\sqrt{1-c_l^2h_l^2/4}   \\         
& &-\sqrt{1-c_l^2h_l^2/4} & 2    
\end{pmatrix}_{N_l}=D_lD_l^T,\notag\\
D_l&:=\dfrac{1}{h_l}\begin{pmatrix}        
\sqrt{1+c_lh_l/2} & -\sqrt{1-c_lh_l/2}  & & & \\         
& \sqrt{1+c_lh_l/2} &-\sqrt{1-c_lh_l/2} & & \\         
& & \ddots &\ddots &    \\         
& & &\sqrt{1+c_lh_l/2} & -\sqrt{1-c_lh_l/2}    
\end{pmatrix}_{N_l\times (N_l+1)}.
\end{align}
For $l$ in $S_2$, by Lemma \ref{lem neumann matrix},
we have \begin{align} \label{hy coefficient matrix s2}  
\widetilde{A}_l=P_lA_lP_l^{-1}&=\dfrac{1}{h_l^2}\begin{pmatrix}     
1+c_lh_l/2 & -\sqrt{1-c_l^2h_l^2/4}  & & & \\        
-\sqrt{1-c_l^2h_l^2/4} & 2 & -\sqrt{1-c_l^2h_l^2/4} & & \\  
&-\sqrt{1-c_l^2h_l^2/4} &\ddots &\ddots &  \\         
& &\ddots & 2 & -\sqrt{1-c_l^2h_l^2/4} \\         
& & & -\sqrt{1-c_l^2h_l^2/4} & 1-c_lh_l/2    
\end{pmatrix}_{N_l}=D_lD_l^T,\notag\\
D_l&:=\dfrac{1}{h_l}\begin{pmatrix}        
\sqrt{1+c_lh_l/2} &  & & & \\        
-\sqrt{1-c_lh_l/2} & \sqrt{1+c_lh_l/2} &  & & \\  
& -\sqrt{1-c_lh_l/2} &\ddots & &  \\         
& &\ddots & \sqrt{1+c_lh_l/2} &  \\         
& & & -\sqrt{1-c_lh_l/2} & 0    
\end{pmatrix}_{N_l}.
\end{align}
For $l$ in $S_3$, the coefficient matrix can also be decomposed
\begin{align} \label{hy coefficient matrix s3}
\widetilde{A}_l=\dfrac{1}{h_l^2}\begin{pmatrix}        
2 & -1  & & & -1 \\        
-1 & 2 & -1 & & \\  
&-1 &\ddots &\ddots &  \\         
& &\ddots & 2 & -1 \\         
-1 & & & -1 & 2    
\end{pmatrix}_{N_l}=D_lD_l^T,\quad
D_l:=\dfrac{1}{h_l}\begin{pmatrix}        
1 & -1 & & & \\         
& 1 & -1 & & \\  
&  &\ddots & \ddots &  \\         
& & & 1 & -1 \\         
-1 & & &  & 1    
\end{pmatrix}_{N_l}.
\end{align}
Let $C_l:=I^{\otimes (l-1)}\otimes D_l \otimes I^{\otimes (d-l)}$, then \begin{align}\label{hy decompose of A}    
\widetilde{A}=\sum_{l=1}^d I^{\otimes (l-1)}\otimes \widetilde{A}_l \otimes I^{\otimes (d-l)}=\sum_{l=1}^d C_l C_l^T.
\end{align}

Next, we introduce a new Hermitian matrix
\begin{align}\label{hy extended matrix}
    H:=\begin{pmatrix}
        O & c_0I & C_1 & \cdots & C_d\\
        c_0I & O & O & \cdots & O\\
        C_1^T & O & O & \cdots & O\\
        \vdots & \vdots & \vdots & \ddots & \vdots\\
        C_d^T & O & O & \cdots & O
    \end{pmatrix}.
\end{align} By Hamiltonian simulation of $H$, we can get the analytic solution of \eqref{hy approx odes} in the following Lemma.

\begin{lemma}\label{hy analytic solution}
Let $H$ be defined as in \eqref{hy extended matrix} and \(\Lambda_s\) denote the top-left block matrix of \(e^{i(T-s)H}\), i.e. \[\Lambda_s:=(I_{N},O,\cdots,O)e^{i(T-s)H}(I_{N},O,\cdots,O)^T\in \mathbb{R}^{N\times N}.\] And let $B(s):=\int_0^s b(\tau)d\tau+\varphi.$
    Then the analytic solution of semi-discrete ODEs \eqref{hy approx odes} is
    \begin{align*}
        v(T)=P^{-1}\Lambda_0 P v_0+\int_0^T P^{-1}\Lambda_s PB(s)ds.
    \end{align*}
\end{lemma}

\begin{proof}
Let $w=Pv$, then \eqref{hy approx odes} is transformed into
\begin{equation}\label{hy trans odes}
\begin{cases}
\dfrac{\mathrm{d}^2}{\mathrm{d}t^2} w(t) = -\widetilde A\,w(t)-c_0^2 w(t)+P\,b(t),\\
w(0)=P v_0,\quad \dfrac{\mathrm{d}}{\mathrm{d}t} w(0)=P\varphi.
\end{cases}
\end{equation}
We extend the unknown quantity $w(t)$ to $W(t):=(w_*(t),w_0(t),w_1(t),\cdots,w_d(t))^T$ and construct the extended ODEs:
\begin{align*}
    \begin{cases}
    \dfrac{\mathrm{d}}{\mathrm{d}t}W=iHW+\tilde{B},\\
    W(0)=W_0.
    \end{cases},\quad \tilde{B}:=\begin{pmatrix}
        PB(t)\\ 0\\ 0\\ \vdots \\0
    \end{pmatrix},\quad W_0:=\begin{pmatrix}
w(0)\\ 0\\ 0\\ \vdots \\0    
\end{pmatrix}.
\end{align*}
As a first-order linear ODE system, it has a unique solution and we prove that the first term of its solution $w_*(t)$ is $w(t)$. The extended ODEs is 
\[\begin{pmatrix}
\frac{\mathrm{d}}{\mathrm{d}t} w_*(t) \\
\frac{\mathrm{d}}{\mathrm{d}t} w_0(t) \\
\frac{\mathrm{d}}{\mathrm{d}t} w_1(t) \\
\vdots \\
\frac{\mathrm{d}}{\mathrm{d}t} w_d(t)
\end{pmatrix}
= i\begin{pmatrix}
0 & c_0 I & C_1 & \cdots & C_d \\
c_0 I & 0 & 0 & \cdots & 0 \\
C_1^T & 0 & 0 & \cdots & 0 \\
\vdots & \vdots & \vdots & & \vdots \\
C_d^T & 0 & 0 & \cdots & 0
\end{pmatrix}
\begin{pmatrix}
w_*(t) \\
w_0(t) \\
w_1(t) \\
\vdots \\
w_d(t)
\end{pmatrix}
+\begin{pmatrix}
P B(t) \\
0 \\
0 \\
\vdots \\
0
\end{pmatrix}
\]
\[
\Rightarrow\begin{cases}
\frac{\mathrm{d}}{\mathrm{d}t} w_*(t) = i\left( c_0 w_0(t) + C_1 w_1(t) + \cdots + C_d w_d(t) \right) + P\left( \int_0^t b(\tau) d\tau + \varphi \right), \\
\frac{\mathrm{d}}{\mathrm{d}t} w_0(t) = i c_0 w_*(t), \\
\frac{\mathrm{d}}{\mathrm{d}t} w_l(t) = i C_l^T w_*(t), \quad l = 1, \dots, d.
\end{cases}\]
By \( w_l(0) = 0 \), setting \( t = 0 \) gives \( \frac{\mathrm{d}}{\mathrm{d}t} w_*(0) = P\varphi \).
Differentiating \( \frac{\mathrm{d}}{\mathrm{d}t} w_*(t) \), by \eqref{hy decompose of A}, we obtain
\[\begin{aligned}
\frac{\mathrm{d}^2}{\mathrm{d}t^2} w_*(t)
&= i\left( c_0 \frac{\mathrm{d}}{\mathrm{d}t} w_0(t) + C_1 \frac{\mathrm{d}}{\mathrm{d}t} w_1(t) + \cdots + C_d \frac{\mathrm{d}}{\mathrm{d}t} w_d(t) \right) + P b(t) \\
&= - \left( c_0^2 w_*(t) + C_1 C_1^T w_*(t) + \cdots + C_d C_d^T w_*(t) \right) + P b(t) \\
&= - \left( \widetilde{A} + c_0^2 I \right) w_*(t) + P b(t).
\end{aligned}\]
Since \( w_*(0) = w(0) \), \( w_*(t) \) satisfies the ODEs \eqref{hy trans odes}. So \( w_*(t) = w(t) \).

Then, solving the extended ODEs, we have
\[
W(t) = e^{i t H} W_0 + \int_0^t e^{i(t-s)H} \widetilde{B}(s) ds\]
\begin{align*}\Rightarrow
w(t) =\begin{pmatrix} I & 0 & \cdots & 0 \end{pmatrix}W(t)&=
\begin{pmatrix} I & 0 & \cdots & 0 \end{pmatrix}
e^{i t H}
\begin{pmatrix} I & 0 & \cdots & 0 \end{pmatrix}^T
w(0)\\
&+\int_0^t
\begin{pmatrix} I & 0 & \cdots & 0 \end{pmatrix}
e^{i(t-s)H}
\begin{pmatrix} I & 0 & \cdots & 0 \end{pmatrix}^T
P B(s) ds\end{align*}
\[\Rightarrow
v(t) = P^{-1} w(t) = P^{-1} \Lambda_0 P v_0 + \int_0^t P^{-1} \Lambda_s P B(s) ds.\]
\end{proof}

\subsection{Error analysis}\label{section 5.4}

Then we analyze the error between the PDE \eqref{hy pde} and the ODEs \eqref{hy approx odes} and compute the scale of $N_l$.

\begin{theorem}\label{thm hy mol}
Assume $u$ is sufficiently smooth. 
For error $\varepsilon_N$, if $N_l$ satisfy \begin{align}     
N_l &\gtrsim \sqrt{ \frac{T^2d}{\varepsilon_N} } a_l \exp(\sum_{l' \in S_1,S_2} \frac{|c_{l'}|a_{l'}}{4})\left( \max \{ |c_l|, 1 \} \cdot \max \{ \| u_{x_l}^{(3)} \|_{\infty}, \| u_{x_l}^{(4)} \|_{\infty} \} \right)^{\frac{1}{2}},\label{thm hy mol N1}
\end{align}
then the overall error between the solution of PDE \eqref{hy pde} and the solution of its semi-discrete ODEs \eqref{hy approx odes} can be controlled by \( \varepsilon_N \), i.e. 
\[
\| [ u(x, T) ]_{x} - v(T) \|_{\star} \leq \varepsilon_N.
\]
\end{theorem}

\begin{proof}
Let the local truncation error and the overall error be
\[  
\tau(t) := \bigl[\Delta u(x,t) + \sum_{l=1}^d c_l u_{x_l}(x,t)\bigr]_x + A[u(x,t)]_x,  
\quad  
e(t) := [u(x,t)]_x - v(t).
\]
Then $e(t)$ satisfies
\[e''(t)=-(A+c_0^2I)e(t)+\tau(t),\quad e(0)=e'(0)=0.\]
Solving the second-order ODEs gives
\begin{align*}
e(T)&=\int_0^T \sum_{i=0}^{\infty} \frac{1}{(2i+1)!}(-A-c_0^2I)^i (T-s)^{2i+1}\tau(s)\,ds\\
&=\int_0^T P^{-1}g(T-s,B)P\tau(s)\,ds,
\end{align*} where $B:=\widetilde{A}+c_0^2I$ and $ g(T-s,B):=\sum_{i=0}^{\infty} \frac{1}{(2i+1)!}(-B)^i (T-s)^{2i+1}.$
Then again we decompose $$\tau=\sum_{l=1}^d \tau_l,\ \tau_l(t):=  
[u_{x_lx_l}(x,t)+c_lu_{x_l}(x,t)]_x  
+ I^{\otimes(l-1)}\otimes A_l\otimes I^{\otimes(d-l)}[u(x,t)]_x,$$ and get
$$e(T)=\sum_{l=1}^d \int_0^T P^{-1}g(T-s,B)P\tau_l(s)\,ds.$$ 
To guarantee $\starnorm{e(T)}\le \varepsilon_N$, it suffices to show that, for each $l$,
\begin{align} \label{thm hy mol 2}
\left\|\int_0^T P^{-1}g(T-s,B)P\tau_l(s)\,ds\right\|_{\star}  
\le \frac{\varepsilon_N}{d}.
\end{align}

\medskip
\noindent
\textbf{Step 1: Dirichlet and periodic boundaries.}
For $l\in S_1\cup S_3$, similar to Theorem \ref{thm para mol}, we can get \eqref{thm hy mol 1}, i.e. 
\[\|\tau_l(t)\|_\star
\le \left(
\frac{1}{12}\|u^{(4)}_{x_l}\|_\infty
+\frac{|c_l|}{6}\|u^{(3)}_{x_l}\|_\infty
\right)\frac{a_l^2}{N_l^2}.
\]
Since $\widetilde{A}$ is symmetric and positive semi-definite, so is $B=\widetilde{A}+c_0^2I$. If $B$ has an eigenvalue/singular value $\lambda$, then $g(T-s,B)=\sum_{i=0}^\infty \frac{1}{(2i+1)!}(-B)^i(T-s)^{2i+1}$ has corresponding eigenvalue/singular value $$\sum_{i=0}^\infty \frac{1}{(2i+1)!}(-\lambda)^i(T-s)^{2i+1}=\begin{cases}    
\dfrac{\sin(T-s)\sqrt{\lambda}}{\sqrt{\lambda}},\ \text{if }\lambda>0\\    
T-s,\ \text{if }\lambda=0
\end{cases}\leq T-s.$$
So $\left\|g(T-s,B)\right\|_2\leq T-s.$
Therefore,
\[
\left\|
\int_0^T P^{-1}g(T-s,B)P\,\tau_l(s)\,ds
\right\|_\star
\le
\int_0^T \operatorname{cond}(P)\,(T-s)\|\tau_l(s)\|_\star\,ds
\le
\frac{T^2}{2}\operatorname{cond}(P)
\left(
\frac{1}{12}\|u^{(4)}_{x_l}\|_\infty
+
\frac{|c_l|}{6}\|u^{(3)}_{x_l}\|_\infty
\right)h_l^2.
\]
Hence to let \eqref{thm hy mol 2} hold true, it is enough to require
\[  
N_l \geq  
\sqrt{\frac{C T^2 d}{2\varepsilon_N}}  
\, a_l  
\exp\!\left(\sum_{l'\in S_1\cup S_2}\frac{|c_{l'}|a_{l'}}{4}\right)  
\left(    
\frac1{12}\|u_{x_l}^{(4)}\|_{\infty} + \frac{|c_l|}{6}\|u_{x_l}^{(3)}\|_{\infty}  
\right)^{1/2}.
\]

\medskip
\noindent
\textbf{Step 2: Neumann boundaries.}
For $l\in S_2$, we apply Lemma \ref{lem para neu mol} in the $l$-th direction. By (1) of Lemma \ref{lem para neu mol}, we have
\[
\tau_l(t)=
\sqrt{1+c_lh_l/2}\,
I^{\otimes(l-1)}\otimes P_l^{-1}D_lP_l\otimes I^{\otimes(d-l)}\,\eta_l(t)
+\rho_l(t),
\]
where
\[
\|\eta_l(t)\|_\star = \mathcal{O}\!\left(h_l^2\|u^{(3)}_{x_l}\|_\infty\right),
\quad
\|\rho_l(t)\|_\star
=
\mathcal{O}\!\left(
h_l^2\|u^{(4)}_{x_l}\|_\infty
+
h_l^2|c_l|\|u^{(3)}_{x_l}\|_\infty
\right).
\]
Let
\[
C_l:=I^{\otimes(l-1)}\otimes D_l\otimes I^{\otimes(d-l)}.
\]
Then, using the tensor-product structure of $P$,
\[
\tau_l(t)
=
\sqrt{1+c_lh_l/2}\,
P^{-1}C_lP\,\eta_l(t)+\rho_l(t).
\]
Hence
\begin{align*}
\left\|
\int_0^T P^{-1}g(T-s,B)P\,\tau_l(s)\,ds
\right\|_\star
&\le
\left\|
\int_0^T\left( \sqrt{1+c_lh_l/2}P^{-1}g(T-s,B)C_lP\eta+P^{-1}g(T-s,B)P\rho \right)ds
\right\|_\star\\
&\le
\int_0^T
\operatorname{cond}(P)
\left(
\sqrt{1+c_lh_l/2}\,
\|g(T-s,B)C_l\|_2\,\|\eta_l(s)\|_\star
+
\|g(T-s,B)\|_2\,\|\rho_l(s)\|_\star
\right)\,ds.
\end{align*}

We already get $\left\|g(T-s,B)\right\|_2\leq T-s.$ Now we estimate $\|g(T-s,B)C_l\|_2$. Since
\(B=\widetilde A+c_0^2I=\sum_{l'=1}^d C_{l'}C_{l'}^T+c_0^2I,\)
we have $C_lC_l^T\preceq B$. Therefore,
\[
\|g(T-s,B)C_l\|_2^2
=
\lambda_{\max}\!\Bigl(g(T-s,B)C_lC_l^T g(T-s,B)\Bigr)
\le
\lambda_{\max}\!\Bigl(g(T-s,B)Bg(T-s,B)\Bigr).
\]
Since $g(T-s,B)$ commutes with $B$, 
\(g(T-s,B)Bg(T-s,B)=
\sin^2\!\bigl((T-s)B^{1/2}\bigr).\)
Hence
\[
\|g(T-s,B)C_l\|_2\le 1.
\]
Combining this with $\|g(T-s,B)\|_2\le T-s$, we obtain
\[
\left\|
\int_0^T P^{-1}g(T-s,B)P\,\tau_l(s)\,ds
\right\|_\star
\le
\operatorname{cond}(P)
\int_0^T
\left(
\sqrt{1+c_lh_l/2}\,
\|\eta_l(s)\|_\star
+
(T-s)\|\rho_l(s)\|_\star
\right)\,ds
\]
\[
\le
\operatorname{cond}(P)
\left(
T\sqrt{1+c_lh_l/2}\,
\sup_{0\le s\le T}\|\eta_l(s)\|_\star
+
\frac{T^2}{2}\sup_{0\le s\le T}\|\rho_l(s)\|_\star
\right).
\]
Thus
\[
\left\|
\int_0^T P^{-1}g(T-s,B)P\,\tau_l(s)\,ds
\right\|_\star
\lesssim
\operatorname{cond}(P)\,h_l^2
\Bigl(
T\|u^{(3)}_{x_l}\|_\infty
+
T^2\|u^{(4)}_{x_l}\|_\infty
+
T^2|c_l|\|u^{(3)}_{x_l}\|_\infty
\Bigr).
\]
Finally, using
\(
\operatorname{cond}(P)\le
C\exp\!\Bigl(\sum_{l\in S_1\cup S_2}\frac{|c_l|a_l}{2}\Bigr)
\) in \eqref{para condP}, let \eqref{thm hy mol 2} hold true and get scale of $N_l$.
\end{proof}

\section{Quantum algorithm for linear hyperbolic PDEs}\label{section 6}
In this chapter, we compute the analytic solution $v(T)$ in Lemma \ref{hy analytic solution}. We first employ the Gaussian quadrature formula to handle the inhomogeneous term. Next, we specify the oracles and construct block encodings of the extended matrix $H$ using elementary gates. We then simulate the solution via Hamiltonian simulation method based on QSVT.
Finally, we derive the total complexity of the overall PDE algorithm.

\subsection{Discretized implementation}
In this subsection, we approximate the solution of the semi-discrete ODEs~\eqref{hy approx odes} at a final time $T$.
In Lemma \ref{hy analytic solution}, we get the analytic solution of it      \begin{align*}
        v(T)=P^{-1}\Lambda_0 P v_0+\int_0^T P^{-1}\Lambda_s PB(s)ds.
    \end{align*}
    Similar to parabolic PDE, we discretize the inhomogeneous integral in $v(T)$ by a composite Gauss-Legendre quadrature.
Fix a step size $h_t>0$ and let $T/h_t$ be the number of time subintervals.
On each subinterval, we again use a $Q_t$-point Gauss--Legendre rule on $[-1,1]$ with nodes $\{y_q\}_{q=0}^{Q_t-1}$
and weights $\{w_q\}_{q=0}^{Q_t-1}$.
Define the quadrature nodes and coefficients
\begin{equation*}
s_{q_t,m_t} \;:=\; h_t\Big(m_t + \tfrac{1}{2} + \tfrac{1}{2}y_q\Big),  
\quad  
c_{q_t,m_t} \;:=\; \frac{h_t}{2}\, w_{q_t}\, \|B(s_{q_t,m_t})\|_2,  
\quad  
0\leq m_t\leq T/h_t-1,\; 0\leq q_t\leq Q_t-1.
\end{equation*}
Then 
\begin{align*}
    v(T)\approx v_Q(T):= P^{-1}\Lambda_0 Pv_0+ \sum_{m_t=0}^{T/h_t-1} \sum_{q_t=0}^{Q_t-1} c_{q_t,m_t} P^{-1} \Lambda_{s_{q_t,m_t}} P |B(s_{q_t,m_t})\rangle.
\end{align*} 
Let $M_t=\dfrac{T}{h_t}Q_t$ and re-index $q_t,m_t$ by a single index $j_t$. Then we write \begin{align}\label{hy vb}
    v_Q(T)=P^{-1}\Lambda_0 Pv_0+ \sum_{j_t=0}^{M_t-1} c_{j_t} P^{-1}\Lambda_{s_{j_t}} P |B(s_{j_t})\rangle.
\end{align} 

In the following Lemmas, we give the scale of $\sum c_{j_t}$ for subsequent LCU operation and the scales of $Q_t$ and $h_t$. The proof is the same as Lemma \ref{lemma para scaling of c} and Lemma \ref{thm para len-gauss scale}.

\begin{lemma}\label{lemma hy scaling of c}
    Let $c_{j_t},j_t=0,\cdots,M_t-1$ be the quadrature coefficients defined in \eqref{hy vb}, then $c_{j_t}\geq 0$ and $\sum_{j_t=0}^{M_t-1}c_{j_t}=\mathcal{O}(\int_0^T \|B(s)\|_2 ds).$
\end{lemma}

\begin{lemma}\label{thm hy len-gauss scale}
Let $v_{Q}(T)$ be defined by~\eqref{hy vb}.
Fix an error $\varepsilon_{Q}\in(0,1)$. Assume $$\Xi := \sup\left\{ (\| B^{(p)} \|_{\star})^{\frac{1}{p+1}} \,\big|\, p \geq 0,\, t \in [0,T] \right\} < \infty.$$ If we choose
\[
Q_t \geq \frac{1}{\ln 4} \ln \frac{ T\,\mathrm{cond}(P) \Xi}{\varepsilon_{Q}}, \quad h_t \leq \frac{4 Q_t}{e ( \| H \|_2 + \Xi)},
\]then $\|v(T)-v_Q(T)\|_{\star}\leq \varepsilon_{Q}$.
\end{lemma}

\subsection{Oracles and Block-encoding of coefficient matrix}

In this subsection, we specify the oracles and construct block-encodings about the Hermitian matrix $H$ \eqref{hy extended matrix}.

\paragraph{State-preparation oracles.}
Let $N_l$ be chosen as in
Theorem~\ref{thm hy mol} and assume $N_l= 2^{n_l}$ is a power of two for quantum implementation. For the initial state $u_0(x)$ and $v_0=[u_0(x)]_{x}\in\mathbb{R}^N$, we assume access to the oracle 
\begin{align}\label{hy oracle Ov}
    O_v: |0\rangle \rightarrow  |v_0\rangle.
\end{align}
For the initial derivative $\phi(x)$ and the inhomogeneous term $f(x,t)$. In Lemma \ref{hy analytic solution}, the solution depends on $B(s)=\int_0^s b(\tau)d\tau+\varphi$. So we assume access to the oracle 
\begin{align}\label{hy oracle Ob}
    O_b: |j_t\rangle |0\rangle \rightarrow |j_t\rangle |B(s_{j_t})\rangle,
\end{align} where $s_{j_t}$ is the integral nodes in \eqref{hy vb}.

For the nodes and coefficients of composite Gauss-Legendre quadrature formula in \eqref{hy vb}, we also assume access to the oracles
\begin{align}
    O_s&\in (1,1)BE\left(\diag(1-\frac{s_0}{T},1-\frac{s_1}{T},1-\frac{s_{M_t-1}}{T})\right),\label{hy oracle Os}\\
    O_c&:\ |0\rangle \rightarrow \frac{1}{\sqrt{\|c\|_1}}\sum_{j_t=0}^{M_t-1}\sqrt{c_{j_t}}|j_t\rangle,\label{hy oracle Oc}
\end{align} where $\|c\|_1=\sum_{j_t=0}^{M_t-1} c_{j_t}.$

\paragraph{Block-encoding of $H$.} The overall idea is to construct block-encoding of every $C_l$ and to add them together by LCU. And note that in \eqref{hy coefficient matrix s1} when $
l\in S_1$, $D_l$ is not square. So when $l\in S_1$, $C_l$ is not square and have more columns than that in $S_2,S_3$. In the extended matrix $H$, compared to using one block to store $C_l,\ l\in S_2,S_3$, we use two blocks in $\mathbb{R}^{N\times N}$ to store $C_l,\ \in S_1$. So the order of matrix $H$ is $(2+2|S_1|+|S_2|+|S_3|)N=\mathcal{O}(d)\cdot N$. Strictly speaking, we will construct a larger matrix than $H$. However, since the added rows and columns are all zero, this does not affect the calculation.

Firstly, for $r=\lceil \log(2+2|S_1|+|S_2|+|S_3|) \rceil=\mathcal{O}(\log d)$, we block-encode the matrix $|0\rangle\langle j|$ where $j=(j_1j_2\cdots j_r)_2$. We have \[
|0\rangle\langle j|
=
|0^r\rangle\langle j_1j_2\cdots j_r|
=
\bigotimes_{t=1}^r |0\rangle\langle j_t|
=
\bigotimes_{t=1}^r \sigma_{0\,j_t}.
\]
We have obtained the $(1,1)$-block-encoding of $\sigma_{00},\sigma_{01}$ using $\mathcal{O}(1)$ basic gates. By Lemma \ref{lemma be tensor}, we can construct a $(1,r)$-block-encoding of $|0\rangle\langle j|$ using $r=\mathcal{O}(\log d)$ basic gates.

Then we construct block-encoding of three boundaries respectively. Throughout, let $\eta_l^+:=\sqrt{1+\tfrac12 c_l h_l}=\mathcal{O}(1),\ \eta_l^-:=\sqrt{1-\tfrac12 c_l h_l}=\mathcal{O}(1)$. And recall in \eqref{para be of t-}
$$t^-_l := \begin{pmatrix} 0 & 1 & & \\ & \ddots & \ddots & \\ & & \ddots & 1 \\ & & & 0 \end{pmatrix}_{N_l},\quad t^+_l=(t^-_l)^T$$ and we have obtained $(n_l,n_l+\lceil\log n_l\rceil)$-block-encoding of $t^-_l$ or $t^+_l$ using $\mathcal{O}(n_l^2\log n_l)$ basic gates.\\
\textbf{(i) Dirichlet Boundary ($l\in S_1$):} In \eqref{hy coefficient matrix s1}, we regard $D_l$ as consisting of two $N$-th order blocks $$D_l=\frac{1}{h_l}(\eta_l^+ I-\eta_l^- t^-_l\ \ -\eta_l^- \sigma_{10}^{\otimes n_l})\in \mathbb{R}^{N\times 2 N}.$$
By Lemma \ref{lemma be tensor}, we can construct a $(\eta_l^-,n_l)$-block-encoding of $\eta_l^- \sigma_{10}^{\otimes n_l}$ using $\mathcal{O}(n_l)$ basic gates.
By Lemma \ref{lemma be lcu}, we can construct a $(\eta_l^+ +n_l\eta_l^-,n_l+\lceil\log n_l\rceil+1)$-block-encoding of $\eta_l^+ I-\eta_l^- t^-_l$ using $\mathcal{O}(n_l^2\log n_l)$ basic gates. Then for some $0\leq j\leq r-1$, we insert matrix block $C_l=I^{\otimes (l-1)}\otimes D_l \otimes I^{\otimes (d-l)}$ into matrix $H$. We consider $$\begin{pmatrix}
    O & \cdots & O & C_l & O & \cdots \\
    O \\
    \vdots\\
    O
\end{pmatrix}=|0\rangle \langle j|\otimes\left(I^{\otimes (l-1)}\otimes \frac{1}{h_l}(\eta_l^+ I-\eta_l^- t^-_l)\otimes I^{\otimes (d-l)} \right)-|0\rangle \langle j+1|\otimes\left( I^{\otimes (l-1)}\otimes \frac{1}{h_l}\eta_l^- \sigma_{10}^{\otimes n_l}\otimes I^{\otimes (d-l)}\right).$$
Through the block-encoding of $|0\rangle \langle j|$ and LCU (Lemma \ref{lemma be lcu}), we can construct a $(\frac{1}{h_l}(\eta^+_l+\eta^-_l+\eta^-_l n_l),n_l+\lceil\log n_l\rceil+2+r)$-block-encoding of it using $\mathcal{O}(n_l^2\log n_l)$ basic gates.

\textbf{(ii) Neumann Boundary ($l\in S_2$):} In \eqref{hy coefficient matrix s2}, we have $$D_l=\frac{1}{h_l}(\eta_l^+(I-\sigma_{11}^{\otimes n_l})-\eta_l^- t_l^+).$$
By Lemma \ref{lemma be lcu}, we can construct a $(\frac{1}{h_l}(2\eta^+_l+\eta^-_l n_l),n_l+\lceil\log n_l\rceil+2)$-block-encoding of $D_l$ using $\mathcal{O}(n_l^2\log n_l)$ basic gates. Then for some $0\leq j\leq r-1$, we insert matrix block $C_l$ into matrix $H$. Through the block-encoding of $|0\rangle \langle j|$, we can construct a $(\frac{1}{h_l}(2\eta^+_l+\eta^-_l n_l),n_l+\lceil\log n_l\rceil+2+r)$-block-encoding of $|0\rangle \langle j|\otimes C_l$. 

\textbf{(iii) Periodic Boundary ($l\in S_3$):} In \eqref{hy coefficient matrix s3}, we have $$D_l=\frac{1}{h_l}(I-t^-_l-\sigma_{10}^{\otimes n_l}).$$
By Lemma \ref{lemma be lcu}, we can construct a $(\frac{1}{h_l}(2+ n_l),n_l+\lceil\log n_l\rceil+2)$-block-encoding of $D_l$ using $\mathcal{O}(n_l^2\log n_l)$ basic gates. Then through the block-encoding of $|0\rangle \langle j|$, we can construct a $(\frac{1}{h_l}(2+ n_l),n_l+\lceil\log n_l\rceil+2+r)$-block-encoding of $|0\rangle \langle j|\otimes C_l$.

Finally, in \eqref{hy extended matrix}, $H$ can be decomposed into a summation of $2d+2$ term
$$H=c_0\begin{pmatrix}
    O & I & O & \cdots\\
    O\\
    \vdots
\end{pmatrix}+c_0\begin{pmatrix}
    O & I & O & \cdots\\
    O\\
    \vdots
\end{pmatrix}^T+\sum_l \begin{pmatrix}
    O & \cdots & O & C_l & O & \cdots \\
    O \\
    \vdots\\
    O
\end{pmatrix}+\sum_l\begin{pmatrix}
    O & \cdots & O & C_l & O & \cdots \\
    O \\
    \vdots\\
    O
\end{pmatrix}^T.$$
By Lemma \ref{lemma be lcu}, we can obtain the following block-encoding of $H$ with its cost $\mathrm{Gate}(U_H)$:
\begin{align}\label{hy be of H}
  U_H &\in (\alpha_H,\; m_H)\mathrm{BE}(H),
  \quad \alpha_H=\mathcal{O}\left(\sum_{l=1}^d \frac{n_lN_l}{a_l} \right)
  ,\notag \\
   m_H&=\max_{1\leq l\leq d}(n_l+\lceil\log n_l\rceil)+2+r+\lceil\log (2d+2)\rceil,\quad
  \mathrm{Gate}(U_H)=\mathcal{O}\!\Big(\log d\sum_{l=1}^{d} n_l^2\log n_l\Big).
  \end{align}

\subsection{Implementation of Hamiltonian simulation}
In this subsection, we implement the approximation to $v_Q(T)=P^{-1}\Lambda_0 Pv_0+ \sum_{j_t=0}^{M_t-1} c_{j_t} P^{-1}\Lambda_{s_{j_t}} P |B(s_{j_t})\rangle$ in \eqref{hy vb} by Hamiltonian simulation based on QSVT.
\paragraph{Homogeneous term.}We consider block-encoding the matrix \[M_0:=P^{-1}\Lambda_0 P .\] We have constructed the block-encoding of $H$ in \eqref{hy be of H}. Firstly, by Theorem \ref{thm qsvt hs}, we simulate $e^{iTH}$ with internal precision
\(\frac{\varepsilon_v}{\mathrm{cond}(P)}.\)
We can construct a $(1,m_H+2,\frac{\varepsilon_v}{\mathrm{cond}(P)})$-block-encoding of $e^{iTH}$ using $\mathcal{O}\left((\alpha_H T+\log \frac{\mathrm{cond}(P)}{\varepsilon_v})\cdot \mathrm{Gate}(U_H)\right)$ basic gates.
Since $\Lambda_0$ is the top-left block of $e^{iTH}$, it is a $(1,m_H+2+r,\frac{\varepsilon_v}{\mathrm{cond}(P)})$-block-encoding of $\Lambda_0$. 

In parabolic situation, we have already constructed the block-encoding of $P$ and $P^{-1}$ in \eqref{para be of P} and \eqref{para be of Pinv}. By Lemma \ref{lemma be times}, we have the following block-encoding
\begin{align}\label{hy be of Uv}
U_{v}&\in (\alpha_v,m_v,\varepsilon_v)\mathrm{BE}\!\left(M_0\right),\quad \alpha_v= \alpha_P \alpha_{P^{-1}}=\mathcal{O}(\mathrm{cond}(P)),\notag\\
m_v&=m_H+2+r+m_P+m_{P^{-1}}=\mathcal{O}\!\left(\max_{1\leq l\leq d}(n_l+\lceil\log n_l\rceil)+\log d+\sum_{l\in S_1\cup S_2} n_l\right),\notag\\ 
\mathrm{Gate}(U_v)&=\mathcal{O}\left( \left(\alpha_H T+\log \frac{\mathrm{cond}(P)}{\varepsilon_v}\right)\cdot \mathrm{Gate}(U_H)\right)=\mathcal{O}\left( \left(T\sum_{l=1}^d\frac{n_lN_l}{a_l}+\log \frac{\mathrm{cond}(P)}{\varepsilon_v}\right)\cdot \log d\sum_{l=1}^dn_l^2\log n_l \right). 
\end{align} 

With the oracle $O_v$ from~\eqref{hy oracle Ov}, the circuit for homogeneous term is
\[
|0\rangle|0\rangle
\xrightarrow{\,I\otimes O_v\,}
|0\rangle|v_0\rangle
\xrightarrow{\,U_v\,}
|0\rangle \frac{1}{\alpha_v}\widetilde M_0 |v_0\rangle + |\perp_v\rangle,
\]
where $\widetilde M_0$ denotes the matrix induced by the top-left block of $U_v$, and
\begin{align}\label{hy error of hom}
\|\widetilde M_0-M_0\|_2\le \varepsilon_v.
\end{align}

\paragraph{Inhomogeneous term.}We now implement all matrices
\[M_{j_t}:=P^{-1}\Lambda_{s_{j_t}}P\]
simultaneously by a block-diagonal construction.
Define
\begin{align}\label{hy SEL H}
\mathrm{SEL}_{H}
&:=
\mathrm{diag}\!\bigl((T-s_0)H,\dots,(T-s_{M_t-1})H\bigr)
=T(\mathrm{diag}\!(
1-\frac{s_0}{T},\,\dots,\,1-\frac{s_{M_t-1}}{T}
)\otimes I_N)(I_{M_t}\otimes H),\\
\mathrm{SEL}_{M}&:=I\otimes P^{-1} \diag(\Lambda_{s_0},\cdots,\Lambda_{s_{M_t-1}})I\otimes P.
\end{align}
Firstly we simulate $e^{i\text{SEL}_{H}}$ with internal precision \(\frac{\varepsilon_b}{\mathrm{cond}(P)}.\) By the oracle \eqref{hy oracle Os} and the block-encoding of $H$ \eqref{hy be of H}, we can obtain a $(T\alpha_H, m_H+1)$-block-encoding of $\mathrm{SEL}_{H}$. Then using QSVT framework (Theorem \ref{thm qsvt hs}), we can obtain a $(1,m_H+3,\frac{\varepsilon_b}{\mathrm{cond}(P)})$-block-encoding of $e^{i\text{SEL}_{H}}$ using $\mathcal{O}\left( (T\alpha_H+\log \frac{\mathrm{cond}(P)}{\varepsilon_b})\cdot \mathrm{Gate}(U_H)\right)$ basic gates and $\mathcal{O}\left( T\alpha_H+\log \frac{\mathrm{cond}(P)}{\varepsilon_b}\right)$ queries for $O_s$.
Since $\diag(\Lambda_{s_0},\cdots,\Lambda_{s_{M_t-1}})$ can be regarded as the top-left block of $e^{i\mathrm{SEL}_{H}}$, it is a $(1,m_H+3+r,\frac{\varepsilon_b}{\mathrm{cond}(P)})$-block-encoding of $\diag(\Lambda_{s_0},\cdots,\Lambda_{s_{M_t-1}})$. 

Then through the block-encoding of $P$ and $P^{-1}$ in \eqref{para be of P} and \eqref{para be of Pinv}, we can obtain the following block-encoding
\begin{align}\label{hy be of Ub}
U_{b}&\in (\alpha_b,m_b,\varepsilon_b)\mathrm{BE}\!\left(\mathrm{SEL}_{M}\right),\quad \alpha_b= \alpha_P \alpha_{P^{-1}}=\mathcal{O}(\mathrm{cond}(P)),\notag\\
m_b&=m_H+3+r+m_P+m_{P^{-1}}=\mathcal{O}\!\left(\max_{1\leq l\leq d}(n_l+\lceil\log n_l\rceil)+\log d+\sum_{l\in S_1\cup S_2} n_l\right),\notag\\ 
\mathrm{Gate}(U_b)&=\mathcal{O}\left( \left(\alpha_H T+\log \frac{\mathrm{cond}(P)}{\varepsilon_b}\right)\cdot \mathrm{Gate}(U_H)\right)=\mathcal{O}\left( \left(T\sum_{l=1}^d\frac{n_lN_l}{a_l}+\log \frac{\mathrm{cond}(P)}{\varepsilon_b}\right)\cdot \log d\sum_{l=1}^dn_l^2\log n_l \right), 
\end{align}
and we additionally need $\mathcal{O}\left(T\sum_{l=1}^d\frac{n_lN_l}{a_l}+\log \frac{\mathrm{cond}(P)}{\varepsilon_b}\right)$ queries for $O_s$.

Using the oracles $O_b$ and $O_c$ from~\eqref{hy oracle Ob} and \eqref{hy oracle Oc}, the circuit for inhomogeneous term is
\[
|0\rangle|0\rangle|0\rangle
\xrightarrow{\,I\otimes O_c\otimes I\,}
|0\rangle \frac{1}{\sqrt{\|c\|_1}}\sum_{j_t=0}^{M_t-1}\sqrt{c_{j_t}}\,|j_t\rangle|0\rangle
\xrightarrow{\,I\otimes O_b\,}
|0\rangle \frac{1}{\sqrt{\|c\|_1}}\sum_{j_t=0}^{M_t-1}\sqrt{c_{j_t}}\,|j_t\rangle|B(s_{j_t})\rangle
\]
\[
\xrightarrow{\,U_b\,}
|0\rangle \frac{1}{\sqrt{\|c\|_1}}\sum_{j_t=0}^{M_t-1}\sqrt{c_{j_t}}\,|j_t\rangle
\frac{1}{\alpha_b}\widetilde M_{j_t}|B(s_{j_t})\rangle
+|\perp_b\rangle
\xrightarrow{\,I\otimes O_c^\dagger\otimes I\,}
|0\rangle|0\rangle
\frac{1}{\alpha_b\|c\|_1}
\sum_{j_t=0}^{M_t-1}
c_{j_t}\widetilde M_{j_t}|B(s_{j_t})\rangle
+|\perp_b'\rangle,
\]
where for each $j_t$,
\begin{align}\label{hy error of nonin}
\|\widetilde M_{j_t}-M_{j_t}\|_2\le \varepsilon_b.
\end{align}

\paragraph{Combining the two terms.}
We now combine the two terms by one additional LCU step.
Similar to parabolic situation, we prepare one control qubit by
\[
V_{\mathrm{LCU}}|0\rangle
=
\frac{1}{\sqrt{\alpha_v\|v_0\|_2+\alpha_b\|c\|_1}}
\left(
\sqrt{\alpha_v\|v_0\|_2}\,|0\rangle+\sqrt{\alpha_b\|c\|_1}\,|1\rangle
\right).
\]
and apply the homogeneous circuit and the inhomogeneous circuit, respectively.
We get
\[
|0\rangle|0\rangle|0\rangle|0\rangle
\longmapsto
|0\rangle|0\rangle|0\rangle\,
\frac{1}{\alpha_v\|v_0\|_2+\alpha_b\|c\|_1}\widetilde{v}(T)
+
|\perp\rangle,
\] where
\begin{equation}\label{hy vapp}
\widetilde{v}(T)
:=
\|v_0\|_2\,\widetilde M_0|v_0\rangle
+
\sum_{j_t=0}^{M_t-1}
c_{j_t}\,\widetilde M_{j_t}|b(s_{j_t})\rangle .
\end{equation}
Standard amplitude amplification can then be used to prepare the normalized state
proportional to $\widetilde{v}(T)$. By $\|c\|_1=\mathcal{O}(\int_{0}^T\|B(s)\|_2ds)$ in Lemma \ref{lemma hy scaling of c}, the amplitude of $|\widetilde{v}(T)\rangle$ is \begin{align}\label{hy complexity of aa}
    \frac{\|\widetilde{v}(T)\|_2}{\alpha_v\|v_0\|_2+\alpha_b\|c\|_1}=\mathcal{O}\left( \frac{\|\widetilde{v}(T)\|_{\star}}{{\alpha_v \|v_0\|_{\star}}+\alpha_b \int_0^T\|B(s)\|_{\star}ds} \right).
\end{align}
By 
\[
v_{Q}(T)-\widetilde{v}(T)
=
\|v_0\|_2\,(M_0-\widetilde M_0)|v_0\rangle
+
\sum_{j_t=0}^{M_t-1}
c_{j_t}(M_{j_t}-\widetilde M_{j_t})|b(s_{j_t})\rangle ,
\]we have
\begin{align}\label{hy error of hs}
\|\widetilde{v}(T)-v_{Q}(T)\|_2
\le
\varepsilon_v\|v_0\|_2+\varepsilon_b\|c\|_1 .
\end{align}

\subsection{Total complexity}
Finally, we recall the process of solving the hyperbolic PDE \eqref{hy pde} and compute the total complexity in Theorem \ref{hy final thm}. \\
\textbf{Step 1:} For the spatial discretization error $\| [ u(x, T) ]_{x} - v(T) \|_{\star} \leq \varepsilon_N$, according to Theorem \ref{thm hy mol} we choose the grid sizes $N_l$ as
\begin{align*}  
N_l &=\mathcal{O} \left( \sqrt{ \frac{T^2d|c_l|}{\varepsilon_N} } a_l \exp(\sum_{l' \in S_1,S_2} \frac{|c_{l'}|a_{l'}}{4})\left( \max \{ \| u_{x_l}^{(3)} \|_{\infty}, \| u_{x_l}^{(4)} \|_{\infty} \} \right)^{\frac{1}{2}}\right).
\end{align*}

\noindent\textbf{Step 2:} For the quadrature error $\|v(T)-v_Q(T)\|_{\star}\leq \varepsilon_{Q}$, according to Lemma \ref{thm hy len-gauss scale} we choose $Q_t,h_t,M_t$ as
\[M_t=\mathcal{O}(T\|H\|_2).\]

\noindent\textbf{Step 3:} For the approximation error $\varepsilon_v,\varepsilon_b$, we use Hamiltonian simulation to construct two block-encodings for homnogeneous term and inhomogeneous term. Then \(\|\widetilde{v}(T)-v_{Q}(T)\|_2
\leq
\varepsilon_v\|v_0\|_2+\varepsilon_b\|c\|_1\) as in \eqref{hy error of hs}.

By Lemma \ref{lemma unit error}, the overall error is
\begin{align*}
\| |[ u(x, T) ]_{x}\rangle-|\widetilde{v}(T)\rangle \|_2&\leq \frac{2}{\| [ u(x, T) ]_{x}\|_{\star}} \| [ u(x, T) ]_{x} -\widetilde{v}(T) \|_{\star}\\
&\leq \frac{2}{\| [ u(x, T) ]_{x}\|_{\star}} \left(\varepsilon_N+\varepsilon_Q+\frac{\varepsilon_v\|v_0\|_2+\varepsilon_b\|c\|_1}{\sqrt{N}}\right)\\
&=\frac{2}{\| [ u(x, T) ]_{x}\|_{\star}} \left(\varepsilon_N+\varepsilon_Q+\varepsilon_v\|v_0\|_{\star}  +\frac{\|c\|_1}{\sqrt{N}}\varepsilon_b\right).
\end{align*}So to let $\| |[ u(x, T) ]_{x}\rangle-|\widetilde{v}(T)\rangle \|_2\leq \varepsilon $, we choose
\begin{align}    
\varepsilon_N&=\varepsilon_Q=\frac{1}{8}\| [ u(x, T) ]_{x}\|_{\star}\varepsilon,\label{hy epsiponN}\\    
\varepsilon_v&=\frac{1}{8}\frac{\| [ u(x, T) ]_{x}\|_{\star}}{ \| v_0 \|_{\star} }\varepsilon,\label{hy epsiponv}\\    
\varepsilon_b&=\frac{1}{8}\frac{\| [ u(x, T)]_{x}\|_{\star}\sqrt{N}}{ \|c \|_1 }\varepsilon=\mathcal{O}\left(\frac{\| [ u(x, T)]_{x}\|_{\star}}{\int_0^T \|B(s)\|_{\star} ds}\varepsilon\right),\label{hy epsiponb}
\end{align}
where the scaling of $\varepsilon_b$ is by Lemma \ref{lemma hy scaling of c}.

\begin{theorem}[Total complexity of the hyperbolic PDE solver]\label{hy final thm}
Let $u(x,T)$ be the solution of the hyperbolic PDE \eqref{hy pde} with boundary conditions (\ref{hy boundary1}-\ref{hy boundary3}) and assume that the solution $u$ is sufficiently smooth.. Let $$\Gamma:=\max_{l=1,\cdots,d}\left(\max \{ \| u_{x_l}^{(3)} \|_{\infty}, \| u_{x_l}^{(4)} \|_{\infty} \}\right), \quad q:=\dfrac{\| [u_0(x)]_x \|_{\star} +T \| [\phi(x)]_x \|_{\star}+\int_0^T \int_0^s \| [f(x,\tau)]_x \|_{\star} \,d\tau ds}{\| [u(x,T)]_x \|_{\star}}.$$
For the error $\varepsilon\in (0,1)$, to get a solution $\widetilde{v}(T)$ such that $\| |[ u(x, T) ]_{x}\rangle-|\widetilde{v}(T)\rangle \|_2\leq \varepsilon $ with $\Omega(1)$ probability and a flag indicating success, we need \\ 
(1)  \begin{align*}     
\widetilde{\mathcal{O}}\Bigg( q \cdot T^2d^{3/2} \cdot \left( \dfrac{\Gamma}{\| [u(x,T)]_x \|_{\star}\varepsilon} \right)^{\frac{1}{2}} \Bigg) 
\end{align*} basic gates in $\mathcal{G}$;\\ 
(2) \begin{align*}     
\widetilde{\mathcal{O}}\Bigg( q \cdot T^2d^{1/2} \cdot \left( \dfrac{\Gamma}{\| [u(x,T)]_x \|_{\star}\varepsilon} \right)^{\frac{1}{2}} \Bigg) 
\end{align*}  queries to oracle $O_s$;\\
(3)   $$\mathcal{O}\left( q \right)$$ queries for oracles $O_v,O_b,O_c$;\\ 
(4)   \begin{align*}
\widetilde{\mathcal{O}}\left(d\log\frac{T\Gamma}{\| [ u(x, T) ]_{x}\|_{\star}\varepsilon} \right)     
\end{align*} qubits.
\end{theorem}

\begin{proof}
(1) Gate complexity: Let $\mathrm{cost}_{AA}$ be the complexity of amplitude amplification. The total gate complexity is 
$$\mathrm{cost}_{AA}(\mathrm{Gate}(U_v)+\mathrm{Gate}(U_b)),$$ where $\mathrm{Gate}(U_v)$ and $\mathrm{Gate}(U_b)$ are the complexity of $U_v$ and $U_b$ in \eqref{hy be of Uv} and \eqref{hy be of Ub}.
By \eqref{hy epsiponv} and \eqref{hy epsiponb}, we have
\begin{align*}\mathrm{Gate}(U_v)+\mathrm{Gate}(U_b)&=\mathcal{O}\left( \left(T\sum_{l=1}^d\frac{n_lN_l}{a_l}+\log \frac{\mathrm{cond}(P)}{\varepsilon_v}+\log \frac{\mathrm{cond}(P)}{\varepsilon_b}\right)\cdot \log d\sum_{l=1}^dn_l^2\log n_l \right)\\
&=\mathcal{O} \left(T\sum_{l=1}^d\frac{n_lN_l}{a_l}\cdot \log d\sum_{l=1}^dn_l^2\log n_l \right)\\
&+\mathcal{O} \left(  \left( \log\frac{ {\mathrm{cond}(P)}\|v_0\|_{\star} }{\| [ u(x, T)]_{x}\|_{\star}}+\log \frac{{\mathrm{cond}(P)}\int_0^T \|B(s)\|_{\star} ds}{\| [ u(x, T)]_{x}\|_{\star}}+\log \frac{1}{\varepsilon}\right)  \cdot \log d\sum_{l=1}^dn_l^2\log n_l \right)
\end{align*}
Since $N_l=2^{n_l}$ and $\sum_{l=1}^dn_l^2\log n_l\leq d\max_l n_l^2\log n_l$, we have
\[\mathrm{Gate}(U_v)+\mathrm{Gate}(U_b)=\widetilde{\mathcal{O}}\left( Td\sum_{l=1}^d\frac{N_l}{a_l} \right)+\widetilde{\mathcal{O}}\left( \left( \log\frac{ {\mathrm{cond}(P)}\|v_0\|_{\star} }{\| [ u(x, T)]_{x}\|_{\star}}+\log \frac{{\mathrm{cond}(P)}\int_0^T \|B(s)\|_{\star} ds}{\| [ u(x, T)]_{x}\|_{\star}}+\log \frac{1}{\varepsilon}\right)\cdot \log d\sum_{l=1}^dn_l^2 \right).\]
In \eqref{hy complexity of aa} and by $\alpha_v,\alpha_b=\mathcal{O}(\mathrm{cond}(P))$, the cost of amplitude amplification is  $$\mathrm{cost}_{AA}=\mathcal{O}\left(\frac{{\alpha_v \|v_0\|_{\star}}+\alpha_b \int_0^T\|B(s)\|_{\star}ds}{\|\widetilde{v}(T)\|_{\star}}\right)=\mathcal{O}\left(\frac{ \mathrm{cond}(P)(\|v_0\|_{\star}+\int_0^T \|B(s)\|_{\star} ds )}{\| [ u(x, T)]_{x}\|_{\star}}\right).$$
So we have $\log \mathrm{cost}_{AA}\gtrsim \log \dfrac{{\mathrm{cond}(P)} \|v_0\|_{\star} }{\| [ u(x, T)]_{x}\|_{\star}}$ and $\log \mathrm{cost}_{AA}\gtrsim \log \dfrac{{\mathrm{cond}(P)}\int_0^T \|B(s)\|_{\star} ds}{\| [ u(x, T)]_{x}\|_{\star}}$. And since $(\log N_l)^2\gtrsim \log \dfrac{1}{\varepsilon}$, we have $$ \mathrm{cost}_{AA}(\mathrm{Gate}(U_v)+\mathrm{Gate}(U_b))=\widetilde{\mathcal{O}}\left(\mathrm{cost}_{AA}  \cdot Td\sum_{l=1}^d\frac{N_l}{a_l} \right).$$
Then by Theorem \ref{thm hy mol}, we have
\begin{align}\label{hy final 1} \sum_{l=1}^d \frac{ N_l}{a_l}=\widetilde{\mathcal{O}}\left( \exp(\sum_{l\in S_1\cup S_2}\frac{1}{4}|c_l|a_l)\left(\frac{T^2d\Gamma}{\varepsilon_N}\right)^{\frac{1}{2}} \sum_{l}\sqrt{|c_l|} \right). \end{align}
By $\mathrm{cond}(P) = \mathcal{O}( \exp( \sum_{l \in S_1,S_2} \frac{|c_l| a_l}{2} ))$ and $B(s)=\int_0^s b(\tau)d\tau+\varphi$, we have 
\begin{align}\label{hy final 2}
\mathrm{cost}_{AA}&=\mathcal{O}\left(\frac{ \mathrm{cond}(P)(\|v_0\|_{\star}+\int_0^T \|B(s)\|_{\star} ds )}{\| [ u(x, T)]_{x}\|_{\star}}\right)\notag \\
&=\mathcal{O}\left(\frac{ \exp( \sum_{l \in S_1,S_2} \frac{|c_l| a_l}{2} )(\|[u_0]_x\|_{\star}+T \| [\phi]_x \|_{\star}+\int_0^T \int_0^s \| [f(x,\tau)]_x \|_{\star} \,d\tau ds )}{\| [ u(x, T)]_{x}\|_{\star}}\right).
\end{align}
Then the total gate complexity is
\begin{align*}&\mathrm{cost}_{AA}(\mathrm{Gate}(U_v)+\mathrm{Gate}(U_b))=\widetilde{\mathcal{O}}\left(\left(\frac{ \|[u_0]_x\|_{\star}+T \| [\phi]_x \|_{\star}+\int_0^T \int_0^s \| [f(x,\tau)]_x \|_{\star} \,d\tau ds }{\| [ u(x, T)]_{x}\|_{\star}}\right)\cdot Td\right.\\
&\left.\cdot  \exp(\sum_{l\in S_1\cup S_2}\frac{3}{4}|c_l|a_l)\left(\frac{T^2d\Gamma}{\| [u(x,T)]_x \|_{\star}\varepsilon}\right)^{\frac{1}{2}} \sum_{l}\sqrt{|c_l|} \right).\end{align*}

(2) Query complexity to $O_s$: By \eqref{hy be of Ub}, the query complexity is
$$\mathrm{cost}_{AA}\cdot \mathcal{O}\left(T\sum_{l=1}^d\frac{n_lN_l}{a_l}+\log \frac{\mathrm{cond}(P)}{\varepsilon_v}\right)=\widetilde{\mathcal{O}}\left(  \mathrm{cost}_{AA}\cdot T\sum_{l=1}^d\frac{N_l}{a_l}\right).$$ Here we obtain the $\widetilde{\mathcal{O}}$ by similar analysis to case (1). Then by \eqref{hy final 1} and \eqref{hy final 2}, the complexity is
\begin{align*}
\widetilde{\mathcal{O}}\left(\left(\frac{ \|[u_0]_x\|_{\star}+T \| [\phi]_x \|_{\star}+\int_0^T \int_0^s \| [f(x,\tau)]_x \|_{\star} \,d\tau ds }{\| [ u(x, T)]_{x}\|_{\star}}\right)\cdot T\cdot  \exp(\sum_{l\in S_1\cup S_2}\frac{3}{4}|c_l|a_l)\left(\frac{T^2d\Gamma}{\| [u(x,T)]_x \|_{\star}\varepsilon}\right)^{\frac{1}{2}} \sum_{l}\sqrt{|c_l|})\right).\end{align*}

(3) Queries to $O_v$, $O_b$, $O_c$: A single implementation of Hamiltonian simulation requires $\mathcal{O}(1)$ queries of $O_v$, $O_b$, $O_c$. So the complexity is
\[
\mathcal{O}(1) \cdot \mathrm{cost}_{AA} = \mathcal{O}\left(\exp\left( \sum_{l \in S_1\cup S_2} \frac{1}{2} |c_l| a_l \right) \frac{\| [u_0(x)]_x \|_{\star} +T \| [\phi(x)]_x \|_{\star}+\int_0^T \int_0^s \| [f(x,\tau)]_x \|_{\star} \,d\tau ds}{\| [u(x,T)]_x \|_{\star}} \right).
\]

(4) Number of qubits: We need to consider three parts of qubits. To encode the solution vector $\widetilde{v}(T)$, we need
\(\log N\) qubits. To encode the linear combination of $M_t$ terms in \eqref{hy vb}, we need $\log M_t$ qubits. To encode the block-encoding in \eqref{hy be of Uv} and \eqref{hy be of Ub}, we need $m_b$ qubits. So the total number is $$\log N+\log M_t+m_b=\mathcal{O}(\log N+\log T+\log \|H\|_2+\log d).$$
In the block-encoding of $H$ \eqref{hy be of H}, we have $\|H\|_2\leq \alpha_H=\mathcal{O}\left(\sum_{l=1}^d\frac{n_lN_l}{a_l}\right)\leq \widetilde{\mathcal{O}}(N).$ So
$$\log N+\log M_t+m_b=\widetilde{\mathcal{O}}(\log N+\log T+\log d).$$
Then by Theorem \ref{thm hy mol}, $$\log N=\sum_{l=1}^d\log N_l=\mathcal{O}\left(d\log\frac{Td\Gamma}{\| [ u(x, T) ]_{x}\|_{\star}\varepsilon}+\sum^d_{l=1}\log(a_l|c_l|)+d\sum_{l\in S_1\cup S_2}a_l|c_l|\right).$$
So the number of qubits is
$$\log N+\log M_t+m_b=\widetilde{\mathcal{O}}\left(d\log\frac{T\Gamma}{\| [ u(x, T) ]_{x}\|_{\star}\varepsilon}+\sum^d_{l=1}\log(a_l|c_l|)+d\sum_{l\in S_1\cup S_2}a_l|c_l|\right).$$
\end{proof}

\section{Numerical experiment}\label{section 7}
In this chapter, we will conduct some numerical experiments on the convection-diffusion equation, inhomogeneous heat equation, and Klein-Gordon equation, to verify some results presented in the previous sections and verify the validity of our algorithms.

\subsection{Parabolic PDE}
\paragraph{Validity of similarity transform.} In Section \ref{section 3.3}, we use similarity transform to make coefficient matrix positive semi-definite. If we don't do this and use transformation $\tilde{v}=e^{-ct}v$, what will happen? We consider the following experiment.

Let $\Omega=[0,1]^2$ and define \[\phi_1(x_1)=e^{-c_1x_1/2}\left(\cos(\pi x_1)+\frac{c_1}{2\pi}\sin(\pi x_1)\right),
\quad
\phi_2(x_2)=e^{-c_2x_2/2}\left(\cos(\pi x_2)+\frac{c_2}{2\pi}\sin(\pi x_2)\right),\quad \Lambda:=2\pi^2+\frac{c_1^2+c_2^2}{4}.\]
We consider the following two-dimensional convection–diffusion equation with homogeneous Neumann boundary condition
\[
\begin{cases}
u_t = \Delta u + \bm{c}\cdot \nabla u + f(x,t),
& (x,t)\in [0,1]^2\times[0,T],\\[4pt]
u(x,0)=u_0(x),
& 
\end{cases}
\]
where \[f(x,t)=
\bigl(1+\Lambda(1+t)\bigr)\phi_1(x_1)\phi_2(x_2),\quad u_0(x)=\phi_1(x_1)\phi_2(x_2).\]
Then the analytic solution is
\[u(x,t)=(1+t)\,\phi_1(x_1)\phi_2(x_2).\]

We solve the equation with similarity transform and without similarity transform respectively. About parameter settings, we choose the first-order coefficient $\bm{c}_1=\bm{c}_2=2$, the Gaussian quadrature parameter $Q_t = 6,\ h_t = 0.025$ and the LCHS parameter $\varepsilon_{lchs} = 0.0001,\ c = 1.0,\ h = 0.025$. 
Let the numbers of nodes $N_l$ be $32$ or $64$ and we get two semi-log plots--Figure \ref{para1 1} and Figure \ref{para1 2}. The abscissa is the time $T$ and the ordinate is normalized error $\||\widetilde{v}\rangle-|u\rangle\|_2$ or unnormalized error $\|\widetilde{v}-u\|_{\star}$.

\begin{figure}[H]
    \centering
    \includegraphics[width=0.95\textwidth]{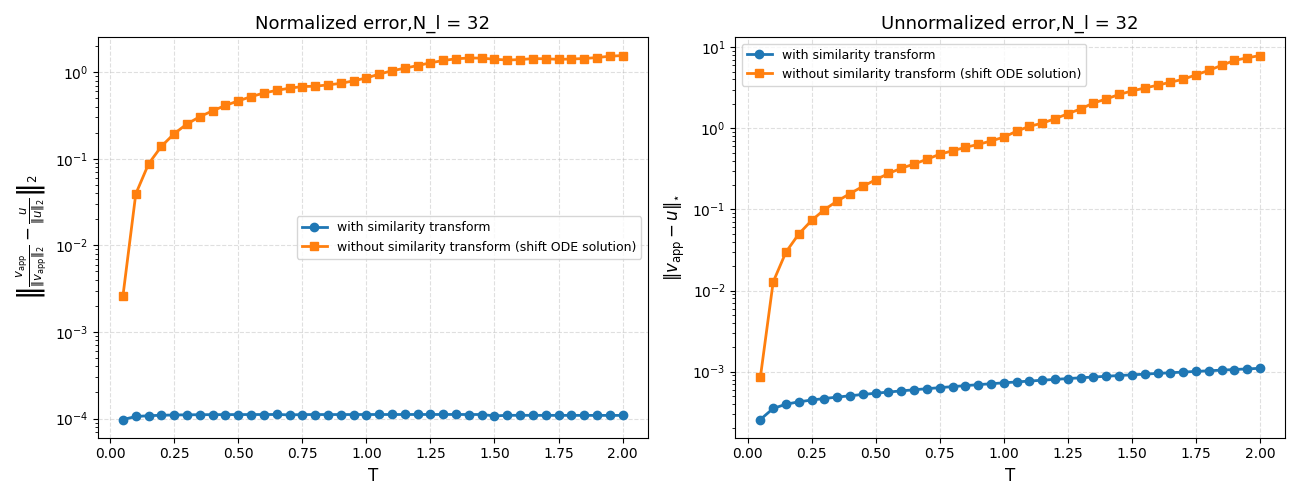}
    \caption{Validity of similarity transform, $N_l=32$}
    \label{para1 1}
\end{figure}
\begin{figure}[H]
    \centering
    \includegraphics[width=0.95\textwidth]{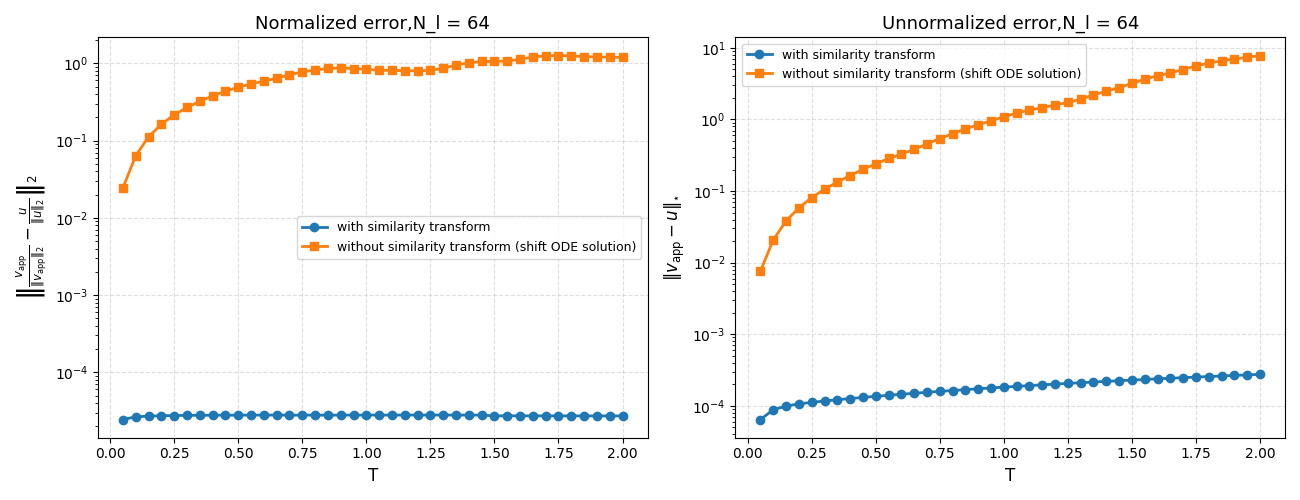}
    \caption{Validity of similarity transform, $N_l=64$}
    \label{para1 2}
\end{figure}

From this plot, we can find that when $T$ increases, the error without similarity transform becomes worse and worse. As $T$ rises to $1.0$, the numerical solution becomes meaningless. But the error with similarity transform is very stable. So although the LCHS framework allows the eigenvalues of coefficient matrix to be negative, it leads to an error growing exponentially with time $T$.

\paragraph{Neumann: midpoint vs ghost-point.} In Section \ref{section 3.2}, we use the midpoint difference for Neumann boundary. What will happen if we use ghost-point difference? Here the matrix of it is \[\frac{1}{h^2}\begin{pmatrix}
    2 & -2 \\
    -1 & 2 & -1 \\
     & -1 & \ddots & \ddots \\
     & & \ddots & 2 & -1\\
     & & & -2 &2
\end{pmatrix},\]which does not have positive semi-definite real part. We consider the following experiment.

Let $\Omega=[0,1]^2$ and define
\[
\Phi(x)=\cos(\pi x_1)\cos(\pi x_2)+0.37\,\cos(2\pi x_1)\cos(\pi x_2).
\]
We consider the inhomogeneous heat equation with homogeneous Neumann boundary condition
\[
\left\{
\begin{aligned}
&u_t = \Delta u + f(x,t), \quad x\in[0,1]^2,\ t\in[0,T],\\
&u(x,0)= u_0(x),
\end{aligned}
\right.
\]
where
\[
f(x,t)=\bigl(1+2\pi^2(1+t)\bigr)\cos(\pi x_1)\cos(\pi x_2)
+0.37\,\bigl(1+5\pi^2(1+t)\bigr)\cos(2\pi x_1)\cos(\pi x_2),\quad u_0(x)=\Phi(x).
\]
Then the analytic solution is
\[
u(x,t)=(1+t)\Phi(x).
\]

We solve the equation using midpoint difference and ghost-point difference. In ghost-point method, we must use transformation $\tilde{v}=e^{-ct}v$. About parameter settings, we choose the fixed time $T=1$, the Gaussian quadrature parameter $Q_t = 6,\ h_t = 0.05$ and the LCHS parameter $\varepsilon_{lchs} = 0.0001,\ c = 1.0,\ h = 0.05$.
Then we get the semi-log plot--Figure \ref{para2 1}. The abscissa is the numbers of nodes $N_l$ and the ordinate is normalized error $\||\widetilde{v}\rangle-|u\rangle\|_2$ or unnormalized error $\|\widetilde{v}-u\|_{\star}$.

\begin{figure}[H]
    \centering
    \includegraphics[width=0.95\textwidth]{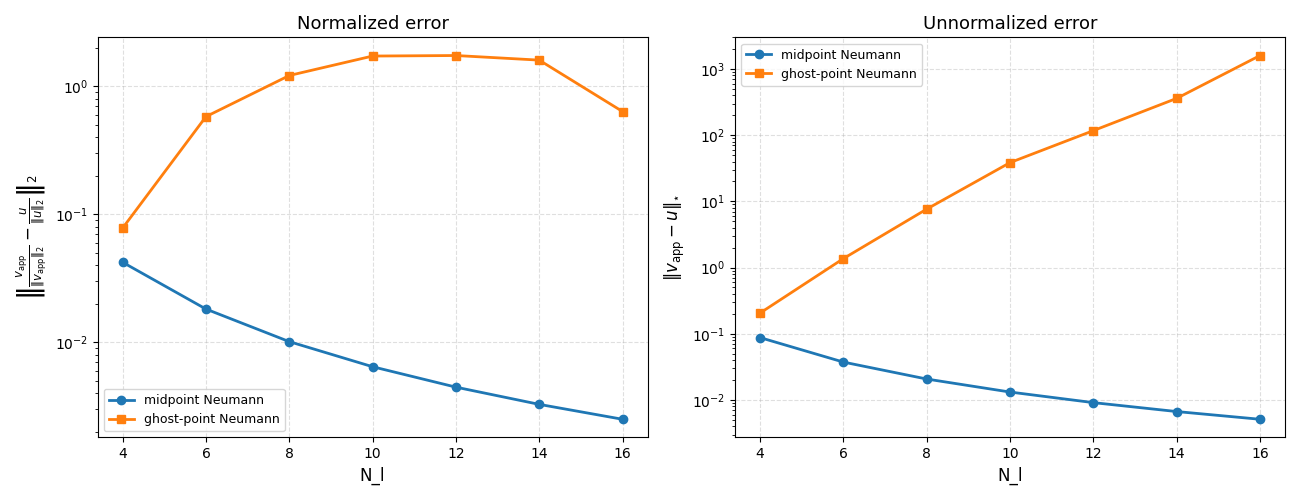}
    \caption{Neumann: midpoint vs ghost-point.}
    \label{para2 1}
\end{figure}

Here when $N=8$, the solution of ghost-point becomes meaningless. Whereas the method without similarity transform remains valid at small times, the ghost-point method breaks down entirely. Similarly, all methods without positive semi-definite real part of coefficient matrix will break down. Because the negative eigenvalue will be amplified by $\mathcal{O}(N^2)$, which is catastrophic. 

For classical methods, we only need that the eigenvalues of coefficient matrix have nonnegative real part. But for LCHS method, we need that the real part of coefficient matrix is positive semi-definite, which is a stronger condition.

\paragraph{Different kernel functions.} Different kernel functions have been proposed during the development of LCHS method. We compare the practical performance of different kernel functions in \cite{lchs1},\cite{lchs2} and \cite{lchs3}. In \cite{lchs1}, the authors use the kernel \[
g_1(k)=\frac{1}{\pi(1+k^2)};
\]
In \cite{lchs2}, the authors use the kernel \[
g_2(k)=\frac{1}{C_\beta(1-ik)e^{(1+ik)^\beta}},
\quad 0<\beta<1;
\]
and in \cite{lchs3}, the authors use \begin{equation*}
g_3(k)=\widehat{f}_2(k;\gamma,c)
=\sqrt{\frac{2}{\pi}}\,\frac{e^{-c(ik-1)}}{1+k^2}\,
\exp\!\left(-\frac{k^2+1}{4\gamma^2}\right),\quad \gamma,c>0,
\end{equation*}which is the kernel used in this paper.

We consider the one-dimensional homogeneous heat equation with homogeneous Dirichlet boundary condition
\[
\left\{
\begin{aligned}
&u_t = u_{xx}, \quad (x,t)\in[0,1]\times[0,T],\\
&u(x,0)=u_0(x),
\end{aligned}
\right.
\]
where the initial value is chosen as
\[
u_0(x)=x(1-x)e^x.
\]
In this experiment, our purpose is to compare the kernel functions themselves rather than the spatial discretization error. Therefore, for each numbers of nodes $N$, we use the exact solution of the semi-discrete ODEs $v$ as the reference solution.
For all three kernel choices, we use the uniform trapezoidal quadrature with interval length $h$ on the truncated interval $[-R,R]$. Because \cite{lchs3} has proven that the uniform trapezoidal quadrature has the same performance as Gaussian quadrature.

In the experiment, we fix $N=64$, $T=0.5$, and $R=50$ or $100$, and vary only $h$, or equivalently, the number of $k$-nodes $M=2R/h+1$. This allows us to compare the three kernel functions under the same numerical backend and the same truncation range.
And we choose $\beta=0.55$ in $g_2(k)$ and $c=0.3$ in $g_3(k)$, at which the errors are relatively small.
We plot two kind of errors against the number of $k$-nodes $M$ for the three kernel functions--Figure \ref{para3 1} and Figure \ref{para3 2}.

\begin{figure}[H]
    \centering
    \includegraphics[width=0.95\textwidth]{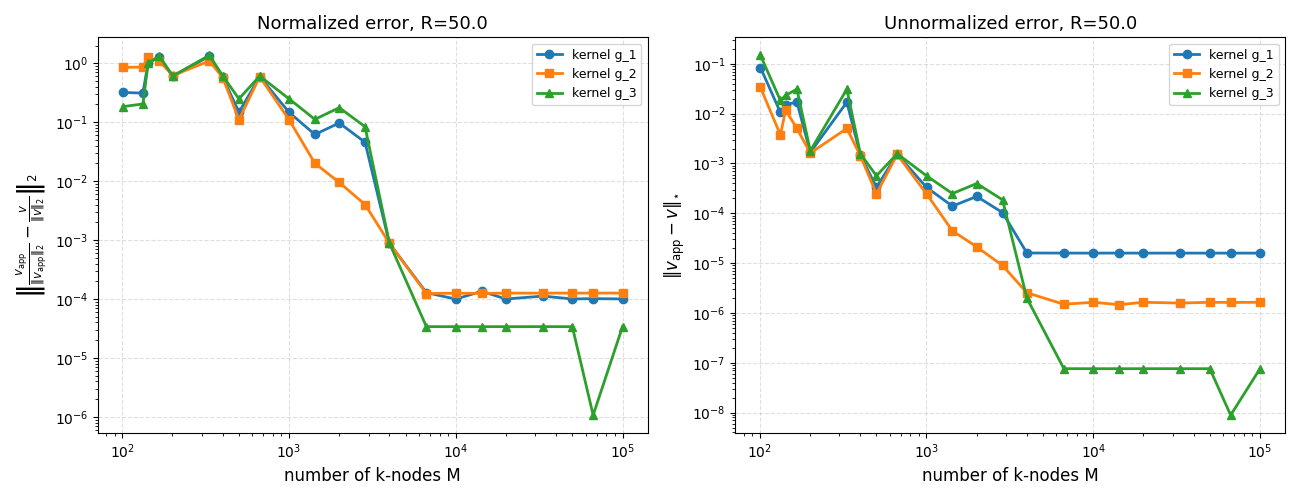}
    \caption{Different kernel functions, $R=50$}
    \label{para3 1}
\end{figure}
\begin{figure}[H]
    \centering
    \includegraphics[width=0.95\textwidth]{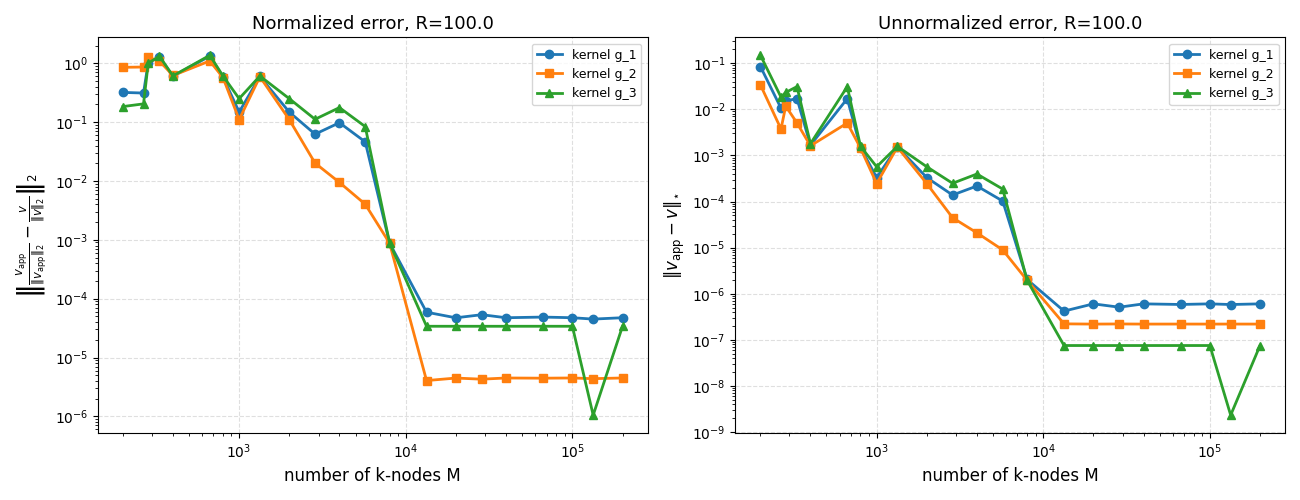}
    \caption{Different kernel functions, $R=100$}
    \label{para3 2}
\end{figure}
This example is constructed as a spatiotemporally inseparable case with a non-trigonometric solution. We can find that the kernel $g_3$ performs best most of the time and $g_2$ also has decent performance in this example. But the performance of kernel functions is dependent on the specific problem. When the spatial part of the analytical solution is trigonometric, $g_1$ usually outperforms $g_2$ and $g_3$. Nevertheless, owing to the theoretical guarantees of $g_2$ and $g_3$, their errors remain acceptable under all circumstances.

\paragraph{Second-order overall error.} 
In Section~\ref{section 3.4}, we show the second-order overall error, even though the truncated errors in Neumann directions are not second-error. In this experiment, we verify this point. Consider the following experiments.

(1) Dirichlet: Let $\Omega=[0,1]^2$ and define
\[
\Phi_{D,1}(x)=e^{-x_1/2}\sin(\pi x_1)\cdot e^{-2x_2/2}\sin(\pi x_2),\quad
\Phi_{D,2}(x)=e^{-x_1/2}\sin(2\pi x_1)\cdot e^{-2x_2/2}\sin(\pi x_2).
\]
Set
\[
\lambda_1=2\pi^2+\frac{5}{4},\quad
\lambda_2=5\pi^2+\frac{5}{4}.
\]
We consider the convection--diffusion equation with homogeneous Dirichlet boundary condition
\[
\left\{
\begin{aligned}
\frac{\partial}{\partial t}u_D
&=\Delta u_D+\frac{\partial}{\partial x_1}u_D+2\frac{\partial}{\partial x_2}u_D+f(x,t),
\quad (x,t)\in[0,1]^2\times [0,T],\\
u_D(x,0)&=u_0(x),
\end{aligned}
\right.
\]
where
\[
u_0(x)=\Phi_{D,1}(x)+0.37\,\Phi_{D,2}(x),
\]
\[
f(x,t)
=
\bigl(1+\lambda_1(1+t)\bigr)\Phi_{D,1}(x)
+0.37\,\bigl(2t+\lambda_2(1+t^2)\bigr)\Phi_{D,2}(x).
\]
The analytic solution is
\[
u_D(x,t)=(1+t)\Phi_{D,1}(x)+0.37(1+t^2)\Phi_{D,2}(x).
\]

(2) Neumann: Let $\Omega=[0,1]^2$ and define
\begin{align*}
\Phi_{N,1}(x)=e^{-x_1/2}\left(\cos(\pi x_1)+\frac{1}{2\pi}\sin(\pi x_1)\right)\cdot e^{-2x_2/2}\left(\cos(\pi x_2)+\frac{1}{\pi}\sin(\pi x_2)\right),\\
\Phi_{N,2}(x)=e^{-x_1/2}\left(\cos(2\pi x_1)+\frac{1}{4\pi}\sin(2\pi x_1)\right)\cdot e^{-2x_2/2}\left(\cos(\pi x_2)+\frac{1}{\pi}\sin(\pi x_2)\right).
\end{align*}
Set
\[\lambda_1=2\pi^2+\frac54,\quad \lambda_2=5\pi^2+\frac54.\]
We consider the convection--diffusion equation with homogeneous Neumann boundary condition
\[
\left\{
\begin{aligned}
\frac{\partial}{\partial t}u_N
&=\Delta u_N+\frac{\partial}{\partial x_1}u_N+2\frac{\partial}{\partial x_2}u_N+f(x,t),
\quad (x,t)\in[0,1]^2\times [0,T],\\
u_N(x,0)&=u_0(x),
\end{aligned}
\right.
\]
where
\[
u_0(x)=\Phi_{N,1}(x)+0.37\,\Phi_{N,2}(x),
\]
\[
f(x,t)
=
\bigl(1+\lambda_1(1+t)\bigr)\Phi_{N,1}(x)
+0.37\,\bigl(2t+\lambda_2(1+t^2)\bigr)\Phi_{N,2}(x).
\]
The analytic solution is
\[
u_N(x,t)=(1+t)\Phi_{N,1}(x)+0.37(1+t^2)\Phi_{N,2}(x).
\]

We solve the two equations numerically. About parameter settings, we choose the fixed time $T=1$, the Gaussian quadrature parameter $Q_t = 7,\ h_t = 0.025$ and the LCHS parameter $R=15,\ \gamma=5,\ c = 1,\ h = 0.05$. Then we get the log-log plot--Figure \ref{para4 1}. 
The abscissa is the numbers of nodes $N_l$ and the ordinate is normalized error $\||\widetilde{v}\rangle-|u\rangle\|_2$ or unnormalized error $\|\widetilde{v}-u\|_{\star}$.
\begin{figure}[H]    
\centering    
\includegraphics[width=0.95\textwidth]{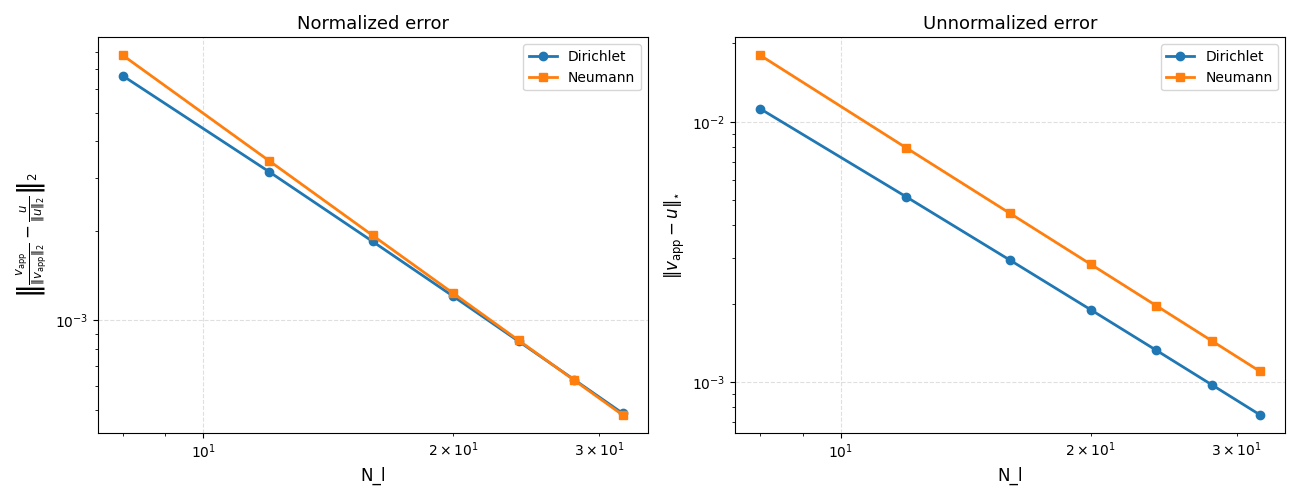} 
\caption{Second-order overall error}    
\label{para4 1}
\end{figure}

The error-$N_l$ curves in Figure \ref{para4 1} are all straight lines with the slope close to $-2$, which agrees well with the derived second-order overall error in Section \ref{section 3.4}.

\subsection{Hyperbolic PDE}
\paragraph{Validity and second-order overall error.} For hyperbolic equations, we don't need to do experiments about similarity transform and midpoint difference. Because if not using these, we can't get the matrix $H$ in \eqref{hy extended matrix} and solve numerically. So we only verify the validity and the second-order overall error in Section \ref{section 5.4}. And hyperbolic equations don't need LCHS substep, thus resulting in a smaller computational cost. So we can perform classical simulations for high-dimensional hyperbolic equations to verify the validity. Consider the following experiments.

(1) Dirichlet: Let $\Omega=[0,1]^5$ and define
\[p_D(x)=e^x-1-(e-1)x,\quad
q_D(x)=e^{2x}-1-(e^2-1)x,\]
\[\Phi_{D,1}(x)=\prod_{l=1}^5 p_D(x_l),\quad
\Phi_{D,2}(x)=q_D(x_1)q_D(x_2)\prod_{l=3}^5 p_D(x_l).\]
We consider the Klein--Gordon equation with homogeneous Dirichlet boundary
\[
\left\{
\begin{aligned}
\frac{\partial^2}{\partial t^2}u_D
&=\Delta u_D-u_D+f(x,t),\quad x\in[0,1]^5,\\
u_D(x,0)&=u_0(x),\quad
\frac{\partial}{\partial t}u_D(x,0)=\phi(x),
\end{aligned}
\right.
\]
where
\[u_0(x)=\Phi_{D,1}(x)+0.37\Phi_{D,2}(x),\quad
\phi(x)=0,\]
\[f(x,t)=-\cos t\,\Delta\Phi_{D,1}(x)
-0.37\cos(2t)\bigl(\Delta\Phi_{D,2}(x)+3\Phi_{D,2}(x)\bigr).\]
The analytic solution is
\[u_D(x,t)=\cos t\,\Phi_{D,1}(x)+0.37\cos(2t)\,\Phi_{D,2}(x).\]

(2) Neumann: Let $\Omega=[0,1]^5$ and define
\[p_N(x)=e^x-x-\frac{e-1}{2}x^2,\quad
q_N(x)=e^{2x}-1-2x-(e^2-1)x^2,\]
\[\Phi_{N,1}(x)=\prod_{l=1}^5 p_N(x_l),\quad
\Phi_{N,2}(x)=q_N(x_1)q_N(x_2)\prod_{l=3}^5 p_N(x_l).\]
We consider the Klein--Gordon equation with homogeneous Neumann boundary
\[
\left\{
\begin{aligned}
\frac{\partial^2}{\partial t^2}u_N
&=\Delta u_N-u_N+f(x,t),\quad x\in[0,1]^5,\\
u_N(x,0)&=u_0(x),\quad
\frac{\partial}{\partial t}u_N(x,0)=\phi(x),
\end{aligned}
\right.\]
where
\[u_0(x)=\Phi_{N,1}(x)+0.37\Phi_{N,2}(x),\quad
\phi(x)=0,\]
\[f(x,t)=-\cos t\,\Delta\Phi_{N,1}(x)
-0.37\cos(2t)\bigl(\Delta\Phi_{N,2}(x)+3\Phi_{N,2}(x)\bigr).\]
The analytic solution is
\[u_N(x,t)=\cos t\,\Phi_{N,1}(x)+0.37\cos(2t)\,\Phi_{N,2}(x).\]

We solve the two equations numerically. About parameter settings, we choose the fixed time $T=1$ and the Gaussian quadrature parameter $Q_t = 8,\ h_t = 0.025$. Then we get the log-log plot--Figure \ref{hyper1 1}. The abscissa is the numbers of nodes $N_l$ and the ordinate is normalized error $\||\widetilde{v}\rangle-|u\rangle\|_2$ or unnormalized error $\|\widetilde{v}-u\|_{\star}$.

\begin{figure}[H]    
\centering    
\includegraphics[width=0.95\textwidth]{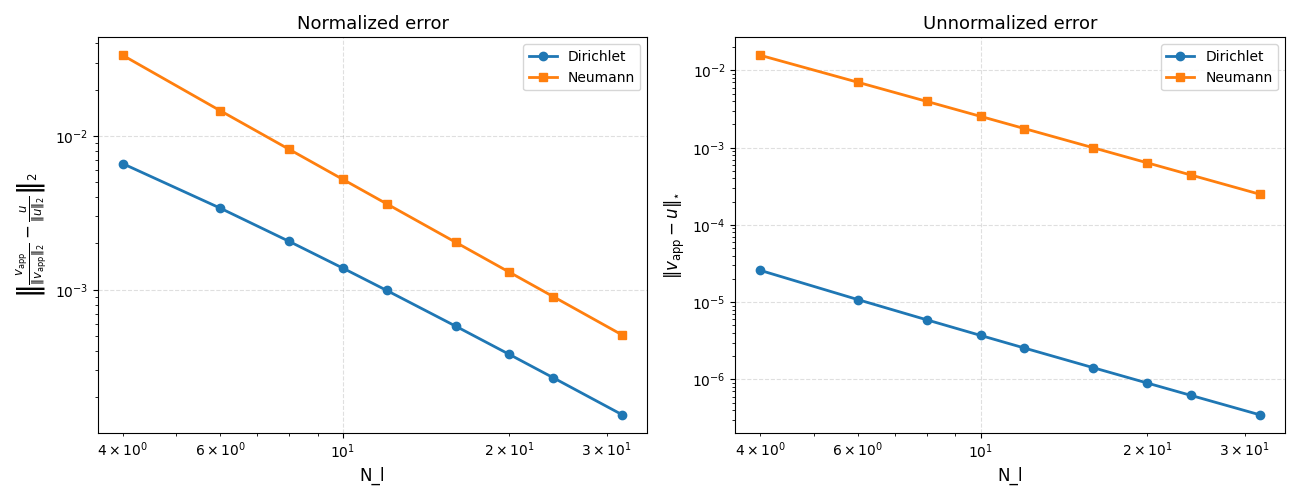} 
\caption{Validity and second-order overall error}    
\label{hyper1 1}
\end{figure}

The error-$N_l$ curves in the log-log plots are all straight lines with the slope close to $-2$. This agrees well with the derived second-order overall error. Moreover, the decline of errors is stable, which indicates that the algorithm is effective in high-dimensional cases.

\section{Discussion}\label{section 8}

In this chapter, we discuss several natural questions and directions for future research.

First, a natural next step is to extend the present framework to more general Robin boundary conditions. Doing so would require new structure-preserving discretization ideas. For the parabolic problem, the main difficulty is to impose the Robin condition while keeping the real part of the semi-discrete coefficient matrix positive semi-definite, which is precisely the structural requirement needed by the LCHS framework. For the hyperbolic problem, the corresponding issue is how to obtain a semi-discrete formulation that still admits an efficient Hamiltonian-simulation-based implementation. By contrast, approaches based on linear-system formulations may be less sensitive to this specific matrix-structure requirement and may therefore handle Robin conditions more directly. 

Second, it is highly desirable to extend the framework to variable coefficients and more general spatial domains. Even for the first-order drift coefficients $c_l(x)$, it is not yet clear under what natural assumptions one can design a spatial discretization such that the resulting semi-discrete system satisfies the stable assumptions required by LCHS. The same question becomes even more subtle on nonrectangular domains, where the tensor-product structure exploited by our current construction is no longer directly available.

Third, our use of the method of lines leads to a polynomial dependence on the target accuracy $\varepsilon$, even though the downstream LCHS and QSVT-based Hamiltonian simulation methods have only polylogarithmic complexity. This suggests the possibility of improving the overall error dependence by replacing finite differences with higher-order or spectral discretizations. In particular, under stronger smoothness assumptions on the solution, it is possible to combine the present ODE-based quantum algorithms with spectral discretization ideas in the spirit of high-precision quantum algorithms for elliptic PDEs \cite{childs2021pde}. Developing such a time-dependent spectral framework could be a route toward optimal or near-optimal quantum algorithms for broader classes of evolutionary PDEs.

Finally, our framework can apply to a number of physically relevant model problems within the two PDE classes studied in this paper. On the parabolic side, it covers linear convection--diffusion equations with source terms, which arise in transport, heat conduction with uniform drift, and linearized diffusion models. It also includes constant-coefficient reaction--diffusion-type equations after shifting the reaction term into the source or zero-order part when appropriate. On the hyperbolic side, the framework applies to wave-type and Klein--Gordon-type equations with transport terms and external forcing, which can be used to model wave propagation, scalar field dynamics, and linearized oscillatory media. More broadly, whenever a model can be reduced to one of the two forms treated here, together with mixed Dirichlet, Neumann, and periodic boundary conditions, the present method provides a systematic route from the PDE to a quantum algorithm with explicit resource bounds.

\section*{Acknowledgments}
JPL acknowledges support from Quantum Science and Technology--National Science and Technology Major Project under Grant No.~2024ZD0300500, Excellent Young Scientists Fund Program, start-up funding from Tsinghua University and Beijing Institute of Mathematical Sciences and Applications.

\bibliographystyle{unsrt} 
\bibliography{references}

\begin{appendices}

\section{Proof of Theorem \ref{thm parabolic unique}, Theorem \ref{thm parabolic vanish} and Theorem \ref{thm hy unique}}\label{appendix proof}

\paragraph{Proof of Theorem \ref{thm parabolic unique}.}
   Let $u_1,u_2$ be two solutions and set $v=u_1-u_2$. We only need to prove $v = 0$. By linearity, we know that $v$ satisfies the equation where $f$ and $u_0$ both vanish
    \[
v_t = \Delta v + \sum_{l=1}^d c_l\,v_{x_l},\quad v(x,0)=0,
\]together with the homogeneous boundary conditions induced by \eqref{boundary1}--\eqref{boundary3}.
 Let the energy be $$E(t) = \frac{1}{2} \int_\Omega v^2 \exp\left(\sum_{l' \in S_2} c_{l'} x_{l'}\right) dx$$ and we will prove $E(t)\equiv 0$. Differentiating and using the PDE for $v$ gives
\begin{align*}
\frac{\mathrm{d}E}{\mathrm{d}t} &= \int_\Omega v v_t \exp\left(\sum_{l' \in S_2} c_{l'} x_{l'}\right) dx \\
&= \sum_{l=1}^d \int_\Omega v (v_{x_{l} x_{l}} + c_{l} v_{x_{l}}) \exp\left(\sum_{l' \in S_2} c_{l'} x_{l'}\right) dx.
\end{align*}

For the summation $\sum_{l=1}^d$ above, we treat each direction $l$ separately. For $l$ in $S_2$, we observe that
\[
\,(v_{x_lx_l}+c_lv_{x_l})\,e^{c_l x_l} = \frac{\partial}{\partial x_l}\!\big(v_{x_l}e^{c_l x_l}\big).
\]
Integrating by parts in $x_l$ and using the Neumann boundary condition $v_{x_l}=0$ at $x_l=0,a_l$ yields
\begin{align*}
\int_{0}^{a_l} v (v_{x_{l} x_{l}}+c_lv_{x_l}) e^{c_l x_l} dx_l=\int_{0}^{a_l} v  d\big(v_{x_l}e^{c_l x_l}\big) =-\int_0^{a_l}v_{x_l}^2e^{c_lx_l}dx_l\leq 0. 
\end{align*}
So $\int_\Omega v (v_{x_{l} x_{l}} + c_{l} v_{x_{l}}) \exp\left(\sum_{l' \in S_2} c_{l'} x_{l'}\right) dx\leq 0$ for $l\in S_2.$

For \( l \) in \( S_1 \cup S_3 \), we also compute the integral about $x_{l}$ by parts
\begin{align*}
\int_{0}^{a_l} v (v_{x_{l} x_{l}}+c_lv_{x_l})dx_l=\left[vv_{x_l}+\frac{1}{2}c_lv^2\right]_0^{a_l}-\int_0^{a_l} v_{x_l}^2dx_l=-\int_0^{a_l} v_{x_l}^2dx_l\leq 0, 
\end{align*}where the term $\left[vv_{x_l}+\frac{1}{2}c_lv^2\right]_0^{a_l}=0$ is by the Dirichlet boundary \eqref{boundary1} or the periodic boundary \eqref{boundary3}. So $\int_\Omega v (v_{x_{l} x_{l}} + c_{l} v_{x_{l}}) \exp\left(\sum_{l' \in S_2} c_{l'} x_{l'}\right) dx\leq 0$ for $l\in S_1\cup S_3.$

Summing over $l=1,\dots,d$ we obtain $\dfrac{\mathrm{d}E}{\mathrm{d}t} \leq 0$. Since $E(t)\geq 0$ and $E(0)=0$, we have $E(t)\equiv 0$. Hence $v\equiv0.$ 

\paragraph{Proof of Theorem \ref{thm parabolic vanish}.}
    1. Firstly, we verify the conditions $v_1$ meets. For the inhomogeneous term and initial value, we have
\begin{align*}
\frac{\partial v_1}{\partial t} &= \mathcal{L}u + f + \left(\sum_{\emptyset \neq A \subseteq S_1} (-1)^{|A|}\text{Interp}_A(x,t)\right)_t\\
&= \mathcal{L}v_1 - \mathcal{L}\left(\sum_{\emptyset \neq A \subseteq S_1} (-1)^{|A|}\text{Interp}_A(x,t)\right) +  f + \left(\sum_{\emptyset \neq A \subseteq S_1} (-1)^{|A|}\text{Interp}_A(x,t)\right)_t\\
&=\mathcal{L}v_1+\hat{f}(x,t),\\
v_1(x,0)&=u_0(x)+\sum_{\emptyset \neq A \subseteq S_1} (-1)^{|A|}\text{Interp}_A(x,0)=\hat{u}_0(x).
\end{align*}

For $l$ in $S_1$, a standard inclusion--exclusion argument shows that, when restricting to the face $\{x_l=0\}$,
terms corresponding to $A$ and $A\cup\{l\}$ cancel each other, and the remaining $A=\{l\}$ term cancels the trace of $u$.
More explicitly, for $A \subseteq S_1$, let $A^+ := A \cup \{l\}$. We have $$\{A \mid \emptyset \neq A \subseteq S_1\} = \{A=\{l\}\} \cup \{A, A^+ \mid l \notin A \neq \emptyset, A \subseteq S_1\}.$$
So \begin{align}\label{thm parabolic vanish 1}
v_1 = u + (-1)\text{Interp}_{\{l\}}(x,t) + \sum_{\substack{l \notin A \subseteq S_1}} (-1)^{|A|}(\text{Interp}_A(x,t) - \text{Interp}_{A^+}(x,t)).    
\end{align}
For $A=\{l\}$ we have \begin{align}\label{thm parabolic vanish 2}
\text{Interp}_{\{l\}}(x,t)=\dfrac{x_l}{a_l}u|_{x_l=a_l}-\dfrac{x_l-a_l}{a_l}u|_{x_l=0}\Rightarrow \  \left.\text{Interp}_{\{l\}}(x,t)\right|_{x_l = 0} = \left.u \right|_{x_l = 0}.    
\end{align}
For $A^+ = \{l_1,\cdots,l_a, l\}$, let $l_{a + 1} = l$ and $\delta_{a + 1}$ be 0 and 1 to compute $\text{Interp}_{A^+}(x,t)$: 
\begin{align*}
\text{Interp}_{A^+}(x,t) &= \sum_{\delta_1,\cdots,\delta_a} \left( \prod_{k = 1}^a \frac{x_{l_k}- \delta_k a_{l_k}}{a_{l_k}} \right) \cdot \left[ \frac{x_l}{a_l}(-1)^{\delta_1 + \cdots + \delta_a} \left.u\right|_{x_{l_k} = a_{l_k} - \delta_ka_{l_k},x_l = a_l} + \right. \\
&\left.\frac{x_l- a_l}{a_l}(-1)^{\delta_1 + \cdots + \delta_a + 1} \left.u \right|_{x_{l_k} = a_{l_k} - \delta_ka_{l_k},x_l = 0} \right]. 
\end{align*}
Let $x_l=0$, then \begin{align}\label{thm parabolic vanish 4}
\left.\text{Interp}_{A^+}(x,t)\right|_{x_l = 0} &= \sum_{\delta_1,\cdots,\delta_a} \left( \prod_{k = 1}^a \frac{x_{l_k} - \delta_k a_{l_k} }{a_{l_k}} \right)(-1)^{\delta_1 + \cdots + \delta_a} \left.u\right|_{x_{l_k} = a_{l_k} - \delta_ka_{l_k},x_l =0}= \left.\text{Interp}_A(x,t)\right|_{x_l=0 }.  
\end{align}
By (\ref{thm parabolic vanish 1}-\ref{thm parabolic vanish 4}), we have $\left.v_1\right|_{x_l = 0} = 0$. Similarly, we can get $\left.v_1\right|_{x_l= a_l} = 0$.

Finally, since $u$ and the interpolation operators $\operatorname{Interp}_A$ preserve periodicity in directions
$l\in S_3$, their linear combination $v_1$ satisfies the periodic boundary conditions \eqref{boundary3}.

2. Then we verify the conditions that $v_2$ meets. For the inhomogeneous term and initial value, we have \begin{align*}
\frac{\partial v_2}{\partial t} &= \mathcal{L}v_1 + \hat{f} + \left(\sum_{\emptyset \neq B \subseteq S_2} (-1)^{|B|}\text{Interp}_B'(x,t)\right)_t\\
&= \mathcal{L}v_2 - \mathcal{L}\left(\sum_{\emptyset \neq B \subseteq S_2} (-1)^{|B|}\text{Interp}_B'(x,t)\right) + \hat{f} + \left(\sum_{\emptyset \neq B \subseteq S_2} (-1)^{|B|}\text{Interp}_B'(x,t)\right)_t\\
&=\mathcal{L}v_2+\tilde{f}(x,t),\\
v_2(x,0)&=v_1(x,0)+\sum_{\emptyset \neq B \subseteq S_2} (-1)^{|B|}\text{Interp}_B'(x,0)=\tilde{u}_0(x).
\end{align*}

For $l$ in $S_1$, since $v_1$ vanishes on each Dirichlet face, hence all tangential derivatives of $v_1$ vanish on that face, i.e.
$$\left.v_1\right|_{x_l = 0} = \left.v_1\right|_{x_l = a_l}=0\ \Rightarrow \ \left.\dfrac{\partial^{b} v_1}{\partial x_{l_1}\cdots\partial x_{l_b}}\right|_{x_l = 0 \text{ or } a_l} = 0\ \text{for } l_1,\cdots,l_b \in S_2.$$So $\left.\text{Interp}_B'(x,t)\right|_{x_l = 0 \text{ or } a_l} = 0$. Then we get $\left.v_2\right|_{x_l = 0} = \left.v_2\right|_{x_l = a_l}=0$ for $l$ in $S_1$.

For $l$ in $S_2$, we use an argument analogous to the Dirichlet situation (now applied to $\partial_{x_l}v_2$ and the family of index sets $B\subseteq S_2$). More explicitly, similarly we have
\begin{align}\label{thm parabolic vanish 5}
\frac{\partial v_2}{\partial x_l} = \frac{\partial v_1}{\partial x_l} - \frac{\partial}{\partial x_l}\text{Interp}'_{\{l\}}(x,t) + \sum_{\substack{l \notin B \subseteq S_2}} (-1)^{|B|}\left(\frac{\partial}{\partial x_l}\text{Interp}_B'(x,t) - \frac{\partial}{\partial x_l}\text{Interp}_{B^+}' (x,t)\right).
\end{align}
Then we get \begin{align}\label{thm parabolic vanish 6}
\text{Interp}'_{\{l\}}(x,t)=\dfrac{\frac{1}{2} x_l^2}{a_l}\left.\dfrac{\partial v_1}{\partial x_l}\right|_{x_l=a_l}-\dfrac{\frac{1}{2} x_l^2-a_lx_l}{a_l}\left.\dfrac{\partial v_1}{\partial x_l}\right|_{x_l=0}\Rightarrow\ 
\left.\frac{\partial}{\partial x_l}\text{Interp}'_{\{l\}}(x,t)\right|_{x_l = 0} = \left.\frac{\partial v_1}{\partial x_l}\right|_{x_l = 0}.    
\end{align}
For $B^+ = \{l_1,\cdots,l_b, l\}$, let $l_{b + 1} = l$ and $\delta_{b + 1} = 0,1$ to compute:
\begin{align*}
&\text{Interp}'_{B^+}(x,t) = \sum_{\delta_1,\cdots,\delta_b} \left( \prod_{k = 1}^b \frac{\frac{1}{2} x_{l_k}^2 - \delta_k a_{l_k}x_{l_k}}{ a_{l_k}} \right)(-1)^{\delta_1 + \cdots + \delta_b} \times \\
&\left( \frac{\frac{1}{2}x_l^2}{a_l} \left.\frac{\partial^{b + 1} v_1}{\partial x_{l_1}\cdots\partial x_{l_b}\partial x_l}\right|_{x_{l_k} = a_{l_k} - \delta_ka_{l_k},x_l = a_l} - \frac{\frac{1}{2}x_l^2 - a_lx_l}{ a_l} \left.\frac{\partial^{b + 1} v_1}{\partial x_{l_1}\cdots\partial x_{l_b}\partial x_l}\right|_{x_{l_k} = a_{l_k} - \delta_ka_{l_k},x_l = 0} \right).    
\end{align*}
Differentiate about $x_l$ and let $x_l=0$, then \begin{align}\label{thm parabolic vanish 7}
\left.\frac{\partial}{\partial x_l}\text{Interp}'_{B^+} (x,t)\right|_{x_l = 0} &= \sum_{\delta_1,\cdots,\delta_b} \left( \prod_{k = 1}^b \frac{\frac{1}{2} x_{l_k}^2 - \delta_k a_{l_k}x_{l_k}}{ a_{l_k}} \right) (-1)^{\delta_1 + \cdots + \delta_b} \left.\frac{\partial^{b + 1} v_1}{\partial x_{l_1}\cdots\partial x_{l_b}\partial x_l}\right|_{x_{l_k} = a_{l_k} - \delta_ka_{l_k},x_l = 0}\notag \\
&= \left.\frac{\partial}{\partial x_l}\text{Interp}'_B (x,t)\right|_{x_l = 0}.  
\end{align}
By (\ref{thm parabolic vanish 5}-\ref{thm parabolic vanish 7}), we have $\left.\dfrac{\partial v_2}{\partial x_l}\right|_{x_l = 0} = 0$. Similarly, we can get $\left.\dfrac{\partial v_2}{\partial x_l}\right|_{x_l = a_l} = 0$.

Finally, periodicity in directions $S_3$ is preserved because $v_1$ is periodic and the coefficients in
$\operatorname{Interp}'_B$ involve only polynomials in Neumann variables multiplied by traces of periodic functions. Thus $v_2$ satisfies \eqref{boundary3}. This completes the proof.

\paragraph{Proof of Theorem \ref{thm hy unique}.}
Let $u_1,u_2$ be two solutions and set $v=u_1-u_2$. Then $v$ satisfies the homogeneous equation
\[
v_{tt}=\Delta v+\sum_{l=1}^d c_l v_{x_l}-c_0^2 v,\quad v(x,0)=0,\quad v_t(x,0)=0,
\]
together with the corresponding homogeneous boundary conditions in each direction.
Define the weighted energy
\[
E(t):=\frac12\int_\Omega\Big(v_t^2+\|\nabla v\|^2+c_0^2 v^2\Big)\exp\left( \sum_{l'=1}^d c_{l'} x_{l'} \right)\,dx.
\]
Differentiating and using the PDE for $v$ gives, \begin{align*}
\frac{\mathrm{d}E}{\mathrm{d}t} &= \int_{\Omega} \left( v_t v_{tt} + \sum_{l=1}^d v_{x_l} v_{x_l t} + c_0^2 v v_t \right) \exp\left( \sum_{l'=1}^d c_{l'} x_{l'} \right) dx \\
&= \int_{\Omega} \left( v_t \left( \Delta v + \sum_{l=1}^d c_l v_{x_l} - c_0^2 v \right) + \sum_{l=1}^d v_{x_l} v_{x_l t} + c_0^2 v v_t \right) \exp\left( \sum_{l'=1}^d c_{l'} x_{l'} \right) dx\\
&= \sum_{l=1}^d \int_{\Omega} \left( v_t v_{x_l x_l} + c_l v_t v_{x_l} + v_{x_l} v_{x_l t} \right) \exp\left( \sum_{l'=1}^d c_{l'} x_{l'} \right) dx.
\end{align*}
Note that \[
\int_0^{a_l} \left( v_t v_{x_l x_l} + c_l v_t v_{x_l} + v_{x_l} v_{x_l t} \right) e^{c_l x_l} dx_l = \left[ e^{c_l x_l} v_t v_{x_l} \right]_0^{a_l}.\]
For Dirichlet boundary, Neumann boundary and periodic boundary together with the condition \eqref{hy stable condition}, we can all get \(\left[ e^{c_l x_l} v_t v_{x_l} \right]_0^{a_l}=0\).
So $\dfrac{\mathrm{d}E}{\mathrm{d}t} =0$. Since $E(0) = 0$, we conclude that $E(t) \equiv 0$ for all $t \geq 0$.
Therefore, $v(t) \equiv 0$, which implies $u_1 = u_2$.

\section{Proof of Lemma \ref{lemma para scaling of c} and Lemma \ref{thm para len-gauss scale}}\label{appendix proof of gauss-len}
\paragraph{Proof of Lemma \ref{lemma para scaling of c}.}
Since in Gauss quadrature, the weights $w_{q_t}>0$, we get $c_{q_t,m_t}\geq 0$.
    By the composite Gauss--Legendre quadrature, for any function $f(s)$ in $[0,T]$,\[
    \int_0^T f(s)ds\approx \sum_{m_t=0}^{T/h_t-1} \sum_{q_t=0}^{Q_t-1} c_{q_t,m_t}\frac{1}{\|b(s_{q_t,m_t})\|_2}f(s_{q_t,m_t}).\]Let $f(s)=\|b(s)\|_2$ and we have $\int_0^T \|b(s)\|_2 ds\approx\sum_{m_t=0}^{T/h_t-1} \sum_{q_t=0}^{Q_t-1} c_{q_t,m_t},$ where the error is the error of the composite Gauss–Legendre quadrature formula. So $\sum_{j_t=0}^{M_t-1}c_{j_t}=\mathcal{O}(\int_0^T \|b(s)\|_2 ds).$

\paragraph{Proof of Lemma \ref{thm para len-gauss scale}.}
Define
\(
h(s):= P^{-1} e^{-(T-s)\widetilde{A}} P\, b(s).
\)
By the standard error formula for the composite Gauss--Legendre rule (with $Q_t$ nodes per subinterval), we have \begin{align*}
    \|v(T)-v_Q(T)\|_2
=&\left\| \int_0^T P^{-1}e^{-(T - s)\widetilde{A}} P b(s) \, ds -\sum_{m_t=0}^{T/h_t-1} \sum_{q_t=0}^{Q_t-1} c_{q_t,m_t} P^{-1}e^{-(T - s_{q_t,m_t})\widetilde{A}} P |b(s_{q_t,m_t})\rangle \right\|_2\\
\leq & \frac{T (Q_t!)^4 h_t^{2Q_t}}{(2Q_t + 1) ((2Q_t)!)^3} \max_{s\in [0,T]} \| h^{(2Q_t)}(s) \|_2.
\end{align*}
 By Stirling's inequality $\sqrt{2\pi n}\left(\frac{n}{e}\right)^n < n! < \sqrt{2\pi n}\left(\frac{n}{e}\right)^n e^{\frac{1}{12n}}$, \begin{align*}
    \frac{T (Q_t!)^4 h_t^{2Q_t}}{(2Q_t + 1) ((2Q_t)!)^3}
    \leq \frac{T(2\pi Q_t)^2(Q_t/e)^{4Q_t}e^{\frac{1}{3Q_t}}h_t^{2Q_t}}{(2Q_t+1)(4\pi Q_t)^{3/2}(2Q_t/e)^{6Q_t}}\leq \frac{eT\sqrt{\pi Q_t}}{4Q_t+2}\cdot \left(\frac{eh_t}{8Q_t}\right)^{2Q_t}.
\end{align*}
So
\begin{align}\label{thm len-gauss 1}
    \|v(T)-v_Q(T)\|_{\star}\leq \frac{eT\sqrt{\pi Q_t}}{4Q_t+2}\cdot \left(\frac{eh_t}{8Q_t}\right)^{2Q_t} \max_s \| h^{(2Q_t)}(s) \|_{\star}.
\end{align}

Next, we have  
\[
\frac{\partial^p }{\partial s^p} e^{-(T-s)\widetilde{A}}= \widetilde{A}^p e^{-(T-s)\widetilde{A}}.
\] And by $L=\frac{1}{2}(\widetilde{A}+\widetilde{A}^\dag)\succeq 0$, $\|e^{-(T-s)\widetilde{A}}\|_2\leq 1$.
So 
\[
\left\| \frac{\partial^p }{\partial s^p} e^{-(T-s)\widetilde{A}} \right\|_2 \leq \|\widetilde{A}\|_2^p.
\]
By the product rule, \begin{align*}
\| h^{(2Q_t)}(s) \|_{\star} = \left\| \sum_{p=0}^{2Q_t} \binom{2Q_t}{p} P^{-1} \frac{\partial^p }{\partial s^p} e^{-(T-s)\widetilde{A}} P b^{(2Q_t - p)}(s) \right\|_{\star}  \leq \text{cond}(P)\sum_{p=0}^{2Q_t} \binom{2Q_t}{p} \|\widetilde{A}\|_2^p \|b^{(2Q_t - p)}(s)\|_{\star}.
\end{align*}

If 
$\Xi := \sup\left\{ (\| b^{(p)} \|_{\star})^{\frac{1}{p+1}} \,\big|\, p \geq 0,\, t \in [0,T] \right\} < \infty$
then $\|b^{(2Q_t - p)}(s)\|_{\star}\leq \Xi^{2Q_t - p + 1}.$ So
\begin{align*}
\| h^{(2Q_t)}(s) \|_{\star}
\leq \text{cond}(P)\sum_{p=0}^{2Q_t} \binom{2Q_t}{p} \|\widetilde{A}\|_2^p \Xi^{2Q_t - p + 1} 
= \text{cond}(P)\Xi \left( \|\widetilde{A}\|_2 + \Xi \right)^{2Q_t}.
\end{align*}
By \eqref{thm len-gauss 1}, 
\[
\|v(T)-v_Q(T)\|_{\star}\leq \frac{eT\sqrt{\pi Q_t}}{4Q_t+2}\cdot \text{cond}(P)\Xi \cdot \left(\frac{eh_t(\|\widetilde{A}\|_2 + \Xi)}{8Q_t}\right)^{2Q_t}\leq T\cdot \text{cond}(P)\Xi \cdot \left(\frac{eh_t(\|\widetilde{A}\|_2 + \Xi)}{8Q_t}\right)^{2Q_t}.
\]
Choose 
\begin{align}\label{thm len-gauss 2}
\frac{e h_t (\|\widetilde{A}\|_2 + \Xi)}{8 Q_t} \leq \frac{1}{2}, \quad  T\cdot \text{cond}(P)\Xi \left( \frac{1}{2} \right)^{2Q_t} \leq \varepsilon_{Q}.
\end{align}
Then we can get $\|v(T)-v_Q(T)\|_{\star} \leq \varepsilon_{Q}$. 
Solving \eqref{thm len-gauss 2}, we have
\[
Q_t \geq \frac{1}{\ln 4} \ln \frac{ T\,\mathrm{cond}(P) \Xi}{\varepsilon_{Q}} , \quad h_t \leq \frac{4 Q_t}{e (\| \widetilde{A} \|_2 + \Xi)}.
\]

\section{Lemmas about the coefficient matrices}

\begin{lemma}\label{lem dirichlet matrix}
For the matrix $(-1<\alpha<1)$
\[
A = 
\begin{pmatrix}
 2 & -1-\alpha &        &        \\
      -1+\alpha & 2 & \ddots &        \\
           & \ddots & \ddots & -1-\alpha \\
           &        & -1+\alpha & 2
\end{pmatrix}_{N},\]
let 
\( P = \text{diag}(1, \theta, \dots, \theta^{N-1}) \) and \( \theta = \sqrt{\dfrac{1+\alpha}{1-\alpha}} \)
. Then

(1) The similar matrix of 
\( A \)
\[
\widetilde{A} = PAP^{-1} = \begin{pmatrix}

 2 & -\sqrt{1-\alpha^2} &        &        \\
      -\sqrt{1-\alpha^2} & 2 & \ddots &        \\
           & \ddots & \ddots & -\sqrt{1-\alpha^2} \\
           &        & -\sqrt{1-\alpha^2} & 2\end{pmatrix}_{N}\]
is symmetric positive definite. Moreover, the minimum eigenvalue of 
\( \widetilde{A} \)
 satisfies
\[
\lambda_{\min}(\widetilde{A}) = 2 - 2\sqrt{1-\alpha^2}\cos\frac{\pi}{N+1} \geq \frac{4\sqrt
{1-\alpha^2}}{(N+1)^2};
\]the maximum eigenvalue of 
\( \widetilde{A} \)
 satisfies
\[
\lambda_{\max}(\widetilde{A}) = 2 + 2\sqrt{1-\alpha^2}\cos\frac{\pi}{N+1} \leq 4.
\]

(2) 
\( \widetilde{A} \) has a decomposition \( \widetilde{A} = DD^T \)
, where
\[D = 
\begin{pmatrix}
\sqrt{1+\alpha} & -\sqrt{1-\alpha} &        &        &        \\
     & \sqrt{1+\alpha} & -\sqrt{1-\alpha} &        &        \\
     &      & \ddots & \ddots &        \\
     &      &        & \sqrt{1+\alpha} & -\sqrt{1-\alpha}
     \end{pmatrix}_{N\times (N+1)}.\]
\end{lemma}

\begin{lemma}\label{lem neumann matrix}
For the matrix $(-1<\alpha<1)$
\[
A = 
\begin{pmatrix}
1+\alpha & -1-\alpha &        &        &        \\
-1+\alpha & 2 & -1-\alpha &        &        \\
     & -1+\alpha & \ddots & \ddots &        \\
     &      & \ddots & 2 & -1-\alpha \\
     &      &        & -1+\alpha & 1-\alpha\end{pmatrix}_{N},\]
let 
\( P = \text{diag}(1, \theta, \dots, \theta^{N-1}) \) and \( \theta = \sqrt{\dfrac{1+\alpha}{1-\alpha}} \). Then

(1) The similar matrix of 
\( A \)
\[
\widetilde{A} = PAP^{-1} = \begin{pmatrix}
1+\alpha & -\sqrt{1-\alpha^2} &        &        &        \\
-\sqrt{1-\alpha^2} & 2 & -\sqrt{1-\alpha^2} &        &        \\
     & -\sqrt{1-\alpha^2} & \ddots & \ddots &        \\
     &      & \ddots & 2 & -\sqrt{1-\alpha^2} \\
     &      &        & -\sqrt{1-\alpha^2} & 1-\alpha
\end{pmatrix}_{N}\]
is symmetric positive semi-definite and $\|\widetilde{A}\|\leq 4$.

(2) 
\( \widetilde{A} \) has a decomposition \( \widetilde{A} = DD^T \), where
\[
D = 
\begin{pmatrix}
\sqrt{1+\alpha} &        &        &        &        \\
-\sqrt{1-\alpha} & \sqrt{1+\alpha} &        &        &        \\
     & -\sqrt{1-\alpha} & \ddots &        &        \\
     &      & \ddots & \sqrt{1+\alpha} &        \\
     &      &        & -\sqrt{1-\alpha} & 0
\end{pmatrix}_{N}. 
\]
\end{lemma}

\begin{lemma}\label{lem periodic matrix}
    The matrix $(-1<\alpha<1)$
\[\widetilde{A} = 
\begin{pmatrix}
2 & -1-\alpha &        &        &  -1+\alpha      \\
-1+\alpha & 2 & -1-\alpha &        &        \\
     & -1+\alpha & \ddots & \ddots &        \\
     &      & \ddots & 2 & -1-\alpha \\
  -1-\alpha   &      &        & -1+\alpha & 2
     \end{pmatrix}_{N},\]
      has $\|\widetilde{A}\|\leq 4$. Moreover, $L=\frac{1}{2}(\widetilde{A}+\widetilde{A}^\dag)$ is symmetric positive semi-definite and $\|L\|_2\leq 4.$
\end{lemma}

\section{Lemmas about block-encoding}\label{appendix blockencoding}

\begin{lemma}[{\cite[Lemma 53]{qsvt}}]\label{lemma be times}
Suppose \( A,B \in \mathbb{C}^{N \times N} \), \( U_A \in (\alpha_1, m_1, \varepsilon_1)BE (A) \) and \( U_B \in (\alpha_2, m_2, \varepsilon_2)BE (B) \).
Then we can construct a \( (\alpha_1 \alpha_2, m_1 + m_2, \alpha_2\varepsilon_1+\alpha_1\varepsilon_2+\varepsilon_1\varepsilon_2)\)-block-encoding of \(AB \), using an oracle $U_A$ and an oracle $U_B$. The circuit is in Figure \ref{fig times}.
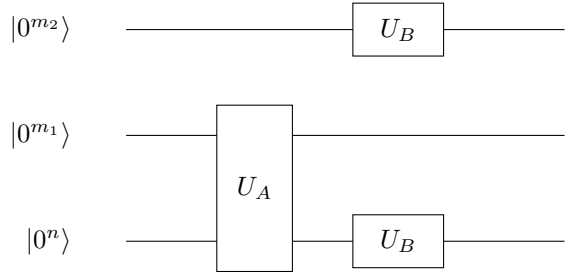
\begin{figure}[htbp]
\centering
\begin{tikzpicture}[
    x=1cm,y=1cm,
    gate/.style={draw, minimum width=1.2cm, minimum height=0.7cm, inner sep=0pt}]
\node[left] at (1.8, 0.0) {$\ket{0^{m_2}}$};
\node[left] at (1.8,-1.4) {$\ket{0^{m_1}}$};
\node[left] at (1.8,-2.8) {$\ket{0^{n}}$};
\node[gate] (UBtop) at (6.0,0.0) {$U_B$};
\node[gate] (UBbot) at (6.0,-2.8) {$U_B$};
\node[
    draw,
    minimum width=1.0cm,
    minimum height=2.2cm,
    inner sep=0pt] (UAbig) at (4.1,-2.1) {$U_A$};
\draw (2.4,0.0) -- (5.4,0.0);
\draw (6.6,0.0) -- (8.2,0.0);
\draw (2.4,-1.4) -- (3.6,-1.4);
\draw (4.6,-1.4) -- (8.2,-1.4);
\draw (2.4,-2.8) -- (3.6,-2.8);
\draw (4.6,-2.8) -- (5.4,-2.8);
\draw (6.6,-2.8) -- (8.2,-2.8);
\end{tikzpicture}
\caption{Circuit for Lemma \ref{lemma be times}.}
\label{fig times}
\end{figure}
\end{lemma}

\begin{lemma}\label{lemma be tensor}
Suppose \( A \in \mathbb{C}^{N_1 \times N_1} \), \( B \in \mathbb{C}^{N_2 \times N_2} \), \( U_A \in (\alpha_1, m_1, \varepsilon_1)BE (A) \) and \( U_B \in (\alpha_2, m_2, \varepsilon_2)BE (B) \). Then we can construct a \( (\alpha_1 \alpha_2, m_1 + m_2, \alpha_2\varepsilon_1+\alpha_1\varepsilon_2+\varepsilon_1\varepsilon_2)\)-block-encoding of \( A \otimes B \), using an oracle $U_A$ and an oracle $U_B$. The circuit is in Figure \ref{fig tensor}.
\begin{figure}[htbp]
\centering
\begin{tikzpicture}[
    x=1cm,y=1cm,
    gate/.style={draw, minimum width=1.2cm, minimum height=0.7cm, inner sep=0pt}]
\node[left] at (1.8, 0.0) {$\ket{0^{m_2}}$};
\node[left] at (1.8,-1.2) {$\ket{0^{m_1}}$};
\node[left] at (1.8,-2.4) {$\ket{0^{n_2}}$};
\node[left] at (1.8,-3.6) {$\ket{0^{n_1}}$};
\coordinate (UB1) at (6.0, 0.0);
\coordinate (UA1) at (4.0,-1.2);
\coordinate (UB2) at (6.0,-2.4);
\coordinate (UA2) at (4.0,-3.6);
\draw (2.4, 0.0) -- (5.4, 0.0);
\draw (6.6, 0.0) -- (8.4, 0.0);
\draw (2.4,-1.2) -- (3.4,-1.2);
\draw (4.6,-1.2) -- (8.4,-1.2);
\draw (2.4,-2.4) -- (5.4,-2.4);
\draw (6.6,-2.4) -- (8.4,-2.4);
\draw (2.4,-3.6) -- (3.4,-3.6);
\draw (4.6,-3.6) -- (8.4,-3.6);
\node[gate] at (UB1) {$U_B$};
\node[gate] at (UA1) {$U_A$};
\node[gate] at (UB2) {$U_B$};
\node[gate] at (UA2) {$U_A$};
\end{tikzpicture}
\caption{Circuit for Lemma \ref{lemma be tensor}.}
\label{fig tensor}
\end{figure}
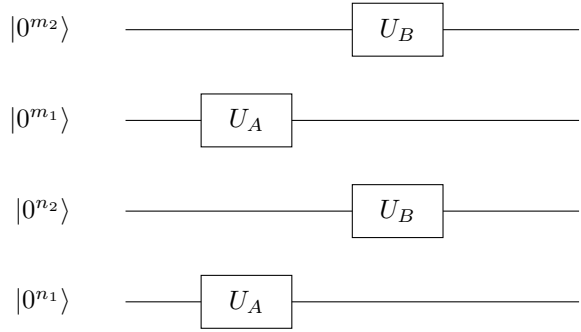
\end{lemma}
\begin{proof}
    Since \( U_A\otimes I \in (\alpha_1, m_1, \varepsilon_1)BE (A\otimes I) ,\ I \otimes U_B \in (\alpha_2, m_2, \varepsilon_2)BE (I \otimes B)\), this Lemma follows from Lemma \ref{lemma be times}.
\end{proof}

\begin{lemma}[Linear Combination of Unitaries {\cite[Lemma 52]{qsvt}}]\label{lemma be lcu}
Let \( A = \sum_{i=1}^s c_i A_i \) with \( c_i > 0 \) for all \( i \) and \(U_i\in (\alpha_i,m,\varepsilon_i)BE(A_i) \). Define \( V \) as an operator satisfying \( V|0 \rangle := \frac{1}{\sqrt{\alpha}} \sum_{i=1}^s \sqrt{c_i\alpha_i} |i \rangle \) where \( \alpha := \sum_{i=1}^s c_i\alpha_i \) and \( U := \sum_i |i \rangle \langle i | \otimes U_i \). Then \[ V^\dagger U V\in (\alpha,m+\lceil\log s\rceil,\sum_i c_i \varepsilon_i)BE(A). \]
\end{lemma}


\begin{lemma} \label{lemma unit error}
Let $\psi$ and $\varphi$ be two non-zero vectors. Then
\[
\left\| \frac{\psi}{\|\psi\|_2} - \frac{\varphi}{\|\varphi\|_2} \right\|_2 \leq \frac{2\|\psi-\varphi\|_2}{\|\psi\|_2}.
\]
\end{lemma}
\begin{proof}
By the triangle inequality,
\[\left\| \frac{\psi}{\|\psi\|_2} - \frac{\varphi}{\|\varphi\|_2} \right\|_2
=
\left\| \frac{\psi}{\|\psi\|_2}-\frac{\varphi}{\|\psi\|_2}+\frac{\varphi}{\|\psi\|_2}-\frac{\varphi}{\|\varphi\|_2}\right\|_2
\leq
\frac{\|\psi-\varphi\|_2}{\|\psi\|_2}
+
\|\varphi\|_2\left|\frac{1}{\|\psi\|_2}-\frac{1}{\|\varphi\|_2}\right|.
\]
Moreover,
\[
\|\varphi\|_2\left|\frac{1}{\|\psi\|_2}-\frac{1}{\|\varphi\|_2}\right|
=
\frac{\bigl|\|\varphi\|_2-\|\psi\|_2\bigr|}{\|\psi\|_2}
\leq
\frac{\|\psi-\varphi\|_2}{\|\psi\|_2}.
\]
Then we get the Lemma.
\end{proof}

\end{appendices}

\end{document}